\documentclass[12pt, reqno, fleqn]{article}
\usepackage[margin=1in]{geometry}
\usepackage{lmodern}
\usepackage[T1]{fontenc} 
\usepackage{natbib}
\usepackage[figuresright]{rotating}
\usepackage{float}
\usepackage{graphicx}
\usepackage{epstopdf}
\usepackage{url}
\usepackage[colorlinks,citecolor=blue,linkcolor=blue,urlcolor=red,pagebackref]{hyperref}
\usepackage{amsmath,amsthm}
\usepackage{longtable}
\usepackage[dvipsnames]{xcolor}
\usepackage[classfont = sanserif, langfont = roman, funcfont = italic]{complexity}
\usepackage{amssymb}
\usepackage[yyyymmdd,hhmmss]{datetime}
\usepackage{algpseudocode}
\usepackage{algorithm}
\usepackage{orcidlink}

\newtheorem{theorem}{Theorem}
\newtheorem*{theorem1}{Theorem 1}
\newtheorem*{theorem2}{Theorem 2}
\newtheorem*{theorem3}{Theorem 3}
\newtheorem*{theorem4}{Theorem 4}
\newtheorem{lemma}{Lemma}
\newtheorem{claim}{Claim}

\newtheorem{corollary}{Corollary}

\theoremstyle{definition}

\theoremstyle{remark}
\newtheorem{definition}{Definition}
\newtheorem*{definition1}{Definition 1}

\providecommand{\keywords}[1]
{
  \small	
  \textbf{Keywords:} #1
}

\newclass{\SHARPP}{\#P}
\newclass{\PPOLY}{P/Poly}
\newclass{\rETH}{RETH}
\newclass{\ETH}{ETH}
\newclass{\SETH}{SETH}
\newclass{\rSETH}{RSETH}

\newlang{\OV}{OV}
\newlang{\HAMPATH}{HAMPATH}
\newlang{\HAMCYCLE}{HAMCYCLE}
\newlang{\CLIQUE}{CLIQUE}
\newlang{\MULT}{MULT}
\newlang{\DLP}{DLP}
\newlang{\dHAMCYCLE}{dHAMCYCLE}
\newlang{\COL}{COLORING}
\newlang{\HALF}{HALFCLIQUE}

\newfunc{\HCY}{HCY}
\newfunc{\HCL}{HCL}
\newfunc{\PER}{PER}
\newfunc{\GPR}{GPR}
\newfunc{\MLP}{MLP}
\newfunc{\Aut}{Aut}

\begin{document}

\title {Hardness Amplification via Group Theory}                      

\author{Tejas Nareddy\orcidlink{0009-0007-7032-6654}\footnote{Department of Computer Science and Information Systems, Birla Institute of Technology and Science, Pilani, Pilani-333031, Rajasthan, I\textsc{ndia}. Email: \texttt{f20211462@pilani.bits-pilani.ac.in}.} \and 
Abhishek Mishra\orcidlink{0000-0002-2205-0514}\footnote{Department of Computer Science and Information Systems, Birla Institute of Technology and Science, Pilani, Pilani-333031, Rajasthan, I\textsc{ndia}. Email: \texttt{abhishek.mishra@pilani.bits-pilani.ac.in}.}}

\maketitle
\thispagestyle{empty}

\begin{abstract}

We employ elementary techniques from group theory to show that, in many cases, counting problems on graphs are almost as hard to solve in a small number of instances as they are in all instances. Specifically, we show the following results.

\begin{enumerate}

\item \cite{Boix2019} showed in FOCS 2019 that given an algorithm $A$ computing the number of $k$-cliques modulo $2$ that is allowed to be wrong on at most a $\delta = O \left( 1 / ( \log k )^{k \choose 2} \right)$-fraction of $n$-vertex simple undirected graphs in time $T_A(n)$, we have a randomized algorithm that, in $O \left( n^2 + T_A(nk) \right)$-time, computes the number of $k$-cliques modulo $2$ on any $n$-vertex graph with high probability. \cite{Goldreich2020} improved the error tolerance to a fraction of $ \delta = 2^{-k^2}$, making $2^{O \left( k^2 \right)}$-queries to the average-case solver in $O \left( n^2 \right)$-time. Both works ask if any improvement in the error tolerance is possible. In particular, \cite{Goldreich2020} asks if, for every constant $\delta < 1 / 2$, there is an $\tilde{O} \left( n^2 \right)$-time randomized reduction from computing the number of $k$-cliques modulo $2$ with a success probability of greater than $2 / 3$ to computing the number of $k$-cliques modulo $2$ with an error probability of at most $\delta$.

In this work, we show that for almost all choices of the $\delta 2^{n \choose 2}$ corrupt answers within the average-case solver, we have a reduction taking $\tilde{O} \left( n^2 \right)$-time and tolerating an error probability of $\delta$ in the average-case solver for any constant $\delta < 1 / 2$. By ``almost all'', we mean that if we choose, with equal probability, any subset $S \subset \{0,1\}^{n  \choose 2}$ with $|S| = \delta2^{n  \choose 2}$, then with a probability of $1-2^{-\Omega \left( n^2 \right)}$, we can use an average-case solver corrupt on $S$ to obtain a probabilistic algorithm.

\item Inspired by the work of \cite{Goldreich2018} in FOCS 2018 to take the weighted versions of the graph counting problems, we prove that if the $\textit{Randomized Exponential Time Hypothesis} (\rETH$) is true, then for a prime $p = \Theta \left( 2^n \right)$, the problem of counting the number of unique Hamiltonian cycles modulo $p$ on $n$-vertex directed multigraphs and the problem of counting the number of unique half-cliques modulo $p$ on $n$-vertex undirected multigraphs, both require exponential time to compute correctly on even a $1 / 2^{n/\log n}$-fraction of instances. Meanwhile, simply printing $0$ on all inputs is correct on at least a $\Omega \left( 1 / 2^n \right)$-fraction of instances.

\end{enumerate}

\end{abstract}


\newpage

\pagenumbering{roman}
\setcounter{page}{1}

\maketitle

\tableofcontents

\newpage

\pagenumbering{arabic}
\setcounter{page}{1}

\section{Introduction}
\label{section:1}

Average-case complexity or typical-case complexity is the area of complexity theory concerning the difficulty of computational problems on not just worst-case inputs but on the majority of inputs \citep{Ben1992, Bogdanov2021, Levin1986}. The theory of worst-case hardness, specifically in the context of proving conditional or unconditional lower bounds for a function $f$, attempts to prove that for every fast algorithm $A$, there is an input $x$ such that the output of the algorithm $A$ on the input $x$ differs from $f(x)$. However, for applications such as cryptography \citep{Impagliazzo1005}, it is not just sufficient that $f$ is hard on some input. It should have a good probability of being hard on an easily samplable distribution of inputs.

Possibly the most famous and instructive example of this is the \textit{Discrete Logarithm Problem} ($\DLP$), which has been proved by \cite{Blum1982} to be intractable for polynomial-time algorithms to compute correctly on even a $1 / h(n)$-fraction of inputs for any polynomial $h$, for all sufficiently large $n$, if it is intractable for polynomial-time randomized algorithms in the worst-case. Given a generator $g$ of a field $\mathbb{Z}_p$ of prime size and an input $y \in \mathbb{Z}_p$, the $\DLP$ asks to find $x \in \mathbb{Z}_{p}$ such that $g^x \equiv y \pmod p$. Suppose that there is an algorithm $A$ computing the $\DLP$ correctly on a $1/h(n)$-fraction of inputs; then, for any $y \in \mathbb{Z}_p$, we can attempt to compute $x$ as follows: Pick a random $x^\prime \in \mathbb{Z}_p$ and ask $A$ to compute $x^{\prime \prime}$ such that $g^{x^{\prime \prime}} \equiv yg^{x^\prime} \pmod p$; we verify if this is true, and if so, we return $x = x^{\prime \prime} - x^\prime$, else we repeat. We will likely obtain the correct answer in $h(n)$ repetitions with a probability of at least $1 - 1 / e$, giving us a polynomial-time randomized algorithm for the problem. \cite{Blum1982} use this to construct pseudorandom generators \citep{Vadhan2012}, algorithms that generate random-looking bits. Algorithms that generate pseudorandomness have many applications: For the quick deterministic simulation of randomized algorithms, for generating random-looking strings for secure cryptography, for zero-knowledge proofs \citep{Goldreich1991, 10.5555/1202577}, and many others.

This hardness result for the $\DLP$ is a special case of rare-case hardness, a term coined by \cite{Goldreich2018}, which refers to computational problems where algorithms with a specific time complexity cannot be correct on even an $o(1)$-fraction of inputs. \cite{Cai1999} proved more such hardness results for the permanent, building on a long line of work, showing intractability results for computing the permanent \citep{Valiant1979, Gemmell1992, Feige1996}.

Most average-case hardness and rare-case hardness results are shown, similar to the case of the $\DLP$, by proving that the tractability of some function $f$ on a small fraction of inputs implies the tractability over all inputs of some function $h$ that is conjectured to be intractable. Most existing techniques use error correction over polynomials to achieve such hardness amplifications \citep{Ball2017, Lund1992, Gemmell1992}. Recently, \citet{Asadi2022} introduced tools from additive combinatorics to prove rare-case hardness results for matrix multiplication and streaming algorithms, revealing new avenues for complexity-theoretic research. A more recent application of additive combinatorics shows that in quantum computing, all linear problems have worst-case to average-case reductions \citep{Asadi2024}.

In this paper, we intend to show that group theory is a powerful tool for
achieving hardness amplification for graph problems. We emphasize that we are far from the first to apply group theory to graphs in the context of theoretical computer science \citep{Babai2006, Luks1982}. The breakthrough quasipolynomial-time algorithm of \cite{Babai2016} for the graph isomorphism problem is a tour de force in the application of group theory to graph theoretic computational problems. Our thesis is that group theory is also a powerful tool in the theory of average-case and rare-case complexity for graph problems.

\subsection{Counting $k$-Cliques Modulo $2$}
\label{section:1.1}

One area that has gained much attention in the past few decades is the paradigm of ``hardness within $\P$'' \citep{Williams2015}. In particular, for practical reasons, it is not just important to us that a problem is in $\P$, but also that it is ``quickly computable'', in one sense of the phrase. Traditionally, in complexity theory, ``quickly computable'' has been used interchangeably with polynomial-time computable. However, considering the vast data sets of today, with millions of entries, $O \left( n^{15} \right)$-time complexity algorithms are not practical. Many works have gone into showing conditional lower bounds as well as tight algorithms \citep{Abboud2014, Williams2019, Williams2018}.

Another practical application, which arguably motivated many works on fine-grained average-case hardness, starting with that of \cite{Ball2017}, is the idea of a proof of work. A proof of work \citep{Dwork1993} is, informally, a protocol where a prover intends to prove to a verifier that they have expended some amount of computational power. The textbook example for an application is combatting spam emails: Forcing a sender to expend some resource for every email sent makes spamming uneconomical. The idea is to have a prover compute an $f(x)$ for some function $f$, which both the prover and verifier know, and an input $x$ of the verifier's choice. In particular, the verifier needs a distribution of inputs for which $f$ is expected to take some minimum amount of time to compute for any algorithm. Another condition is that the distribution should be easy to sample from $f$ and should not be too hard. We do not want to make sending emails impossible. This is where average-case fine-grained hardness for problems computable in polynomial-time comes into the picture. Many other such works have explored these ideas further \citep{Enric2019, Mina2020, Goldreich2018, Asadi2022}.

We also emphasize that these works are not only important for practical reasons but also give insights into the structure of $\P$ itself.

\subsubsection{Background}
\label{section:1.1.1}

One problem that has become a central figure in this paradigm of ``hardness within $\P$'' is the problem of counting $k$-cliques for a fixed $k$, or $k$ fixed as a parameter. Many works have explored the fine-grained complexity of variants of this problem and many others \citep{Mina2020, Goldreich2018, Goldreich2023}.

One interesting direction is to explore how the difficulty of computing the number of $k$-cliques modulo $2$ correctly on some fraction of graphs (say $0.75$) relates to the complexity of computing this number correctly for all inputs with high probability. This is simultaneously a counting problem, a class of problems for which many average-case hardness results are known, and a decision problem, where finding worst-case to average-case reductions is more complicated. \cite{Boix2019} explore this question for general hypergraphs, and for the case of simple undirected graphs, they show that if there is an algorithm $A$ computing the number of $k$-cliques modulo $2$ correctly on over a $1-\Omega \left( 1 / (\log k)^{k \choose 2} \right)$-fraction of instances, then we can obtain an algorithm that computes the number of $k$-cliques modulo $2$ correctly on all inputs in $O\bigg((\log k)^{k \choose 2}(T_{A}(nk) + (nk)^2)\bigg)$-time, where $T_A(m)$ is the time taken by $A$ on input graphs with $m$ vertices. They do this by reducing the problem of counting the number of $k$-cliques on $n$ vertices graphs to the problem of counting $k$-cliques on $k$-partite graphs where each partition has $n$ vertices. The $k$-partite setting is reduced to the problem of computing a low-degree polynomial, where the average-case hardness is obtained.

\cite{Goldreich2020}\footnote{They call it $t$-cliques, possibly to emphasize that $t$ is a parameter.} improves the error tolerance from
\begin{equation*}
O \left( (\log k)^{-{k \choose 2}} \right) = 2^{-\Omega \left( k^2 \log \log k \right)}
\end{equation*}
to $2^{-k^2}$ and simplifies the reduction. They make $2^{O \left( k^2 \right)}$ queries and use $O \left( n^2 \right)$-time. More specifically, they construct a new polynomial such that one of the evaluations gives us our answer of interest. They use the crucial insight that the sum of a low-degree polynomial (degree less than the number of variables) over all inputs in $\mathbb{Z}_2$ is $0$, and hence summing over all inputs other than the one of interest gives us our answer. They make $2^{k \choose 2}-1$ correlated queries, each of which is uniformly distributed over the set of all simple undirected $n$ vertex graphs. Using the union bound, if the fraction of incorrect instances in the average-case solver is $2^{-k^2}$, the probability of error for the reduction is bounded by $2^{k \choose 2} / 2^{k^2} = o(1)$.

Both \cite{Boix2019} and \cite{Goldreich2020} ask whether there are similar worst-case to average-case reductions tolerating more substantial error. In particular, \cite{Goldreich2020} asks whether there is a randomized reduction taking $\tilde{O}(n^2)$ time from computing the number of $k$-cliques modulo $2$ on any graph with a success probability of larger than $2 / 3$ to computing the number of $k$-cliques modulo $2$ on a $(1 / 2 + \epsilon)$-fraction of instances for any arbitrary constant $\epsilon > 0$.

\subsubsection{Our Results}
\label{section:1.1.2}

First, we define a random experiment $O^{H_n}_{c}$.

\begin{definition}
\label{def:1}

\textit{\textbf{The Random Experiment $O^{H_n}_{c}$.} \\
Given any set $\mathbb{D}$ and a function $H_n: \{\, 0,1 \,\}^{n \choose 2} \to \mathbb{D}$ defined over $n$-vertex simple undirected graphs that is invariant under graph isomorphism, the random experiment $O^{H_n}_{c}$ selects a set $S \subset \{\, 0,1 \,\}^{n \choose 2}$ of size $c 2^{n \choose 2}$ with uniform probability and gives an oracle $O$ that correctly answers queries for computing $H_n$ on the set $S$. The other answers of $O$ can be selected adversarially, randomly, or to minimize the time complexity,  $T_O$, of the fastest deterministic algorithm implementing it.}

\end{definition}

In Section \ref{section:8}, we will prove the following results, crucially relying on these functions or problems being invariant under graph isomorphism. Our results hold regardless of $O$'s answers to $\overline{S}$.

\begin{theorem}
\label{thm:1}

For any $k \in \mathbb{N}$ (not necessarily a constant), given an $\epsilon = \omega \left( n^{3/2} / \sqrt{n!} \right)$, given an oracle $O$ sampled from $O^{H_n}_{1 / 2 + \epsilon}$, where $H_n: \{\, 0,1 \,\}^{n \choose 2} \to \mathbb{D}$ is any function defined over $n$-vertex undirected simple graphs that is invariant under graph isomorphism and can be computed in $O \left( n^{8 + o(1)} / \epsilon^{4 + o(1)} \right)$-time given the number of $k$-cliques in the graph, then with a probability of at least $1 - 2^{-\Omega \left( n^2 \right)}$ over the randomness of $O^{H_n}_{1 / 2 + \epsilon}$, we have an algorithm that, with access to $O$ computes $H_n$ with a high probability in time $O \left( \left( n^{8 + o(1)} / \epsilon^{2 + o(1)} + T_{O} \right) / \epsilon^2 \right)$, where $T_O$ is the time complexity of a hypothetical algorithm simulating the oracle $O$.

\end{theorem}

Informally, this says that for almost all subsets $S$ of $\{\, 0,1 \,\}^{n \choose 2}$ with $|S| = (1 / 2 + \epsilon)2^{n \choose 2}$, an algorithm that computes $H_n$ correctly on the set $S$ is nearly as hard, computationally speaking, as computing $H_n$ correctly on all instances with a randomized algorithm with high probability. Due to work of \cite{Feigenbaum1993}, and \cite{Bogdanov2006}, under the assumption that the $\PH$ does not collapse, this is a near-optimal result for non-adaptive\footnote{This is when the inputs of the queries we make to $O$ do not depend on any of the answers. Another interpretation is that the inputs to query on $O$ must be decided before making any queries to it.} querying when no other assumptions are made of $H_n$. In particular, when applied to the $\NP$-complete problem of deciding whether a simple undirected graph with $n$ vertices has a clique of size $\lfloor n / 2 \rfloor$, $\HALF$, we show a non-adaptive polynomial-time reduction from computing $\HALF$ over any instance to computing $\HALF$ correctly on $S$ for almost all $S \subset \{\, 0,1 \,\}^{n \choose 2}$ with $|S| = (1 / 2 + \epsilon)2^{n \choose 2}$ for any $\epsilon = 1 / \poly(n)$. If this reduction can be extended to show this for all $S$, instead of almost all, $\PH$ would collapse to the third level \citep{Feigenbaum1993, Bogdanov2006}.

We also show the following result, making progress on the open problem of \cite{Goldreich2020}.

\begin{theorem}
\label{thm:2}

Given any constants $k > 2$ and $\epsilon > 0$, with a probability of at least $1 - 2^{-\Omega \left( n^2 \right)}$ over the randomness of sampling $O$ from $O^{H_n}_{1 / 2 + \epsilon}$, where $H_n$ is the function counting the number of $k$-cliques modulo $2$ in an $n$-vertex undirected simple graph, we have an $\tilde{O} \left( n^2 \right)$-time randomized reduction from counting $k$-cliques modulo $2$ on all instances to counting $k$-cliques modulo $2$ correctly over the $1 / 2 + \epsilon$-fraction of instances required of $O$. Moreover, this reduction has a success probability of greater than $2 / 3$.

\end{theorem}

Whereas \cite{Goldreich2020} asks whether, for every $\epsilon > 0$, there is a randomized reduction in $\tilde{O} \left( n^2 \right)$-time, with a success probability of larger than $2/3$, from computing the number of $k$-cliques modulo $2$ to computing the number of $k$-cliques modulo $2$ correctly on any subset $S \subset \{\, 0,1 \,\}^{n \choose 2}$ with $|S| = (1 / 2 + \epsilon) 2^{n \choose 2}$. We answer in the affirmative, not for all such subsets $S$, but for almost all such subsets $S$. We stress that while our result significantly improves the error tolerance from the previous state of the art of $2^{-k^2}$ of \cite{Goldreich2020} to $1/2-\epsilon$ for ``almost all $S$''-type results, the error tolerance of \cite{Goldreich2020} is still state of the art for ``all $S$.'' This introduces a tradeoff between error tolerance and universality of instances, where we have a sharp gain in error tolerance at the cost of universality of instances.

\subsection{Worst-Case to Rare-Case Reductions for Multigraph Counting Problems}
\label{section:1.2}

It has been believed for decades that there are no polynomial-time algorithms for $\NP$-hard problems, at least with sufficient faith in the conjecture that $\P \neq \NP$. In the past decade, efforts have been made to determine the exact complexities of $\NP$-hard problems. The world of fine-grained complexity attempts to create a web of reductions analogous to those of made by Karp reductions \citep{Karp1972}, except the margins are ``fine''. One such connection, proved by \cite{Williams2005} is that if the \textit{Orthogonal Vectors} ($\OV$) problem has an $n^{2-\epsilon}$-time algorithm for dimension $d = \omega(\log n)$, then the \textit{Strong Exponential Time Hypothesis} ($\SETH$) \citep{Calabro2009} is false.

Not only do minor algorithmic improvements under the framework of fine-grained complexity imply faster algorithms for many other problems, they can also prove structural lower bounds. \cite{Williams2013}, in pursuit of the answer to the question, ``What if every $\NP$-complete problem has a slightly faster algorithm?,'' proved that faster than obvious satisfiability algorithms for different classes of circuits imply lower bounds for that class. Soon after, by showing that there is a slightly better than exhaustive search for $\ACC$ circuits\footnote{More concretely, \cite{Williams2014} showed that for $\ACC$ circuits of depth $d$ and size $2^{n^{\epsilon}}$ for any $0 < \epsilon < 1$, there is a satisfiability algorithm taking $2^{n-n^{\delta}}$ time for some $\delta > 0$ depending on $\epsilon$ and $d$.}, \cite{Williams2014} showed that $\NEXP \not\subset \ACC$.

Currently, the fastest algorithm for computing the number of Hamiltonian cycles on digraphs takes $O^* \left( 2^{n - \Omega(\sqrt{n})} \right)$-time due to \cite{Li2023}. Some other algorithmic improvements upon $O^{*}(2^{n})$ (including the parameterized cases) are due to \cite{Bjorklund2019A}, \cite{Bjorklund2019B}, and \cite{Bjorklund2016}.

\subsubsection{Background}
\label{section:1.2.1}

\cite{Cai1999} proved that the permanent of an $n \times n$ matrix over $\mathbb{Z}_p$, a polynomial that, in essence counts the number of cycle covers of a multigraph modulo $p$ is as hard to evaluate correctly on a $1 / \poly(n)$-fraction of instances in polynomial-time as it is to evaluate over all instances in polynomial-time. Here, they used the list decoder of \cite{Sudan1996} to show a worst-case to rare-case reduction: A reduction using a polynomial-time algorithm that evaluates the polynomial on an $o(1)$-fraction of instances to construct a polynomial-time randomized algorithm that computes the permanent over this field over any input with high probability.

\cite{Goldreich2018} consider the problem of counting the number of $t$-cliques in an undirected multigraph. A $t$-clique is a complete subgraph of $t$ vertices ($K_t$). They give an $\left( \tilde{O} \left( n^2 \right), 1 / \polylog(n) \right)$-worst-case to rare-case reduction from counting $t$-cliques in $n$-vertex undirected multigraphs to counting $t$-cliques in undirected multigraphs generated according to a specific probability distribution. That is, given an oracle $O$ that can correctly count the number of $t$-cliques in undirected multigraphs generated according to a certain probability distribution on at least a $1 / \polylog(n)$-fraction of instances, in $\tilde{O} \left( n^2 \right)$-time, using the oracle $O$, we can count the number of $t$-cliques correctly in $n$-vertex undirected multigraphs with a success probability of at least $2 / 3$. Combined with the work of \cite{Valiant1979} and \cite{Cai1999}, they also show $1 / 2^{o(n)}$-hardness for computing the permanent in a setup similar to the one in our work.

Given a constant-depth circuit $C_L$ for verifying an $\NP$-complete language $L$, \cite{Tejas2024} created a generalized certificate counting function, $f^\prime_{L, p}: \mathbb{Z}_p^{n + 2 n^c} \to \mathbb{Z}_p$, where $p$ is a prime and $n^c$ is the certificate size for $L$. Further, using an appropriate set of functions, $f^\prime_{L, p}$, they prove that for all $\alpha > 0$, there exists a $\beta > 0$ such that the set of functions $f^{\prime \prime}_{L, \beta}$ is $1 / n^\alpha$-rare-case hard to compute under various complexity-theoretic assumptions.

There are two observations in the works of \cite{Tejas2024}.

\begin{enumerate}

\item The set of functions, $f^{\prime \prime}_{L, \beta}$, is artificially generated using the circuit $C_L$.

\item Proving $f^\prime_{L, p}$ to be rare-case hard for any $\NP$-complete language $L$ seems infeasible using their work.

\end{enumerate}

\subsubsection{Our Results}
\label{section:1.2.2}

In contrast to the above observations, our contributions in this paper are as follows.

\begin{enumerate}

\item From the problem description itself, we construct a generalized certificate counting polynomials, $f^\prime_{L, p}$, for two natural $\NP$-complete languages, which look more natural as compared to the above ``artificially'' generated functions, $f^{\prime \prime}_{L, \beta}$. The first problem counts the number of Hamiltonian cycles in a directed multigraph over $\mathbb{Z}_p$. The second problem counts the number of $\lfloor n / 2 \rfloor$-cliques in $n$-vertex undirected multigraphs over $\mathbb{Z}_p$.

\item We prove rare-case hardness results for the above two ``natural'' problems ($f^\prime_{L, p}$) for a prime $p = \Theta \left( 2^{n} \right)$ by exploiting their algebraic and combinatorial structures. 

\end{enumerate}

Assuming the \textit{Randomized Exponential Time Hypothesis} ($\rETH$) \citep{Dell2014}, the conjecture that any randomized algorithm for $3\SAT$ on $n$ variables requires $2^{\gamma n}$ time for some $\gamma > 0$, we show the following results.

\begin{theorem}
\label{thm:3}

Unless $\rETH$ is false, counting the number of unique Hamiltonian cycles modulo $p$ on an $n$-vertex directed multigraph requires $2^{\gamma n}$-time for some $\gamma > 0$ even to compute correctly on a $1 / 2^{n/\log n}$-fraction of instances for a prime, $p = \Theta \left( 2^{n} \right)$.

\end{theorem}

\begin{theorem}
\label{thm:4}

Unless $\rETH$ is false, counting the number of unique cliques of size $\lfloor n /2 \rfloor$ modulo $p$ on an $n$-vertex undirected multigraph requires $2^{\gamma n}$-time for some $\gamma > 0$ even to compute correctly on a $1 / 2^{n/\log n}$-fraction of instances for a prime, $p = \Theta(2^{n})$.

\end{theorem}

Meanwhile, for both problems, simply printing $0$ all the time, without even reading the input, in this setting yields the correct answer on at least an $\Omega \left( 1 / 2^{n} \right)$-fraction of instances.

Using ``weighted'' versions of counting problems in \cite{Goldreich2018} inspired our choice to use multigraphs. By ``unique,'' in the example of triangle counting, we take a choice of three vertices and compute the number of ``unique'' triangles between them by multiplying the ``edge weights.'' This multiplicative generalization, precisely to count the number of choices of one edge between any two vertices in our subgraph structure, is what we mean when we say ``unique cliques'' or ``unique Hamiltonian cycles.''

Our results extend the results obtained for the permanent \citep{Valiant1979, Feige1996, Cai1999, Dell2014, Bjorklund2019B, Li2023} to these two problems. This provides heuristic evidence that significantly improving algorithms for these two problems might be infeasible. Under $\rETH$, one needs exponential time to marginally improve the $O(1)$-time algorithm of always printing $0$.

\subsection{Techniques}
\label{section:1.3}

This paper's central unifying theme is using group theoretic arguments to obtain our results. In particular, a small subset of elementary arguments gives great mileage for both results. Moreover, the arguments made in both results complement each other in the following ways.

\begin{enumerate}

\item For the hardness amplification achieved for counting $k$-cliques modulo $2$ on an $n$-vertex simple undirected graph, the most important tool for us from the theory of group actions, is the \textit{Orbit Stabilizer Theorem} (Lemma \ref{lemma:4}). In the context of graphs, we can interpret this as saying that the automorphism group of a simple undirected $n$-vertex graph $U_n$, $\Aut \left( U_n \right)$, the subgroup of $S_n$ such that permuting the vertices and edges of $U_n$ by a permutation $\pi \in \Aut \left( U_n \right)$ conserves the adjacency matrix of $U_n$, is related to the isomorphism class $\mathcal{C}_n$ of distinct\footnote{We say that two $n$-vertex graphs are different if their adjacency matrices are different.} graphs isomorphic to $U_n$ as $\left| \Aut \left( U_n \right) \right| \left| \mathcal{C}_n \right| = n!$. As will be seen in section \ref{section:1.3.1}, this is the most important insight for us, along with the result of \cite{Polya1937} and \cite{Erdos1963} that almost all graphs have a trivial automorphism group.

\item For our results obtained for counting problems on multigraphs, our objects of algebraic study are the functions themselves. Our protagonists, weighted counting functions on multigraphs, form a vector space, also a group under addition. The space of weighted counting functions that are invariant under graphs isomorphism forms a subspace. We provide a valuable classification of these functions based on conjugacy class structure, and use these results. In particular, the intention of our usage of group theory here is to, given oracle access to a function $f$ under some constraints, find out if it computes our function of interest. First, for both problems, we test whether $f$ is a counting function on multigraphs, then we check if it is invariant under graph isomorphism, and finally, check whether $f$ is our function of interest.

\end{enumerate}

\subsubsection{For Counting $k$-Cliques Modulo $2$}
\label{section:1.3.1}

The following ``key ideas'' are helpful to keep in mind while going through the technical details of this work.

\subsubsection*{The Fraction of Correct Answers Over Large Isomorphism Classes is Usually not Too Far From the Expected Fraction of Correct Answers Over the Oracle}

When we sample $O$ from $O^{H_n}_{1 / 2 + \epsilon}$, we can think of the correct fraction of instances as being distributed over the isomorphism class partitions of $O$. While our proof in Section \ref{section:8.1} formalizes this fact using the Chernoff bounds \citep{Mitzenmacher2005}, as is usually the case with tail bounds, the intuitive picture to have in mind is the central limit theorem. As $n$ grows, for large isomorphism classes $\mathcal{C}_n$, the random variable representing the fraction of correct instances resembles a normal distribution centered at $1 / 2 + \epsilon$. As $n$ grows, for sufficiently large isomorphism classes, almost all the weight of the distribution is concentrated between $1 / 2 + \epsilon / 2$ and $1 / 2 + 3 \epsilon/ 2$. The fact that the weight in the region $[0, 1/2+\epsilon]$ is small, in fact exponentially low, is useful to us. In fact, using the union bound, we show that for sufficiently large $n$, all sufficiently large (say $\left |\mathcal{C}_n \right| \geq n^3$) isomorphism classes have a correctness fraction greater than $1 / 2 + \epsilon$ over $O$. 

\subsubsection*{Almost All Graphs Belong to Isomorphism Classes of the Largest Possible Size}

In their work, \cite{Polya1937} and \cite{Erdos1963} showed that almost all $n$-vertex undirected simple graphs have a trivial automorphism group. More specifically, if we randomly sample a graph $U_n$ uniformly from the set of all undirected simple graphs with $n$ vertices, with a probability of $1 - {n \choose 2} 2^{-n - 2} (1 + o(1))$, $\left| \Aut \left( U_n \right) \right| = 1$. Due to our version of the orbit stabilizer theorem (Lemma \ref{lemma:4}), this means that almost all graphs belong to an isomorphism class of size $n!$.

\subsubsection*{Graphs With Very Large Automorphism Groups are Easy to Count Cliques Over}

One can imagine that with a graph whose automorphism group is of ``almost full size,'' perhaps when seen on a logarithmic scale, counting $k$-cliques is easy. With a highly symmetric graph, if we have a $k$-clique, we have many others in predictable positions. It is also likely that the number of $k$-cliques in this graph is represented by a small arithmetic expression consisting of binomial coefficients. For progress on the problem of \cite{Goldreich2020}, for sufficiently large $n$, we classify all graphs with $n$ vertices whose automorphism group is of size $\omega \left( n! / n^3 \right)$. For sufficiently large $n$, there are only twelve non-isomorphic graphs of this type. All of them either have an independent set with $n - 2$ vertices or a clique containing $n - 2$ vertices. Also, six of these classes have zero $k$-cliques for $k > 2$, five have the number of $k$-cliques described by an arithmetic expression containing one binomial coefficient, and only one has its $k$-clique count as the difference between two binomial coefficients.

Keeping this intuition in mind, the paradigm for our reduction is as follows:

\begin{enumerate}

\item Check if our graph, $U_n$, belongs to an isomorphism class that is large enough to have good probabilistic guarantees of having a $1 / 2 + \epsilon / 2$-fraction of correctness over the randomness of $O^{H_n}_{1 / 2 + \epsilon}$. In particular, this size threshold grows as $\Theta \left( n^2 / \epsilon^2 \right)$.

\item If the isomorphism class is large enough, then with a very high probability over the randomness of $O^{H_n}_{1 / 2 + \epsilon}$, this class has at least a $1 / 2 + \epsilon / 2$-fraction of correctness over $O$. We sample random permutations $\pi$ from $S_n$ and permute the vertices and edges of $U_n$ accordingly to obtain a graph $U^\prime_n$ isomorphic to $U_n$. We query $O$ on the input $U^\prime_n$ and note down the answer. We repeat this process $O \left( 1 / \epsilon^2 \right)$ times and take the majority answer. Due to the Chernoff bound, once again, if we do have a $1 / 2 + \epsilon / 2$-fraction of correctness within the isomorphism class for $O$, this is correct with high probability over the randomness of the algorithm.

\item If the isomorphism class is small, the graph is highly symmetric, and we count the number of $k$-cliques ourselves.

\end{enumerate}

We execute this paradigm differently for a constant $\epsilon > 0$ and for an $\epsilon$ varying as a function of $n$.

\textbf{For a Constant $\epsilon > 0$.} When this is the case, notice that our critical threshold for isomorphism class size is $O \left( n^2 \right)$. Due to the orbit stabilizer theorem (Lemma \ref{lemma:4}) for graphs, this means that the automorphism group of every graph $U_n$ with isomorphism class size $O \left( n^2 \right)$ has $ \left| \Aut \left( U_n \right) \right| = \Omega \left( n! / n^2 \right) = \omega \left( n! / n^3 \right)$. In Section \ref{section:8.2.1}, we will prove in Lemma \ref{lemma:23} that for sufficiently large $n$, the following are the only kinds of graphs with automorphism group of size $\omega \left( n! / n^3 \right)$.

\begin{enumerate}

\item $K_n$ and its complement.

\item $K_n$ with one edge missing and its complement.

\item $K_{n - 1}$ with an isolated vertex and its complement.

\item $K_{n - 1}$ with one vertex of degree $1$ adjacent to it and its complement.

\item $K_{n - 2}$ with two isolated vertices and its complement.

\item $K_{n - 2}$ with two vertices adjacent to each other and its complement.

\end{enumerate}

In $\tilde{O} \left( n^2 \right)$-time, by checking each case, we can tell whether $U_n$ is isomorphic to any of these graphs and quickly compute the number of $k$-cliques if so. If $U_n$ is not isomorphic to any of these, then, due to the orbit stabilizer theorem for graphs (Lemma \ref{lemma:4}), its isomorphism class size is above the critical threshold, and we can query on $O$ for answers.

\textbf{For an $\epsilon$ Varying as a Function of $n$.} When $\epsilon$ varies as a function of $n$, the procedure here varies since obtaining a complete classification of graphs whose automorphism class is above the size threshold is impractical. Let $t(n) = O \left( n^2 / \epsilon^2 \right)$ be the threshold isomorphism class size in this case. We estimate whether the automorphism class of $U_n$ is larger than $n! / t(n)$ or smaller than $n! / t(n)^{1 + \alpha}$ for some $\alpha > 0$. We can do this by taking $nt(n)$ random permutations $\pi$ from $S_n$ and counting how often permuting the vertices and edges of the graph $U_n$ as specified by $\pi$ gives us the same adjacency list as $U_n$. If the automorphism group is larger than $n! / t(n)$, then with high probability, this count is larger than $n / 2$. If the automorphism group size is smaller than $n! / t(n)^{1 + \alpha}$, then this is very likely to be less than $n / 2$; hence, we decide based on comparing this number to $n / 2$.

The algorithm to count $k$-cliques on the symmetric case is also different since we no longer have a convenient classification of graphs anymore. In particular, we first attempt to list all (at most $t(n)$) distinct graphs isomorphic to $U_n$. We can do this by picking $n^2 t(n)$ random permutations $\pi$ from $S_n$ and permuting $U_n$ according to $\pi$. If this is a graph we have not yet seen, then we add it to the list. With high probability, we will have seen all graphs. In each of these graphs, we count how many cases the first $k$ vertices form a $k$-clique. As shown in Section \ref{section:8.3}, the number of $k$-cliques in this graph is a simple function of this number.

When the isomorphism class is of size above the critical threshold, we can, of course, use the querying procedure to $O$ and obtain good probabilistic guarantees over the randomness of $O^{H_n}_{1 / 2 + \epsilon}$.

\subsubsection{For the Rare-Case Hardness of Counting on Multigraphs}
\label{section:1.3.2}

We will discuss the overview of the proof for the problem of counting Hamiltonian cycles on directed multigraphs. The techniques to prove the analogous results counting the number of unique cliques of size $\lfloor n / 2 \rfloor$ are very similar.

\subsubsection*{$\ETH$-Hardness of Computing the Number of Hamiltonian Cycles Modulo $p$ on a Directed Multigraph}

Note that due to the $O(n + m)$-space reduction from $3\SAT$ on $n$ variables and $m$ clauses to the problem of deciding whether there is a clique of size $\lfloor n / 2 \rfloor$ in an undirected multigraph (Appendix \ref{appendix:A}) or deciding whether there is a Hamiltonian cycle in a directed multigraph, along with the \textit{Sparsification Lemma} of \cite{Impagliazzo2001} (Lemma \ref{lemma:6}), neither of these problems should have $2^{o(n)}$-time algorithms under the \textit{Exponential Time Hypothesis} ($\ETH$) \citep{Impagliazzo2001b}, the hypothesis that $3\SAT$ on $n$ variables requires $2^{\gamma n}$-time for some $\gamma > 0$. We show, due to a randomized reduction from the decision problems (Lemmas \ref{lemma:7} and \ref{lemma:8}) that we cannot count for growing $p$, the number of unique cliques of size $\lfloor n / 2 \rfloor$ in an undirected multigraph or Hamiltonian cycles in a directed multigraph in $2^{o(n)}$-time under $\rETH$; however, since the algorithm for $3\SAT$ would be randomized in the case of an algorithm for these problems.

\subsubsection*{Hardness Amplification Using the STV List Decoder}

The \textit{STV List Decoder} of \cite{Sudan2001} (Lemma \ref{lemma:2}) is a potent tool for error correction. Formally, we speak more about it in Section \ref{section:2.3}, but in essence, given an oracle that is barely, but sufficiently correct on some polynomial $f$ of degree at most $d$, the STV list decoder gives us some number of machines $M$ computing polynomials of degree at most $d$, one of which is our function of interest. We use this list decoder to obtain a probabilistic algorithm correct on all inputs from an algorithm that is correct on a small, vanishing fraction of instances. We are not the first to use the STV list decoder to prove hardness results. Our usage of it is inspired by its usage in \cite{Goldreich2018}. \cite{Goldenberg2020} shows one more such application of this tool to amplify hardness.

\textbf{Identifying the Correct Machine.} On the problem of amplifying from a barely correct algorithm, we use the STV list decoder (Lemma \ref{lemma:2}), which gives us some machines $M$, all of which compute polynomials of degree upper bounded by the degree of our function of interest. So, we iterate through each machine and test whether it computes our function of interest. A rough outline of this test is as follows.

\begin{enumerate}

\item Given a machine $M$, we first test whether it computes a ``valid'' multigraph counting function. The techniques we use here are the pigeonhole principle based techniques for counting Hamiltonian cycles and interpolation techniques for counting half-cliques.

\item Given that the function is promised to compute a ``valid'' multigraph counting function, how do we know if it is invariant under graph isomorphism? The test relies on straightforward ideas: Lagrange's theorem \citep{Herstein1975}, the idea for finite groups that the order of a subgroup $H$ (of $G$) must divide the order of $G$ and the somewhat silly fact that the smallest integer larger than $1$ is $2$. Suppose we have a counting function $H_{n, p}$ on multigraphs. Let $\Pi \left( H_{n, p} \right)$ be the subgroup of $S_n$ such that permuting the vertices and edges to the input graph of $H_{n, p}$, for any input graph, does not change the output. If $H_{n, p}$ is invariant under graph isomorphism, then $\Pi \left( H_{n, p} \right)$ is $S_n$. However, if $\Pi \left( H_{n, p} \right)$ is not $S_n$, then it is at most half the size of $S_n$. Indeed, this is precisely the insight we use. We pick a random graph and a random permutation from $S_n$. For sufficiently large $n$, the probability that the function $H_{n, p}$ does not change throughout this operation is close to $\left| \Pi \left( H_{n, p} \right) \right| / |S_n|$. If $H_{n, p}$ is indeed invariant under graph isomorphism, then $\left| \Pi \left( H_{n, p} \right) \right| / |S_n| = 1$ and otherwise, $\left| \Pi \left( H_{n, p} \right) \right| / |S_n| \leq 1 / 2$, and we reject with a probability of roughly $1 / 2$. 

\item In this step, we try to identify our functions of interest, guaranteed that the machine computes an invariant function under graph isomorphism. For both problems, we classify all graph counting functions based on insight from conjugacy classes and use that to our advantage. For the problem of counting Hamiltonian cycles, the insight is that this function is the only one that places zero weight on any cycle cover other than the Hamiltonian cycles. In the case of counting half-cliques, the argument is more complicated.

\end{enumerate}

\section{Preliminaries}
\label{section:2}

\subsection{Notations}
\label{section:2.1}

$\mathcal{P}[\mathcal{E}]$ denotes the probability of the event $\mathcal{E}$ occurring. $\mathbb{N}$ is the set of natural numbers $\{\, 1, 2, 3, \ldots \,\}$. $\mathbb{R}$ is the set of real numbers. $\mathbb{Z}_p$ refers to the field of size $p$, where $p$ is a prime number. $\mathbb{F}$ denotes a finite field. If an algorithm has a time complexity of $O^{*}(T(n))$, where $T$ is super-polynomial, then the time complexity is bounded by $T(n)h(n)$ for a polynomial $h$. The notation $[n]$ is for the set $\{\, 1, 2, \ldots, n \,\}$. For a set $S$ and an integer $n$, $\displaystyle \binom{S}{n}$ is the set of all subsets of $S$ with $n$ elements. For a real number $r$, $\lfloor r\rfloor$ and $\lceil r\rceil$ are the smallest integer not larger than $r$ and the largest integer not smaller than $r$, respectively. $K_n$ is the complete simple undirected graph on $n$ vertices with an edge between any pair of vertices. The functions $\poly(n)$ and $\polylog(n)$ denote any function bounded from above by a polynomial in $n$ and a polynomial in the logarithm of $n$, respectively. $\tilde{O} (T(n))$ means $O(T(n) \polylog(n))$. The notation $x \gets_r S$ means that the variable $x$ takes on a random value uniformly sampled from a finite set $S$. $(x_i)_{i \in [n]}$ is a short representation of the list $X = (x_1, x_2, \ldots, x_n)$. Similarly, $\left( \left( x_{i,j} \right)_{j \in [n]} \right)_{i \in [m]}$ is a short representation of the list $(x_{1,1}, \ldots, x_{1, n}, x_{2, 1}, \ldots x_{2, n}, \ldots, x_{m, 1}, \ldots, x_{m, n})$. $X.x_i$ is for the variable $x_i$ within the ordered list $X$. Group theoretic notation is defined when used and often restated where required. $\Aut \left( U_n \right)$ refers to the automorphism group of a simple $n$ vertex undirected graph $U_n$.

\subsection{The Schwartz-Zippel Lemma}
\label{section:2.2}

The analogue of the fundamental theorem of algebra in the multilinear setting is the following lemma of \cite{Schwartz1980}, \cite{Zippel1979}, and \cite{Demillo1978}. The lemma has seen many forms across papers; here, we present the form we use.

\begin{lemma}\textbf{The Schwartz-Zippel Lemma.}\\
\label{lemma:1}
 
Given a multivariate polynomial $f:\mathbb{F}^m \to \mathbb{F}$ of degree $d$, given any subset $\mathcal{D} \subset \mathbb{F}$,
\begin{equation*}
\mathbb{P}_{x \leftarrow_{r}\mathcal{D}^m}[f(x) = 0] \leq \frac{d}{|\mathcal{D}|}.
\end{equation*}
\end{lemma}
 
The Schwartz-Zippel lemma is the most crucial algebraic tool used in this paper to make statements of a combinatorial nature.

\subsection{The STV List-Decoder}
\label{section:2.3}

The problem of list decoding asks whether, given an oracle $O$ computing a specific function $f$, but on a tiny fraction of inputs, one can recover the function $f$? This is a problem with important applications in practice and theory, especially in the amplification of hardness and the theory of pseudorandomness \citep{Vadhan2012}. In their paper constructing pseudorandom generators without using the XOR lemma, \cite{Sudan2001} proved the following list-decoding lemma, allowing a robust hardness amplification. This list decoder is the final product of a long line of work in coding theory, each implying stronger and stronger amplifications of hardness \citep{Gemmell1992, Sudan1996, Sudan2001}. 

\begin{lemma}\textbf{The STV List Decoder.}\\
\label{lemma:2}

Given any oracle $O$ that computes a polynomial $p : \mathbb{F}^n \to \mathbb{F}$ of degree $d$ correctly on over an $\epsilon > \sqrt{2d / |\mathbb{F}|}$-fraction of instances, in $\poly(n, d, 1 / \epsilon, \log |\mathbb{F}|)$-time, we can produce $O(1 / \epsilon)$ randomized oracle machines (with oracle access to $O$), all of which compute some multivariate polynomial from $\mathbb{F}^n$ to $\mathbb{F}$ of degree $d$, one of which computes $f$. Moreover, each machine runs in $\poly(n, d, 1 / \epsilon, \log |\mathbb{F}|)$-time and disagrees with the polynomial it intends to compute with a probability of at most $1 / 2^{q(n)}$ for some polynomial $q$.
\end{lemma}

Due to this lemma, we will refer to the list decoder as the ``STV list decoder.'' \cite{Goldreich2018} also used this to show the rare-case hardness of the function counting $t$-cliques on multigraphs.

\subsection{The Symmetric Group $S_n$}
\label{section:2.4}

The symmetric group, $S_n$, is the group of permutations of the set $[n]$ or the bijections $\pi:[n] \to [n]$ under function composition \citep{Herstein1975}. For any $j \in [n]$, if we consider the set $\{\, \pi^i(j) \mid i \in \{\, 0 \,\} \cup [n] \,\}$, we observe that it is a set having at most $n$ elements because using the \textit{Pigeonhole Principal}, there are at least two repeated values. The result is a cycle. Thus, we can uniquely decompose each permutation $\pi \in S_n$ as a set of disjoint cycles. Two permutations $\pi_1 \in S_n$ and $\pi_2 \in S_n$, are said to have similar cycle decompositions if the lengths of their constituent cycles are identical.

We define the conjugacy classes of a group as follows.

\begin{definition}
\label{def:2}

\textit{\textbf{Conjugacy Classes of a Group.} \\
We say that an element $b$ of a group $G$ is conjugate to $a \in G$ if there is an element $g \in G$ such that $b = g^{-1} a g$. We represent this relation as $a \sim b$. The set $C(a) = \{b \in H| a \sim b\}$ is called the conjugacy class of $a$ in $G$ under the relation $\sim$.}

\end{definition}

It is easy to see that the relation $\sim$ is an equivalence relation, following directly from the group axioms and that the conjugacy classes partition the group $G$. Two permutations $\pi_1 \in S_n$ and $\pi_2 \in S_n$ are conjugates if and only if both have a similar cycle decomposition. Therefore, the number of conjugacy classes of $S_n$ is the number of partitions of $n$. We denote by $C_n$, the conjugacy class of $S_n$ having all permutations $\pi$, which are cycles of length $n$.

We define the subgraphs induced by a permutation as follows.

\begin{definition}
\label{def:3}

\textit{\textbf{Induced Subgraphs.} \\
Given a directed multigraph $D = (V, E)$ on $n$ vertices ($V = [n]$), and a permutation $\pi \in S_n$, we define the subgraph of $D$ induced by $\pi$, $D_\pi$, as 
\begin{equation*}
D_\pi = (V, E_\pi),
\end{equation*}
where
\begin{equation*}
E_\pi = \{\, (i, \pi(i)) \mid i \in [n] \text{ and } (i, \pi(i)) \in E  \,\}.
\end{equation*}
}

\end{definition}

Now, using the induced subgraphs, we define the cycle cover of a multigraph.

\begin{definition}
\label{def:4}

\textit{\textbf{Cycle Cover of a Multigraph.} \\
Given a directed multigraph $D = (V, E)$ on $n$ vertices ($V = [n]$), and a permutation $\pi \in S_n$, we say that the induced subgraph $D_\pi$ is a cycle cover of $D$ if 
\begin{equation*}
|E_\pi| = n.
\end{equation*}
That is, if each permutation mapping $i \to \pi(i)$ for $i \in [n]$ is present as an edge in $D$. The permutation $\pi$ decomposes $D_\pi$ as a set of disjoint cycles.}

\end{definition}

We give the usual definition of directed graph isomorphism.

\begin{definition}
\label{def:5}

\textit{\textbf{Directed Graph Isomorphism.} \\
Two directed graphs, $D_1 = (V, E_1)$ and $D_2 = (V, E_2)$, where $V = [n]$, are called isomorphic (denoted as $D_1 \sim D_2$) if there exists a permutation $\pi \in S_n$ such that
\begin{equation*}
(i, j) \in E_1 \iff (\pi(i), \pi(j)) \in E_2, \quad \forall (i, j) \in [n]^2.
\end{equation*}
That is, by relabeling the vertices of $D_1$ by applying the permutation $\pi \in S_n$, we can get $D_2$.}

\end{definition}

Now, we present and prove a vital lemma concerning the cycle covers of a multigraph.

\begin{lemma}
\label{lemma:3}

\textbf{Any Two Cycle Covers of a Directed Multigraph Corresponding to the Same Conjugacy Class of $S_n$ are Isomorphic.} \\
Let $\pi_1 \in S_n$ and $\pi_2 \in S_n$ be such that both $D_{\pi_1}$ and $D_{\pi_2}$ are cycle covers of $D$. Then
\begin{equation*}
\pi_1 \sim \pi_2 \iff D_{\pi_1} \sim D_{\pi_2}.
\end{equation*}

\begin{proof}

\begin{equation*}
\begin{split}
& D_{\pi_1} \sim D_{\pi_2} \\
& \implies \text{ $\pi_1$ and $\pi_2$ have similar cycle decomposition} \\
& \implies \pi_1 \sim \pi_2.
\end{split}
\end{equation*}
\begin{equation}
\label{eq:1}
\begin{split}
& \pi_1 \sim \pi_2 \\
& \implies \exists \pi \text{ such that } \pi_2 = \pi^{-1} \pi_1 \pi \\
\end{split}
\end{equation}
Let $D = (V, E)$, such that $V = [n]$. From Equation \eqref{eq:1}, we have
\begin{equation}
\label{eq:2}
\begin{split}
E_{\pi_2} & = \{\, (i, \pi_2(i)) \mid i \in [n]  \,\} \\
& = \left \{ \, \left( i, \pi^{-1} \pi_1 \pi (i) \right) \mid i \in [n]  \, \right \} \\
& = \left \{\, \left( i, \pi \left( \pi_1 \left( \pi^{-1} (i) \right) \right) \right) \mid i \in [n]  \, \right \}.
\end{split}
\end{equation}
Applying $\pi^{-1}$ to Equation \eqref{eq:2}, we get
\begin{equation}
\label{eq:3}
\begin{split}
E_{\pi_2 \pi^{-1}} & = \left \{\, \left( \pi^{-1}(i), \pi_2 \pi^{-1} (i) \right) \mid i \in [n]  \, \right \} \\
& = \left \{ \, \left( \pi^{-1} (i), \pi^{-1} \pi_1 \pi \pi^{-1} (i) \right) \mid i \in [n]  \, \right \} \\
& = \left \{\, \left( \pi^{-1} (i), \pi_1 \left( \pi^{-1} (i) \right) \right) \mid i \in [n]  \, \right \} \\
& = \{\, (j, \pi_1(j)) \mid j \in [n] \,\} \\
& = E_{\pi_1},
\end{split}
\end{equation}
where, $j = \pi^{-1}(i)$. Equation \eqref{eq:3} implies that
\begin{equation*}
D_{\pi_1} \sim D_{\pi_2}.
\end{equation*}

\end{proof}

\end{lemma}

\subsection{Group Actions}
\label{section:2.5}

For a group $G$ and a set $X$, a group action \citep{Smith2015} is a function $\alpha:X \times G \to X$ satisfying the following axioms:

\begin{enumerate}

\item $\alpha(x, e) = x$ for every $x \in X$, where $e$ is the identity of $G$.

\item $\alpha(x, ab) = \alpha(\alpha(x, a), b)$ for every $x \in X$ and $a, b \in G$.

\end{enumerate}

Informally, group actions induce symmetries on $X$. In this work, $X$ is the set of $n$-vertex undirected graphs, and $G$ is usually the symmetric group on $n$ vertices, $S_n$. We will discuss group actions in the context of actions on graphs. Given a simple undirected graph $U_n$ on $n$ vertices, $\Aut \left( U_n \right)$ is the subgroup of $S_n$ such that all actions of permutations from $\Aut \left( U_n \right)$ fix $U_n$. That is, $\alpha(U_n, \sigma) = U_n$ if and only if $\sigma \in \Aut \left( U_n\right)$. The set $\mathcal{C}_n$ of a graph represented by $U_n$ is $\{\, \alpha \left( U_n, \sigma \right) | \sigma \in S_n \,\}$. We call it the isomorphism class of the graph, but in the language of group actions, this is also the orbit of an element. Now, the most critical insight from the theory of group actions is the following simple consequence of the \textit{Orbit Stabilizer Theorem}.

\begin{lemma}
\label{lemma:4}

For every simple undirected $n$-vertex graph represented by $U_n$,
\begin{equation*}
\left| \Aut \left( U_n \right) \right| \left| \mathcal{C}_n \right| = n!, 
\end{equation*}
where $\mathcal{C}_n$ is the orbit or isomorphism class of $U_n$.

\end{lemma}

This is the single most helpful lemma for our results in Section \ref{section:8}. The picture to remember is that the automorphism group is a subgroup of $S_n$, and the cosets correspond to isomorphic graphs to our original graph. We give one more useful straightforward lemma below.

\begin{lemma}
\label{lemma:5}

For every $n$ vertex simple, undirected graph represented by $U_n$,
\begin{equation*}
\Aut \left( U_n \right) \cong \Aut \left( \overline{U_n} \right),
\end{equation*}
and hence, 
\begin{equation*}
\left| \Aut \left( U_n \right) \right| = \left| \Aut \left( \overline{U_n} \right) \right|.
\end{equation*}

\end{lemma}

\subsection{Property Testing}
\label{section:2.6}

\cite{Goldreich2017} describes that ``property testing is concerned with the design of super-fast randomized algorithms for approximate decision-making.'' Approximate decision-making, in this context, means distinguishing objects with a property from objects far from having the property\footnote{When ``far from having the property'' is strengthened to not having the property, this task is generally not tractable. For example, asking if a $3$CNF formula is satisfiable, or, asking if $x \in L$ or not, where $L$ is undecidable. Even distinguishing $\SAT$ formulas that have exactly one satisfying assignment from those that have none in randomized polynomial-time would imply $\NP = \RP$, due to the famous theorem of \cite{Valiant1986}.}. The textbook example is \textit{Blum Luby Rubinfeld (BLR)} linear testing \citep{Blum1993}, where one wants to distinguish oracles computing linear functions in $\mathbb{Z}_2$ from oracles that compute functions that are far from linear for randomly sampled strings $x$ and $y$ from the uniform distribution and reject otherwise. The test is straightforward, to merely accept if $f(x) + f(y) = f(x + y)$. \textit{Probabilistically Checkable Proofs (PCP)} \citep{Dinur2007, Feige1991} are also an example of this. Showing that any proof that is verifiable in polynomial-time can be turned into a proof that is verifiable by a small number of randomly chosen tests was a breakthrough in the field. Under reasonable cryptographic assumptions, the zero-knowledge proof for $3\COL$ relies on several simple tests \citep{Goldwasser1985, Blum1986, Goldreich1991}.

In this paper, in Section \ref{section:6}, we show the positive algorithmic result that given an oracle machine $M$ promised to compute some polynomial $H_{n, p}: \mathbb{Z}^{n^2}_p \to \mathbb{Z}_p$ of degree at most $n$ and sufficiently large $p$, we can check quickly whether $H_{n, p}$ is the polynomial computing the number of Hamiltonian cycles over $\mathbb{Z}_p$ in a directed multigraph. Moreover, we have a series of simple tests that obtain this result. We show similar results for the half-clique polynomial as well in Section \ref{section:7}. The ``sufficiently far apart'' clause is taken care of by the Schwartz-Zippel Lemma (\ref{lemma:1}): Distinct polynomials over large finite fields seldom align.

Note that in the modern sense, property testing refers to quick sublinear-time decision-making. In this work, for the results on counting directed Hamiltonian cycles and half-cliques, given some oracle along with some guarantees, we use polynomially many queries to make our decisions, which is sublinear, given the number of queries we make is polylogarithmic in the size of the domain - $\poly(n)$ many queries while the domain is of size $2^{\poly(n)}$.

\subsection{$\ETH$ for $3\SAT$ with $O(n)$ Clauses}
\label{section:2.7}

In our pursuit for tight reductions, the standard reduction from a $3\SAT$ instance with $n$ variables and $m$ clauses to instances of our graph problems with $O(n+m)$ vertices would not be sufficient without some modifications. The modification that gets us over the line is the \textit{Sparsification Lemma} \citep{Impagliazzo2001} which is used very frequently in the world of fine-grained complexity to get tight and optimal reductions.

\begin{lemma}
\label{lemma:6}

\textit{\textbf{The Sparsification Lemma.}}\\
For every $\epsilon > 0$, we have an algorithm taking any $k\SAT$ instance $\phi$ on $n$ variables and outputting $t$ $k\SAT$ instances $\left( \phi_i \right)_{i = 1}^t$ such that $t = O^{*}(2^{\epsilon n})$, each $\phi_i$ has $O(n)$ clauses, and $\phi$ is satisfiable if and only if at least one of $\left( \phi_i \right)_{i = 1}^t$ is satisfiable. Moreover, this algorithm runs in $O^{*}(2^{\epsilon n})$-time.

\end{lemma}

Due to this, we have the following very useful corollary.

\begin{corollary}
\label{corollary:1}

If $\ETH$ is true, then for all $\delta > 0$, there is no algorithm taking $2^{o(n)}$-time deciding $3\SAT$ with $O(n)$ clauses, where $n$ is the number of variables.
\end{corollary}
\begin{proof}
Due to the Sparsification Lemma (\ref{lemma:6}), for any $3\SAT$ instance, we can get $O^{*}(2^{\epsilon n})$ linear instances in $O^{*}(2^{\epsilon n})$-time. If we could solve each of these in $2^{o(n)}$-time, then we would have an $O^{*}(2^{\epsilon n + o(n)})$-time algorithm for $3\SAT$ for every $\epsilon > 0$, violating $\ETH$.
\end{proof}

Note that the constant $c$ we get from the $cn = O(n)$ clauses from the algorithm of the Sparsification lemma does not hinder us since $2^{o(cn)} = 2^{o(n)}$ for all constant $c > 0$.

\section{Problem Formulation}
\label{section:3}

This section defines and gives formulations for the following two problems.

\begin{enumerate}

\item Counting the number of \textit{Hamiltonian cycles} over $Z_p$ in a \textit{directed multigraph}.

\item Counting the number of \textit{half-cliques} over $Z_p$ in an \textit{undirected multigraph}.

\end{enumerate}

\subsection{Counting the Number of Hamiltonian Cycles Over $Z_p$ in a Directed Multigraph}
\label{section:3.1}

Let $D_{n, p} = \left( V_n, E_{n, p} \right)$ be a directed multigraph. The vertex set of $D_{n, p}$ is given as $V_n = [n]$. The adjacency matrix of $D_{n, p}$ is given as an ordered $n^2$-tuple,
\begin{equation}
\label{eq:4}
E_{n, p} = \left( \left( e_{(i, j)} \right)_{j \in [n]} \right)_{i \in [n]}, \quad e_{(i, j)} \in \mathbb{Z}_p, \quad \forall (i, j) \in [n]^2.
\end{equation}

We consider the problem of counting over $\mathbb{Z}_p$, the number of unique Hamiltonian cycles in $D_{n, p}$, with the number of edges from the vertex $i$ to the vertex $j$ from $\mathbb{Z}_p$. Using Equation \eqref{eq:4}, consider the polynomial $HCY_{n, p}: \mathbb{Z}_p^{n^2} \to \mathbb{Z}_p$ defined as
\begin{equation}
\label{eq:5}
\HCY_{n, p} \left( E_{n, p} \right) = \sum_{\sigma \in C_n} \prod_{i \in [n]} e_{(i, \sigma(i))}, 
\end{equation}
where $C_n$ is the set of cycles of length $n$ in the symmetric group $S_n$. One can easily see that the polynomial $\HCY_{n, p}$ is exactly equal to the number of Hamiltonian cycles modulo $p$ on the directed multigraph $D_{n, p}$.

We define the language $\dHAMCYCLE$ as the set of all \textit{directed graphs} having at least one Hamiltonian cycle, and the language $\dHAMCYCLE_p$ as the set of all \textit{directed multigraphs} having a non-zero number of Hamiltonian cycles over $\mathbb{Z}_p$.

\begin{definition}
\label{def:6}

\textit{\textbf{The Language $\dHAMCYCLE$.} \\
Let $D_n = \left( V_n, E_n \right)$ be a directed graph with $V_n = [n]$. The language $\dHAMCYCLE$ is defined as 
\begin{equation*}
\dHAMCYCLE = \{\, D_n \mid n \in \mathbb{N}, \exists \pi \in C_n, \left( D_n \right)_\pi \text{ is a cycle cover of } D_n \,\}.
\end{equation*}
}

\end{definition}

\begin{definition}
\label{def:7}

\textit{\textbf{The Language $\dHAMCYCLE_p$.} \\
Let $D_{n, p} = \left( V_n, E_{n, p} \right)$ be a directed multigraph with $V_n = [n]$ such that the number of vertices from any vertex $i$ to any vertex $j$ is from $\mathbb{Z}_p$. The language $\dHAMCYCLE_p$ is defined as 
\begin{equation*}
\dHAMCYCLE_p = \{\, D_{n, p} \mid n \in \mathbb{N}, |\{\, \pi \mid \pi \in C_n, \left( D_{n, p} \right)_\pi \text{ is a cycle cover of } D_{n, p} \,\}| \not \equiv 0 \mod p \,\}.
\end{equation*}
}

\end{definition}

\subsection{Counting the Number of Half-Cliques Over $Z_p$ in an Undirected Multigraph}
\label{section:3.2}

Let $U_{n, p} = \left( V_n, F_{n, p} \right)$ be an undirected multigraph. The vertex set of $U_{n, p}$ is given as $V_n = [n]$. The adjacency matrix of $U_{n, p}$ is given as an ordered $\displaystyle \binom{n}{2}$-tuple,
\begin{equation}
\label{eq:6}
F_{n, p} = \left( \left( e_{\{\, i, j \,\}} \right)_{j \in [n]} \right)_{1 \leq i < j}, \quad e_{\{\, i, j \,\}} \in \mathbb{Z}_p, \quad \forall (i, j) \in [n]^2.
\end{equation}

We consider the problem of counting over $\mathbb{Z}_p$, the number of cliques of size $\lfloor n / 2 \rfloor$ in $U_{n, p}$, with the number of edges between any two vertices $i$ and $j$ from $\mathbb{Z}_p$. Each clique is a collection of $\lfloor n / 2 \rfloor$ vertices and one edge between each pair of vertices. Using Equation \eqref{eq:6}, consider the polynomial $\HCL_{n, p}: \mathbb{Z}_p^{\binom{n}{2}} \to \mathbb{Z}_p$ defined as
\begin{equation}
\label{eq:7}
\HCL_{n, p} \left( F_{n, p} \right) = \sum_{S \in \binom{[n]}{\lfloor n / 2 \rfloor}} \prod_{\substack{(i, j) \in S^2 \\ 1 \leq i < j}} e_{\{\, i, j \,\}}, 
\end{equation}
where $S$ is a subset of $[n]$ of size $\lfloor n / 2 \rfloor$. We can easily see that the polynomial $\HCL_{n, p}$ is exactly equal to the number of cliques of size $\lfloor n / 2 \rfloor$ in the undirected multigraph $U_{n, p}$.

We define the language $\HALF$ as the set of all \textit{undirected graphs} having at least one half-clique, and the language $\HALF_p$ as the set of all \textit{undirected multigraphs} having a non-zero number of half-cliques over $\mathbb{Z}_p$.

\begin{definition}
\label{def:8}

\textit{\textbf{The Language $\HALF$.} \\
Let $U_n = \left( V_n, F_n \right)$ be an undirected graph with $V_n = [n]$. The language $\HALF$ is defined as 
\begin{equation*}
\HALF = \left \{\, U_n \mid n \in \mathbb{N}, K_{\lfloor n / 2 \rfloor} \text{ is a subgraph of } U_n \, \right \}.
\end{equation*}
}

\end{definition}

\begin{definition}
\label{def:9}

\textit{\textbf{The Language $\HALF_p$.} \\
Let $U_{n, p} = \left( V_n, F_{n, p} \right)$ be an undirected multigraph with $V_n = [n]$ such that the number of vertices between any vertices $i$ and $j$ is from $\mathbb{Z}_p$. The language $\HALF_p$ is defined as 
\begin{equation*}
\HALF_p = \left \{\, U_{n, p} \mid n \in \mathbb{N}, |\{\, K_{\lfloor n / 2 \rfloor} \mid K_{\lfloor n / 2 \rfloor} \text{ is a subgraph of } U_{n, p} \,\}| \not \equiv 0 \mod p \, \right \}.
\end{equation*}
}

\end{definition}

\section{$\NP$-Hardness of Counting Problems on Multigraphs Under Randomized Reductions}
\label{section:4}

We now show randomized reductions from the $\NP$-complete version of our problems ($\dHAMCYCLE$ and $\HALF$) to the ``counting on multigraphs'' version of our problems ($\dHAMCYCLE_p$ and $\HALF_p$), respectively. As can be seen in both cases, counting modulo $p$ is equivalent to evaluating a nice multilinear polynomial over $\mathbb{Z}_p$ (Equations \eqref{eq:5} and \eqref{eq:7}, respectively).

\subsection{$\NP$-Hardness of $\dHAMCYCLE_p$ Under Randomized Reductions}
\label{section:4.1}

Suppose that, as defined in Section \ref{section:2}, we have a directed multigraph $D_{n, p} = \left( V_n, E_{n, p} \right)$ on $n$ vertices, with an edge list $E_{n, p}$ of length $n^2$, each entry being an integer from $\mathbb{Z}_p$, which can also be seen as a string in $\mathbb{Z}_p^{n^2}$. Now, consider the polynomial $\HCY_{n, p}$ as defined in Equation \eqref{eq:5}. One can see that counting the number of unique Hamiltonian cycles modulo $p$ is precisely equivalent to evaluating the polynomial $\HCY_{n, p}$. If they both take the same input formats, an algorithm or circuit evaluating the polynomial is also an algorithm or circuit counting on graphs. We will now prove that doing either is $\NP$-hard under randomized reduction from $\dHAMCYCLE$.

\begin{lemma}
\label{lemma:7}

For any prime $p > n + 1$, there is a polynomial-time randomized reduction from an instance $D_n = \left( V_n, E_n \right)$ of $\dHAMCYCLE$ on $n$ vertices to a directed multigraph $D_{n, p} = \left( V_n, E_{n, p} \right)$ on $n$ vertices as defined above, such that:
\begin{enumerate}
\item If $D_n$ does not have a directed Hamiltonian cycle, then $D_{n, p}$ has $0$ Hamiltonian cycles modulo $p$ with a probability of $1$.
\item If $D_n$ has a directed Hamiltonian cycle, then $D_{n, p}$ has $0$ Hamiltonian cycles modulo $p$ with a probability of at most $n / (p-1)$.
\end{enumerate}

\end{lemma}

\begin{proof}

The randomized reduction proceeds as follows. For each $(i, j) \in [n]^2$ such that there is no edge from $i$ to $j$ in $D_n$, for the edge list $E_{n, p}$ of $D_{n, p}$, we set $e_{(i, j)} = 0$. Whenever there is an edge from $i$ to $j$, we select a random value $r$ from $1$ to $p-1$ and set $e_{(i, j)} = r$.
    
Notice that if $D_n$ has no Hamiltonian cycles, there are none in for $D_{n, p}$. For every $\sigma \in C_n$, there exists an $i \in [n]$ such that $e_{(i, \sigma(i))} = 0$. This makes $\HCY_{n, p} \left( E_{n, p} \right)$ uniformly $0$ since every monomial has a variable set to $0$. With a probability of $1$, the evaluation of $\HCY_{n, p} \left( E_{n, p} \right)$ is $0$.
    
On the other hand, if $D_n$ does contain a Hamiltonian cycle, then there is a $\sigma \in C_n$ such that for all $i \in [n]$, $e_{(i, \sigma(i))} \neq 0$. The value of $\HCY_{n, p} \left( E_{n, p} \right)$ is not uniformly $0$ considering the randomness. Here, we use the Schwartz-Zippel Lemma (\ref{lemma:1}), showing that 
\begin{equation*}
\mathcal{P}_{\text{non-zero entries of }E_{n, p} \leftarrow_r [p - 1]^m} \left[ \HCY_{n, p} \left( E_{n, p} \right) = 0 \right] \leq \frac{n}{p - 1},
\end{equation*}
where $m$ is the number of edges in $D_n$, completing our proof.

\end{proof}

Already, notice that if there is a polynomial-time randomized algorithm counting the number of unique Hamiltonian cycles on directed multigraphs modulo $p$ for $p > n + 1$, then immediately, $\NP = \RP$. Due to the reduction in \cite{Arora2009}, there is a reduction from $3\SAT$ on $n^\prime$ variables and $m^\prime = O \left( n^\prime \right)$ clauses to a $\dHAMCYCLE$ instance on $O\left( n^\prime + m^\prime \right) = O \left( n^\prime \right) = n$ vertices. The best algorithms to solve $\dHAMCYCLE$ take $2^{O(n)}$-time, and due to the Sparsification Lemma (\ref{lemma:2}) and Corollary \ref{corollary:1}, if $\ETH$ is true, there should be no algorithm deciding $\dHAMCYCLE$ in $2^{n^{1-\delta}} = 2^{o(n)} = 2^{o \left( n^\prime \right)}$-time for any $\delta > 0$. 

We can prove stronger results for the problem of counting unique directed Hamiltonian cycles modulo $p$. As we will see in Section \ref{section:6}, for $p > 2^{\polylog(n)}$, one can not hope to have an algorithm that is correct on even a $1/2^{\polylog(n)}$-fraction of instances in $2^{n^{1-\epsilon}}$-time for any $\epsilon > 0$. In other words, the naive algorithm that takes $O^{*}(n!) = O^{*}(2^{n\log n})$-time to try every permutation and add the number of cycles in these positions is optimal, in one sense.

\subsection{$\NP$-Hardness of $\HALF_p$ Under Randomized Reductions}
\label{section:4.2}

The proof we present in this section proceeds very similarly to the proof for counting unique directed Hamiltonian cycles (Lemma \ref{lemma:7}). We are given an undirected multigraph $U_{n, p} = \left( V_n, F_{n, p} \right)$ on $n$ vertices such that $F_{n, p} \in \mathbb{Z}_p^{\binom{n}{2}}$. Suppose we have the polynomial $HCL_{n, p} \left( F_{n, p} \right)$ as defined in Equation \eqref{eq:7}. Computing $HCL_{n, p} \left( F_{n, p} \right)$ does not only evaluate this polynomial but also counts the number of unique half-cliques modulo $p$. Similar to Lemma \ref{lemma:7}, we state and prove that $\HALF_p$ is $\NP$-hard under randomized reductions due to a reduction from the $\NP$-complete language $\HALF$.

\begin{lemma}
\label{lemma:8}

For any prime $p > \displaystyle \binom{n}{2} + 1$, there is a polynomial-time randomized reduction from an instance $U_n = \left( V_n, F_n \right)$ of $\HALF$ on $n$ vertices to an instance $U_{n, p} = \left( V_n, F_{n, p} \right)$ of $\HALF_p$ on $n$ vertices as defined above, such that:

\begin{enumerate}

\item If $U_n$ does not have a clique of size $\lfloor n/2 \rfloor$, then $U_{n, p}$ has $0$ unique $\lfloor n/2 \rfloor$-sized cliques modulo $p$ with a probability of $1$.

\item If $U_n$ has a clique of size $\lfloor n/2 \rfloor$, then $U_{n, p}$ has $0$ unique $\lfloor n/2 \rfloor$-sized cliques modulo $p$ with a probability of at most $\displaystyle \binom{n}{2} / (p-1)$.

\end{enumerate}

\end{lemma}

\begin{proof}

We construct the edge list $F_{n, p}$ of $U_{n, p}$ as follows. For every $(i, j)$, with $1 \leq i < j \leq n$ and $\{\, i, j \,\} \in F_n$, we choose a random non-zero value from $\mathbb{Z}_p$ for the entry $e_{\{\, i, j \,\}}$ in $F_{n, p}$. If $\{\, i, j \,\} \not\in F_{n, p}$, then we assign the value $0$ to $e_{\{\, i, j \,\}}$.

Notice that if $U_n$ does not have a clique of size $\lfloor n/2 \rfloor$, then for every set $S \in \displaystyle \binom{[n]}{\lfloor n/2 \rfloor}$, we have $\prod_{(i, j) \in S^2, 1 \leq i < j}^n e_{\{\, i, j \,\}} = 0$, resulting in $\HCL_{n, p} \left( F_{n, p} \right)$ being uniformly zero. However, if $U_n$ does contain a half-clique, then there is at least one set $S \in \displaystyle \binom{[n]}{\lfloor n/2 \rfloor}$ such that $\prod_{(i, j) \in S^2, 1 \leq i < j} e_{\{\, i, j \,\}} \neq 0$. Since the polynomial $\HCL_{n, p} \left( F_{n, p} \right)$ is not uniformly $0$, due to the Schwartz-Zippel Lemma (\ref{lemma:1}), we have
\begin{equation*}
\mathcal{P}_{\text{non-zero entries of } F_{n, p} \leftarrow_r (p - 1)^m} \left[ \HCL_{n, p} \left( F_{n, p} \right) = 0 \right] \leq \frac{\displaystyle \binom{n}{2}}{p-1},
\end{equation*}
where $m$ is the number of edges in $U_n$.

\end{proof}

We include a proof in Appendix \ref{appendix:A} showing that $\HALF$ is $\NP$-complete, with a reduction from $\CLIQUE$. There is a reduction from $3\SAT$ to $\CLIQUE$, taking a $3\SAT$ instance with $n$ variables and $m$ clauses and giving us a $\CLIQUE$ instance with $O(n+m)$ vertices. We also have a reduction from a $\CLIQUE$ instance with $n$ vertices to a $\HALF$ instance with $O(n)$ vertices. Due to the reduction from $3\SAT$ to $\HALF$, giving us a graph of size $O(n)$ from a $3\SAT$ instance with $n$ variables and $O(n)$ clauses, along with our discussion on the $\ETH$-hardness of $3\SAT$ on $O(n)$ clauses (Section \ref{section:2.7}), there should be no $2^{o(n)}$-time algorithm for $\HALF$ and, subsequently, for $\HALF_p$. However, note that there is a naive $O^* \left( 2^{n} \right)$-time\footnote{Interestingly, ${n \choose \lfloor n/2 \rfloor} = \Theta \left( 2^n / \sqrt{n} \right)$, due to Stirling's approximation, but this does not give us any real savings in analysis.} algorithm for counting unique half-cliques modulo $p$. We show in Section \ref{section:7} that for sufficiently large $p$, growing as $\Theta \left( 2^{\polylog(n)} \right)$, any $2^{o(n)}$-time algorithm can be correct on a $\left( 1 / 2^{\polylog(n)} \right)$-fraction of instances unless $\rETH$ is not true.

\section{Warmup: Hardness Against Circuits}
\label{section:5}

Before moving on to the main results, we will demonstrate a simple technique to show rare-case hardness against polynomial-sized circuit families \citep{Tejas2024, Cai1999} for these problems and many others. We stress that this result works for these problems and many others. Suppose one has a polynomial $h$ which is ``$\NP$-hard'' to evaluate due to a randomized polynomial-time reduction from $\NP$, similar to those in Lemmas \ref{lemma:7} and \ref{lemma:8}. In that case, this result applies to $h$ as well. We now state the theorem and prove it.

\begin{theorem}
\label{theorem:5}

\textbf{Rare-Case Hardness of $\HCY_{n, p}$ and $\HCL_{n, p}$ Against Polynomial-Sized Circuit Families.} \\
If $\NP \not\subset \PPOLY$, then for every $\alpha > 0$, for every prime $p$ such that $\Omega(2^{\polylog(n)}) \leq p \leq O(2^{\poly(n)})$, for any circuit family $\{\, C_n \,\}_{n \in \mathbb{N}}$\footnote{$C_n$ takes the input $E_{n, p}$ of size $\left| E_{n, p} \right| = n^2 \log p$, where $n$ is the number of vertices in the multigraph.} that attempts to count the number of unique Hamiltonian cycles in directed multigraphs over $\mathbb{Z}_p$, given $E_{n, p}$ (Equation \eqref{eq:4}), we have
\begin{equation*}
\mathcal{P}_{E_{n, p} \leftarrow_r \mathbb{Z}_p^{n^2}} \left[ C_n \left( E_{n, p} \right) = \HCY_{n, p} \left( E_{n, p} \right) \right] < \frac{1}{n^{\alpha}},
\end{equation*}
and similarly, for any circuit family $\left \{\, C^\prime_n \, \right \}_{n \in \mathbb{N}}$\footnote{$C^\prime_n$ takes the input $F_{n, p}$ of size $\left| F_{n, p} \right| = \binom{n}{2} \log p$, where $n$ is the number of vertices in the multigraph.} that attempts to count the number of unique half-cliques on undirected multigraphs over $\mathbb{Z}_p$, given $F_{n, p}$ (Equation \eqref{eq:6}), we have
\begin{equation*}
\mathcal{P}_{F_{n, p} \leftarrow_r \mathbb{Z}_p^{n \choose 2}} \left[ C^\prime_n(F_{n, p}) = \HCL_{n, p} \left( F_{n, p} \right) \right] < \frac{1}{n^{\alpha}},
\end{equation*}
for infinitely many $n$, depending on $\alpha$ and the size-bound of the circuit family.

\end{theorem}

\begin{proof}

We will prove the theorem for $\HCY_{n, p}$, the polynomial counting the number of unique Hamiltonian cycles on directed multigraphs over $\mathbb{Z}_p$. Nevertheless, without loss of generality, this extends perfectly to $\HCL_{n, p}$ as well.

Suppose for some $\alpha > 0$ and $\beta > 0$, we have a circuit family $\{\, C_n \,\}_{n \in \mathbb{N}}$, where the circuit $C_n$ corresponds to the counting problem on $\HCY_{n, p}$ for multigraphs with $n$ vertices. Also, assume that $C_n$ has at most $n^{\beta}$ gates and for all but finitely many $n$'s, that is, for sufficiently large $n$, we have
\begin{equation*}
\mathcal{P}_{E_{n, p} \leftarrow_r \mathbb{Z}_p^{n^2}} \left[ C_n \left( E_{n, p} \right) = \HCY_{n, p} \left( E_{n, p} \right) \right] \geq \frac{1}{n^{\alpha}}.
\end{equation*}

Suppose that we provide the circuit as advice to a Turing machine under the non-uniform model of computation (an equivalent definition of $\PPOLY$ \citep{Karp1982}). Using the STV list decoder\footnote{Recall that given a $1/\epsilon$-fraction of correctness and a target degree $d$ on $m$ variables, the STV list decoder takes $\poly(m, d, \log p, 1/\epsilon)$-time.} (Lemma \ref{lemma:2}), in $\poly(n)$-time, we can get $k = O \left( n^{\alpha} \right)$ machines $(M_i)_{i = 1}^k$, each computing a polynomial, one of which is of interest to us. Suppose that $j \in [k]$ corresponds to $M_i$ such that $M_j$ computes $\HCY_{n, p}$ with high probability. In fact, due to the theorem of \cite{Adleman1978}, we can fix the random string to be always correct. Now, we add some more advice to help us find $M_j$. We have a string
\begin{equation*}
\left( \left( E_i, \HCY_{n, p} \left( E_i \right) \right)_{i = 1}^{j - 1}, \left( E_i, \HCY_{n, p} \left( E_i \right) \right)_{i = j + 1}^k \right),
\end{equation*}
such that for every $i \neq j$, $M_{i} \left( E_i \right) \neq \HCY_{n, p} \left( E_i \right)$, but $M_j \left( E_i \right) = \HCY_{n, p} \left( E_i \right)$ for all $i$, and for all inputs in general. This string helps us identify $j$, and we can now use this machine for all our computations.

Our reductions in Lemmas \ref{lemma:7} and \ref{lemma:8} extend to non-uniform computation due to ideas similar to the theorem of \cite{Adleman1978}. We can fix a random string so that the reduction is sound and complete on all inputs. Hence, such circuit families would imply that $\NP \subset \PPOLY$.

The proof extends perfectly to $\HCL_{n, p}$, and any other polynomials one can show ``$\NP$-hardness'' similar to the lemmas in Section \ref{section:4}.

\end{proof}

\section{Hamiltonian Cycle Counting is Hard for Algorithms}
\label{section:6}

In this section, first, we will show that if we have an oracle machine $M$ that computes a polynomial $H_{n, p}: \mathbb{Z}_p^{n^2} \to \mathbb{Z}_p$ that is unknown to us, then we have a fast test that can tell us with high probability whether $H_{n, p}$ is $\HCY_{n, p}$. Following that, in Subsection \ref{section:6.2}, we use it to the STV list decoder (Lemma \ref{lemma:2}) to produce $O(1 / \epsilon)$ machines $M_i$ computing polynomials that are $O(\epsilon)$-close to our oracle machine $M$. These machines have exponentially small probability of error. We will assume in Subsection \ref{section:6.1} that $M$ is always correct for logical simplicity. Due to the $2^{\polylog(n)} / {2^{\Omega(n)}}$ union bound on probability, even with the error from the machine that the STV list decoder gives us, with a probability of $1-2^{-\Omega(n)}$, it behaves identically to a deterministic machine that makes no mistakes.

\subsection{Generalized Permanents, Isomorphic Polynomials, and Cycle Covers}
\label{section:6.1}

Before diving into the specifics of our test or even describing the strategy, we will try to generalize a function that looks very similar to our function, $\HCY_{n, p}$. The permanent, $\PER_{n, p}$, on $E_{n, p} = \left( \left( e_{(i, j)} \right)_{j \in [n]} \right)_{i \in [n]} \in \mathbb{Z}_p^{n^{2}}$ over $\mathbb{Z}_p$ is defined as
\begin{equation*}
\PER_{n, p} \left( E_{n, p} \right) = \sum_{\sigma \in S_n} \prod_{i \in [n]} e_{(i, \sigma(i))}.
\end{equation*}
We can see that there is an uncanny similarity to our function, $\HCY_{n, p}$:
\begin{equation*}
\HCY_{n, p} \left( E_{n, p} \right) = \sum_{\sigma \in C_n} \prod_{i \in [n]} e_{(i, \sigma(i))}.
\end{equation*}

\cite{Valiant1979} proved that if one could evaluate the permanent on $0/1$ matrices in polynomial-time, then $\P^{\SHARPP} = \P$. More specifically, he proved that $\P^{\SHARPP} = \P^{\PER}$. The permanent on $0/1$ matrices count the number of cycle covers in a graph $G$ modulo $p$, which Valiant used to prove his theorem. We will now generalize the permanent in the following way:

\begin{definition}
\label{def:10}

\textit{\textbf{The Generalized Permanent Function.} \\
Suppose we have a list $A_{n, p} \in \mathbb{Z}_p^{n!}$ such that for each $\sigma \in S_n$, we have an entry $a_{\sigma} \in \mathbb{Z}_p$ in $A_{n, p}$, then the generalized permanent, $\GPR_{n, p}: \mathbb{Z}_p^{n!} \times \mathbb{Z}_p^{n^2} \to \mathbb{Z}_p$, on $E_{n, p}$ with respect to $A_{n, p}$ is defined as}
\begin{equation*}
\GPR_{n, p} \left( A_{n, p}, E_{n, p} \right) = \sum_{\sigma \in S_n} a_{\sigma} \prod_{i \in [n]} e_{(i, \sigma(i))}.
\end{equation*}

\end{definition}

One can see that
\begin{equation*}
\PER_{n, p} \left( E_{n, p} \right) = \GPR_{n, p} \left( \left( 1 \right)_{i \in [n!]}, E_{n, p} \right).
\end{equation*}

The generalized permanent offers us a way to change the ``weights'' of cycle covers. Notice that the set of generalized permanent polynomials form a vector space of dimension $n!$ over $\mathbb{Z}_p$, which we will denote by $\mathbb{V}_{\GPR_p}$:

\begin{definition}
\label{def:11}

\textit{\textbf{The Generalized Permanent Polynomial Vector Space.} \\
We define the vector space $\mathbb{V}_{\GPR_{n, p}}$ over $\mathbb{Z}_p$, with the usual operations of polynomial addition and scalar multiplication as}
\begin{equation*}
\mathbb{V}_{\GPR_{n, p}} = \left \{\, \GPR_{A_{n, p}} \mid A_{n, p} \in \mathbb{Z}_p^{n!}, \GPR_{A_{n, p}} \left( E_{n, p} \right) = \GPR_{n, p} \left( A_{n, p}, E_{n, p} \right) \, \right\}.
\end{equation*}

\end{definition}

\begin{definition}
\label{def:12}

\textit{\textbf{Permutation of a Directed Multigraph.} \\
Suppose we have the formal list
\begin{equation*}
E_{n, p} = \left( \left( e_{(i, j)} \right)_{j \in [n]} \right)_{i \in [n]}.
\end{equation*}
We define the formal list}
\begin{equation}
\label{eq:8}
\pi \left( E_{n, p} \right) = \left( \left( e_{(\pi(i), \pi(j))} \right)_{j \in [n]} \right)_{i \in [n]}, \quad \forall \pi \in S_n.
\end{equation}

\end{definition}
This is like an isomorphic graph that differs from our original graph by the permutation $\pi^{-1}$, and permuting the edge list accordingly.

\begin{definition}
\label{def:13}

\textit{\textbf{The Subgroup Induced by $H_{n, p}$.} \\
We define the set of permutations, $\Pi \left( H_{n, p} \right)$, as
\begin{equation}
\label{eq:9}
\Pi \left( H_{n, p} \right) = \{\, \pi \in S_n | H_{n, p} \left( \pi \left( E_{n, p} \right) \right) = H_{n, p} \left( E_{n, p} \right) \,\}, \quad \forall H_{n, p} \in \mathbb{V}_{\GPR_{n, p}},
\end{equation}
in the context of both being the same formal polynomials. It is easy to see from the group axioms that $\Pi \left( H_{n, p} \right)$ is always a group and, more particularly, a subgroup of $S_n$.}

\end{definition}

We can define a subspace of $\mathbb{V}_{\GPR_p}$, which we denote by $\mathbb{I}_{\GPR_p}$, of ``very nice'' polynomials. 
\begin{definition}
\label{def:14}

\textit{\textbf{The Subspace of Isomorphic Polynomials.} \\
Using Equation \eqref{eq:9}, we define the set of functions, $\mathbb{I}_{\GPR_{n, p}}$, as
\begin{equation*}
\mathbb{I}_{\GPR_{n, p}} = \left \{\, H_{n, p} \mid H_{n, p} \in \mathbb{V}_{\GPR_{n, p}}, \Pi \left( H_{n, p} \right) = S_n \, \right \}.
\end{equation*}
One can see that $\mathbb{I}_{\GPR_p}$ is a subspace of $\mathbb{V}_{\GPR_p}$.}

\end{definition}

The motivation for this is as follows. We can see $H_{n, p} \in \mathbb{I}_{\GPR_{n, p}}$ as a function that counts weighted unique cycle covers. $\Pi \left( H_{n, p} \right)$ is the set of ways we can take an isomorphic graph and get the same weight as the original graph for any edge list. If $H_{n, p} \in \mathbb{I}_{\GPR_{n, p}}$, the function $H_{n, p}$ weighs all isomorphic graphs equally. $\HCY_{n, p}$ is one of these functions. 

From Equation \eqref{eq:5}, we notice that the function $\HCY_{n, p}$ sums over the conjugacy class $C_n$. Now we generalize the function $\HCY_{n, p}$ into a function $H_{n, p}^{C(\gamma)}$ that sums over the conjugacy class of $\gamma$.

\begin{definition}
\label{def:15}

\textit{\textbf{The Function $H_{n, p}^{C(\gamma)}$.} \\
For all $\gamma \in S_n$, we define the function $H_{n, p}^{C(\gamma)}$ as}
\begin{equation}
\label{eq:10}
H_{n, p}^{C(\gamma)} \left( E_{n, p} \right) = \sum_{\sigma \in C(\gamma)} \prod_{i \in [n]} e_{(i, \sigma(i))}.
\end{equation}

\end{definition}

We will also benefit greatly from the following characterization of $\mathbb{I}_{\GPR_{n, p}}$.
\begin{lemma}
\label{lemma:9}

The set of functions
\begin{equation*}
\left \{\, H_{n, p}^{C(\gamma)} \mid \gamma \in S_n \, \right \}
\end{equation*}
forms a basis for $\mathbb{I}_{\GPR_{n, p}}$. 

\end{lemma}

\begin{proof}

First, we will prove that $H_{n, p}^{C(\gamma)} \in \mathbb{I}_{\GPR_{n, p}}$. For all $(\gamma, \pi) \in S_n^2$, from Equations \eqref{eq:8} and \eqref{eq:10}, we have
\begin{equation}
\label{eq:11}
\begin{split}
H_{n, p}^{C(\gamma)} \left( \pi \left( E_{n, p} \right) \right) & = H_{n, p}^{C(\gamma)} \left( \left( \left( e_{(\pi(i), \pi(j))} \right)_{j \in [n]} \right)_{i \in [n]} \right) \\
& = \sum_{\sigma \in C(\gamma)} \prod_{i \in [n]} e_{(\pi(i), \pi( \sigma(i)))} \\
& = \sum_{\sigma \in C(\gamma)} \prod_{j \in [n]} e_{\left( j, \pi \left( \sigma \left( \pi^{-1}(j)) \right) \right) \right)} \\
& = \sum_{\sigma \in C(\gamma)} \prod_{j \in [n]} e_{\left( j, \pi^{-1} \sigma \pi(j) \right)} \\
& = \sum_{\pi \sigma^\prime \pi^{-1} \in C(\gamma)} \prod_{j \in [n]} e_{\left( j, \sigma^\prime (j) \right)} \\
& = \sum_{\sigma^\prime \in \pi^{-1} C(\gamma) \pi} \prod_{j \in [n]} e_{\left( j, \sigma^\prime (j) \right)} \\
& = \sum_{\sigma^\prime \in C(\gamma)} \prod_{j \in [n]} e_{\left( j, \sigma^\prime (j) \right)} \\
& = H_{n, p}^{C(\gamma)} \left( E_{n, p} \right) \\
& \implies \Pi \left( H_{n, p}^{C(\gamma)} \right) = S_n \\
& \implies H_{n, p}^{C(\gamma)} \in \mathbb{I}_{\GPR_{n, p}},
\end{split}
\end{equation}
where, we have used the change of variables $j = \pi(i)$ and $\sigma^\prime = \pi^{-1} \sigma \pi$, and from Section \ref{section:2.4}, since $C(\gamma)$ is a conjugacy class, we have $\pi^{-1} C(\gamma) \pi = C(\gamma)$. As a special case of Equation \eqref{eq:11}, by taking $\gamma$ as a cycle of length $n$, we get the result that $\HCY_{n, p} \in \mathbb{I}_{\GPR_{n, p}}$.

From Equation \eqref{eq:11}, we have $H_{n, p}^{C(\gamma)} \in \mathbb{I}_{\GPR_{n, p}}$. From Lemma \ref{lemma:3}, $H_{n, p}^{C(\gamma)}$ counts the number of unique subgraphs isomorphic to $D_\gamma$. It is easy to see that when $\gamma$ is not conjugate to $\delta \in S_n$, $H_{n, p}^{C(\gamma)}$ and $H_{n, p}^{C(\delta)}$ are linearly independent. Our strategy now is to show that if $H_{n, p} \left( E_{n, p} \right) = \GPR_{n, p} \left( A_{n, p}, E_{n, p} \right) \in \mathbb{I}_{\GPR_{n, p}}$, and we have a $\gamma$ such that $(\sigma_1, \sigma_2) \in C(\gamma)^2$ and $a_{\sigma_1} \neq a_{\sigma_2}$, we arrive at a contradiction. Let 
\begin{equation}
\label{eq:12}
\sigma_2 = \pi^{-1} \sigma_1 \pi.
\end{equation}
We have 
\begin{equation}
\label{eq:13}
\begin{split}
H_{n, p} \left( E_{n, p} \right) = \sum_{\sigma \in S_n} a_\sigma \prod_{i \in [n]} e_{(i, \sigma(i))}.
\end{split}
\end{equation}
Since $H_{n, p} \in \mathbb{I}_{\GPR_{n, p}}$, we also have
\begin{equation}
\begin{split}
\label{eq:14}
H_{n, p} \left( E_{n, p} \right) & = H_{n, p} \left( \pi \left( E_{n, p} \right) \right) \\  
& = \sum_{\sigma \in S_n} a_\sigma \prod_{i \in [n]} e_{(\pi(i), \pi(\sigma(i)))} \\
& = \sum_{\sigma \in S_n} a_\sigma \prod_{j \in [n]} e_{\left( j, \pi \left( \sigma \left( \pi^{-1}(j) \right) \right) \right)} \\
& = \sum_{\sigma \in S_n} a_\sigma \prod_{j \in [n]} e_{\left( j, \pi^{-1} \sigma \pi(j) \right)},
\end{split}
\end{equation}
where, $j = \pi(i)$. From Equation \eqref{eq:12}, the monomial corresponding to $a_{\sigma_1}$ in Equation \eqref{eq:14} is
\begin{equation*}
\prod_{j \in [n]} e_{\left( j, \pi^{-1} \sigma_1 \pi(j) \right)} = \prod_{j \in [n]} e_{(j, \sigma_2 \pi(j))},  
\end{equation*}
which is the same as the monomial corresponding to $a_{\sigma_2}$ in Equation \eqref{eq:13}, implying that
\begin{equation*}
a_{\sigma_1} = a_{\sigma_2}.
\end{equation*}

\end{proof}

More informally, our function $H_{n, p}$ should equally weigh all cycle covers of the same conjugacy class. The weight assigned to the induced graph $D_{\sigma}$ should be balanced for all $\sigma$ in the same conjugacy class.

\subsection{Fast Test to Identify Polynomials in $\mathbb{V}_{\GPR_{n, p}}$}
\label{section:6.2}

We are now ready to obtain a rough outline of our tests to see if our oracle machine $M$ computes a function $H_{n, p} \in \mathbb{V}_{\GPR_{n, p}}$.

\begin{enumerate}

\item We first show that if $M$ computes a polynomial $H_{n, p}$ of degree at most $n$ over our domain and range, then we have a probabilistic test to check if $H_{n, p} \in \mathbb{V}_{\GPR_{n, p}}$.

\item If we know that $H_{n, p} \in \mathbb{V}_{\GPR_{n, p}}$, then we have a probabilistic test to check if $H_{n, p} \in \mathbb{I}_{\GPR_{n, p}}$.

\item If we know that $H_{n, p} \in \mathbb{I}_{\GPR_{n, p}}$, then we have a probabilistic test to check if $M$ computes $H_{n, p}$.

\end{enumerate}

We outline our first test in Algorithm \ref{alg:1}.

\begin{algorithm}
\caption{$\text{GenPerm}\left( M, n, p \right)$}
\label{alg:1}
\Comment{The Generalized Permanent Test \hspace{10.1cm}} \\
\Comment{Input: $n$ is the number of vertices in the directed multigraph} \\
\Comment{Input: $p > n$ is a prime} \\
\Comment{Input: The oracle machine $M$ computes a polynomial $H_{n, p}: \mathbb{Z}_p^{n^2} \to \mathbb{Z}_p$} \\
\Comment{Output: \texttt{ACCEPT} if $H_{n, p} \in \mathbb{V}_{\GPR_{n, p}}$, \texttt{REJECT} otherwise} \\
\Comment{All computations are done over $\mathbb{Z}_p$}
\begin{algorithmic}
\For{each $i \in [n]$}
    \State $E_{n, p} \gets (0)_{l = 1}^{n^2}$
    \For{each $j \neq i \in [n]$} 
    \Comment{The $i$'th row of $E_{n, p}$ is $0$}
        \For{each $k \in [n]$}
            \State $e_{(j, k)} \gets_{r} \mathbb{Z}_p$
        \EndFor
    \EndFor
    \State $g \gets M \left( E_{n, p} \right)$
    \If{$g \neq 0$}
        \State \Return \texttt{REJECT}
    \EndIf
    \State $E_{n, p} \gets (0)_{l = 1}^{n^2}$
    \For{each $j \neq i \in [n]$}
    \Comment{The $i$'th column of $E_{n, p}$ is $0$}
        \For{each $k \in [n]$}
            \State $e_{(k, j)} \gets_{r} \mathbb{Z}_p$
        \EndFor
    \EndFor
    \State $g \gets M \left( E_{n, p} \right)$
    \If{$g \neq 0$}
        \State \Return \texttt{REJECT}
    \EndIf
\EndFor
\State \Return \texttt{ACCEPT}
\end{algorithmic}
\end{algorithm}

We now prove that this test works with very high probability.

\begin{lemma}
\label{lemma:10}

Provided that $M$ computes a polynomial $H_{n, p}: \mathbb{Z}^{n^2}_{p} \to \mathbb{Z}_p$ and $p > n$,

\begin{enumerate}

\item If $H_{n, p} \in \mathbb{V}_{\GPR_{n, p}}$, then $\text{GenPerm}(M, n, p)$ returns \texttt{ACCEPT} with a probability of $1$.

\item If $H_{n, p} \not \in \mathbb{V}_{\GPR_{n, p}}$, then $\text{GenPerm}(M, n, p)$ returns \texttt{ACCEPT} with probability a probability of at most $n / p$.

\end{enumerate}

\end{lemma}

\begin{proof}

To prove the first assertion, assume that the polynomial $H_{n, p}$ computed by $M$ is indeed in $\mathbb{V}_{\GPR_{n, p}}$. Since for every $i$, every monomial contains a variable of the form $e_{(i, j)}$, which remains $0$ throughout the test, $M$ always returns $0$ on these inputs since every monomial evaluates to $0$. Similar reasoning holds when $e_{(j, i)}$ is set to $0$ for all $j$. The test always \texttt{ACCEPT}s when $H_{n, p} \in \mathbb{V}_{\GPR_{n, p}}$.
  
Assume that $H_{n, p} \left( E_{n, p} \right)$ is uniformly $0$ in all $E_{n, p}$ constructed throughout this test. If for each $i \in [n]$, each monomial contains a variable of the form $e_{(i, j)}$ for some $j \in [n]$, and similarly for each $i \in [n]$, each monomial contains a variable of the form $e_{(j, i)}$ for some $j \in [n]$, then, since the degree of $H_{n, p}$ is at most $n$, due to the pigeonhole principle, every monomial must be of the form $a \prod_{i = 1}^n e_{(i, \sigma(i))}$ for some $\sigma \in S_n$ and some $a \in \mathbb{Z}_p$, implying that $H_{n, p} \in \mathbb{V}_{\GPR_{n, p}}$.
  
The contrapositive of the previous paragraph is that if $H_{n, p} \not \in \mathbb{V}_{\GPR_{n, p}}$, then on at least one of the edge lists, $E_{n, p}$, that we constructed, $H_{n, p} \left( E_{n, p} \right)$ is not uniformly $0$. Due to the Schwartz-Zippel Lemma (\ref{lemma:2}),
\begin{equation*}
\mathcal{P}_{E_{n, p} \leftarrow_r \mathbb{Z}_p^{n^2}} \left[ H_{n, p} \left( E_{n, p} \right) = 0 \right] \leq \frac{n}{p},
\end{equation*}
causing the test to return \texttt{REJECT} with a probability of at least $1 - n / p$ on at least one of the steps.
\end{proof}

Moreover, one might also note that this test takes $O \left( n^3 \right)$-time and $O(n)$-queries to $M$. We can get exponentially small error probability from the test by repetition, as is usual for probabilistic algorithms.

\subsection{Fast Test to Identify Polynomials in $\mathbb{I}_{\GPR_{n, p}}$}
\label{section:6.3}

We now present a straightforward test, Algorithm \ref{alg:2}, to check if the polynomial $H_{n, p} \in \mathbb{V}_{\GPR_{n, p}}$ computed by $M$ is also in $\mathbb{I}_{\GPR_{n, p}}$ by making just two queries to $M$. By simply testing this on a random graph, and on a random graph isomorphic to it, we can check if $M$ weighs isomorphic graphs equally, with high probability, of course.

\begin{algorithm}
\caption{$\text{IsoPerm} \left( M, n, p \right)$}
\label{alg:2}
\Comment{The Isomorphic Polynomial Test \hspace{10.2cm}} \\
\Comment{Input: $n$ is the number of vertices in the directed multigraph} \\
\Comment{Input: $p > n$ is a prime} \\
\Comment{Input: The oracle machine $M$ computes a polynomial $H_{n, p}:\mathbb{Z}_p^{n^2} \to \mathbb{Z}_p \in \mathbb{V}_{\GPR_{n, p}}$} \\
\Comment{Output: \texttt{ACCEPT} if $H_{n, p} \in \mathbb{I}_{\GPR_{n, p}}$, \texttt{REJECT} otherwise} \\
\Comment{All computations are done over $\mathbb{Z}_p$}
\begin{algorithmic}
\State $E_{n, p} \gets_r \mathbb{Z}_p^{n^2}$
\State $\pi \gets_r S_n$
\State $E^\prime_{n, p} \gets \pi \left( E_{n, p} \right)$ 
\Comment{Using Equation \eqref{eq:8}}
\State $g \gets M \left( E_{n, p} \right)$
\State $g^\prime \gets M \left( E^\prime_{n, p} \right)$
\If{$g \neq g^\prime$}
    \State \Return \texttt{REJECT}
\EndIf
\State \Return \texttt{ACCEPT}
\end{algorithmic}
\end{algorithm}

We now show that this test is correct.
\begin{lemma}
\label{lemma:11}

Given that the polynomial $H_{n, p}$ computed by $M$ is in $\mathbb{V}_{\GPR_{n, p}}$ and $p > n$,

\begin{enumerate}

\item If $H_{n, p} \in \mathbb{I}_{\GPR_{n, p}}$, then $\text{IsoPerm} \left( M, n, p \right)$ returns \texttt{ACCEPT} with a probability of $1$.

\item If $H_{n, p} \not \in \mathbb{I}_{\GPR_{n, p}}$, then $\text{IsoPerm} \left( M, n, p \right)$ returns \texttt{ACCEPT} with a probability of at most $1 / 2 + n / p$.

\end{enumerate}

\end{lemma}

\begin{proof}

The first assertion is true by Definitions \ref{def:13} and \ref{def:14}. When $H_{n, p} \not \in \mathbb{I}_{\GPR_{n, p}}$, that is, $\Pi \left( H_{n, p} \right) \neq S_n$, from Equation \eqref{eq:9} and using Lagrange's theorem \citep{Herstein1975}, given that $\left| \Pi \left( H_{n, p} \right) \right|$ divides $|S_n|$, we have
\begin{equation}
\label{eq:15}
\begin{split}
P_1 & = \mathcal{P}_{\substack{\pi \leftarrow_r S_n \\ E_{n, p} \leftarrow_r \mathbb{Z}_p^{n^2}}} \left[ H_{n, p} \left( E_{n, p} \right) = H_{n, p} \left( \pi \left( E_{n, p} \right) \right) \mid \text{these are formally the same polynomials} \right] \\
& = \frac{|\Pi \left( H_{n, p} \right)|}{|S_n|} \\
& \leq \frac{1}{2}.
\end{split}
\end{equation}

Using the Schwartz-Zippel Lemma (\ref{lemma:1}), we also have
\begin{equation}
\label{eq:16}
\begin{split}
P_2 & = \mathcal{P}_{\substack{\pi \leftarrow_r S_n \\ E_{n, p} \leftarrow_r \mathbb{Z}_p^{n^2}}} \left[ H_{n, p} \left( E_{n, p} \right) = H_{n, p} \left( \pi \left( E_{n, p} \right) \right) \mid \text{ these are not formally the same polynomials} \right] \\
& \leq \frac{n}{p}.
\end{split}
\end{equation}

From Equations \eqref{eq:15} and \eqref{eq:16}, we get
\begin{equation*}
\begin{split}
\mathcal{P}_{\pi \leftarrow_r S_n, E_{n, p} \leftarrow_r \mathbb{Z}_p^{n^2}} \left[ H_{n, p} \left( E_{n, p} \right) = H_{n, p} \left( \pi \left( E_{n, p} \right) \right) \right] & = P_1 + P_2 \\
& \leq \frac{1}{2} + \frac{n}{p}
\end{split}
\end{equation*}
    
Using this, we have the probability that both evaluations are the same is at most $1 / 2 + n / p$.

\end{proof}

Once again, by repetition, we can make the error probability of this test exponentially small.

\subsection{Fast Test to Identify the $\HCY_{n, p}$ Polynomials}
\label{section:6.4}

We now present the last puzzle in determining whether $M$ computes the $\HCY_{n, p}$ polynomial. The observation is that $C_n$, as a conjugacy class, contains only cyclic permutations with cycles of length $n$ and that no other conjugacy class contains a cycle of length $n$.

We now know that the polynomial $H_{n, p}$ that $M$ computes is in $\mathbb{I}_{\GPR_{n, p}}$. We need to show that $a_{\sigma}$ in the $\GPR_{n, p}$ representation of $H_{n, p}$ is $0$ if $\sigma$ has a cycle of length less than $n$. We do not even need to test till $n$. It is sufficient for us to test for cycles of size up to $\lfloor n / 2 \rfloor$. Each conjugacy class corresponds to an integer partition of $n$: Ways of representing $n$ as a sum of positive numbers, uniquely up to changing the order. $C_n$ is the conjugacy class that corresponding the integer partition $n = n$. All other conjugacy classes have at least two terms, at least one of which, due to the pigeonhole principle, can be at most $\lfloor n / 2 \rfloor$. We present the test in Algorithm \ref{alg:3} that returns \texttt{ACCEPT} if $a_{\sigma} = 0$ for all $\sigma \in S_n$ containing a cycle of length $k$, and \texttt{REJECT} with high probability if not.

\begin{algorithm}
\caption{$\text{NoCycle} \left( M, n, p, k \right)$}
\label{alg:3}
\Comment{The No Cycle of Length $k$ Test \hspace{10.5cm}} \\
\Comment{Input: $n$ is the number of vertices in the directed multigraph} \\
\Comment{Input: $p > n$ is a prime} \\
\Comment{Input: $k \in [n] - [1]$} \\
\Comment{Input: The oracle machine $M$ computes a polynomial $H_{n, p}: \mathbb{Z}_p^{n^2} \to \mathbb{Z}_p \in \mathbb{I}_{\GPR_{n, p}}$} \\
\Comment{Output: \texttt{ACCEPT} if $a_{\sigma} = 0$ for all $\sigma \in S_n$ with cycles of length $k$, \texttt{REJECT} otherwise} \\
\Comment{All computations are done over $\mathbb{Z}_p$}
\begin{algorithmic}
\State $E_{n, p} \gets (0)_{l = 1}^{n^2}$
\For{each $1 \leq i < k$}
\Comment{The vertices in $[k]$ make a cycle}
    \State $e_{(i, i+1)} \gets 1$
\EndFor
\State $e_{(k, 1)} \gets 1$
\For{each $i > k$ and $j > k$}
\Comment{All the edges in $[n] - [k]$ are random in number}
    \State $e_{(i, j)} \gets_r \mathbb{Z}_p$
\EndFor
\State $g \gets M \left( E_{n, p} \right)$
\If{$g \neq 0$}
    \State \Return \texttt{REJECT}
\EndIf 
\State \Return \texttt{ACCEPT}
\end{algorithmic}
\end{algorithm}

Let us now prove that this test works.
\begin{lemma}
\label{lemma:12}

If $H_{n, p} \in \mathbb{I}_{\GPR_{n, p}}$, then we have the following.

\begin{enumerate}

\item If $a_{\sigma} = 0$ for all $\sigma \in S_n$ with a cycle of length $k$, then $\text{NoCycle} \left( M, n, p, k \right)$ returns \texttt{ACCEPT} with a probability of $1$.

\item If there is a $\sigma \in S_n$ with a cycle of length $k$ such that $a_{\sigma} \neq 0$, then $\text{NoCycle} \left( M, n, p, k \right)$ returns \texttt{ACCEPT} with a probability of at most $n / p$.

\end{enumerate}

\end{lemma}

\begin{proof}

Using Lemma \ref{lemma:9}, taking only one $\gamma$ per conjugacy class, we can write $H_{n, p} \left( E_{n, p} \right)$ as
\begin{equation*}
\begin{split}
H_{n, p} \left( E_{n, p} \right) & =  \GPR_{n, p} \left( A_{n, p}, E_{n, p} \right) \\
& = \sum_{C(\gamma) \text{ is a conjugacy class}} a_{\gamma} \sum_{\sigma \in C(\gamma)} \prod_{i \in [n]}e_{(i, \sigma(i))} \\
& = \sum_{C(\gamma) \text{ is a conjugacy class}} a_{\gamma} \sum_{\sigma \in C(\gamma)} \prod_{i \in [k]} e_{(i, \sigma(i))} \prod_{i \in [n] - [k]} e_{(i, \sigma(i))} \\
& = \sum_{C(\gamma) \text{ is a conjugacy class with no cycles of length } k} a_{\gamma} \sum_{\sigma \in C(\gamma)} \prod_{i \in [k]} e_{(i, \sigma(i))} \prod_{i \in [n] - [k]} e_{(i, \sigma(i))} \\
& + \sum_{C(\gamma) \text{ is a conjugacy class with a cycle of length } k} a_{\gamma} \sum_{\sigma \in C(\gamma)} \prod_{i \in [k]} e_{(i, \sigma(i))} \prod_{i \in [n] - [k]} e_{(i, \sigma(i))} \\
& = \sum_{C(\gamma) \text{ is a conjugacy class with a cycle of length } k} a_{\gamma} \sum_{\sigma \in C(\gamma)} \prod_{i \in [k]} e_{(i, \sigma(i))} \prod_{i \in [n] - [k]} e_{(i, \sigma(i))},
\end{split}
\end{equation*}
since if the conjugacy class does not have a cycle of length $k$, then $\prod_{i \in [k]} e_{(i, \sigma(i))} = 0$, no matter which $\gamma$ we take. If a conjugacy class has a cycle of length $k$, then we will choose any $\gamma$ from that conjugacy class such that $\gamma(i) = i + 1$ for $1 \leq i < k$ and $\gamma(k) = 1$, and we will have $\prod_{i \in [k]} e_{(i, \gamma(i))} = 1$. This will make $H_{n, p} \left( E_{n, p} \right)$ identically non-zero.  

If $a_{\sigma} = 0$ for all $\sigma \in S_n$ with a cycle of length $k$, then $H_{n, p} \left( E_{n, p} \right)$, over our distribution is uniformly $0$, proving our first assertion. If any of the $a_{\sigma}$'s are left in the remaining sum is non-zero, then due to the Schwartz-Zippel Lemma (\ref{lemma:1}), we have
\begin{equation*}
\mathcal{P}_{E_{n, p} \leftarrow_r \mathcal{D}} \left[ H_{n, p} \left( E_{n, p} \right) = 0 \right] \leq \frac{n}{p},
\end{equation*}
where, $\mathcal{D}$ is the random distribution in Algorithm \ref{alg:3}.

\end{proof}

Now, we will prove that there is a test certifying that $M$ computes $\HCY_{n, p}$.
\begin{lemma}
\label{lemma:13}

Given an oracle machine $M$ computing a polynomial $H_{n, p}: \mathbb{Z}_p^{n^2} \to \mathbb{Z}_p$ of degree at most $n$ and $p > n$, in polynomial-time, we have a test such that,

\begin{enumerate}

\item If $H_{n, p} = \HCY_{n, p}$, then with a probability of $1$, the test accepts.

\item If $H_{n, p} \neq \HCY_{n, p}$, then the test accepts with a probability of at most $1 / 2^n$.

\end{enumerate}

\end{lemma}

\begin{proof}

First, we are promised that $M$ computes a polynomial $H_{n, p}: \mathbb{Z}_p^{n^2} \to \mathbb{Z}_p$ of degree at most $n$. We run $\text{GenPerm}(M, n, p)$ (Algorithm \ref{alg:1}) polynomially many times so that if $H_{n, p} \in \mathbb{V}_{\GPR_{n, p}}$, then this test always returns \texttt{ACCEPT}, but if $H_{n, p} \not \in \mathbb{V}_{\GPR_{n, p}}$, then with a probability of greater than $1 - 2^{-n^2}$, the test returns \texttt{REJECT} at least once. If all returns are \texttt{ACCEPT}, then we move forward and repeat the same with $\text{IsoPerm}(M, n, p)$ (Algorithm \ref{alg:2}). If $H_{n, p} \in \mathbb{I}_{\GPR_{n, p}}$, then with a probability of $1$, the result is always \texttt{ACCEPT}. If $H_{n, p} \not \in \mathbb{I}_{\GPR_{n, p}}$, then the probability that all the returns are \texttt{ACCEPT} is at most $2^{-n^2}$. We go ahead only if all returns are \texttt{ACCEPT}.

Now, we are reasonably sure that $H_{n, p} \in \mathbb{I}_{\GPR_{n, p}}$. We need to certify that for every $\sigma \in S_n - C_n$, $a_{\sigma} = 0$, and that $a_{\sigma} = 1$ if $\sigma \in C_n$. We run $\text{NoCycle}(M, n, p, k)$ (Algorithm \ref{alg:3}) on $M$ with $k$ from $1$ to $\lfloor n / 2 \rfloor$. We repeat each one polynomially many times to get an error probability of $2^{-n^2}$ for each iteration. If all repetitions of all iterations return \texttt{ACCEPT}, then we proceed forward. Otherwise, we reject it.

Now, all that is left to determine if our polynomial $H_{n, p}$ is $\HCY_{n, p}$ is to determine if $a_{\sigma} = 1$ for all $\sigma \in C_n$. We construct $E_{n, p}$ such that $e_{(i, i+1)} = 1$ for all $i \in [n - 1]$, and $e_{(n, 1)} = 1$, and all other entries are set to $0$. One can see that $a_{\sigma} = M \left( E_{n, p} \right)$. If $M \left( E_{n, p} \right) = 1$, we accept; otherwise, we reject.

If $H_{n, p} = \HCY_{n, p}$, all tests will pass with a probability of $1$. If that is not the case, then due to the union bound, the probability that all tests will accept is at most
\begin{equation*}
\frac{\poly(n)}{2^{n^2}} \leq 2^{-n}.
\end{equation*}
Moreover, we only use $\poly(n)$-time and $\poly(n)$-queries to $M$.

\end{proof}

In the last step, even if we find out that $H_{n, p} = a \HCY_{n, p}$ for some $a \neq 0$, a quick algorithm computing $H_{n, p}$ should not exist since this is also an oracle from which we can decide $\dHAMCYCLE$.

\subsection{Main Results for Hamiltonian Cycle Counting}
\label{section:6.5}

We are now ready to state our main results for counting unique Hamiltonian cycles modulo $p$ on directed multigraphs.

\begin{theorem}
\label{theorem:6}

Given a prime $p = 2^{n^{O(1)}} > 2n$, and an oracle $O$ counting the number of unique Hamiltonian cycles modulo $p$ on $n$-vertex directed multigraphs correctly on more than an $\epsilon$-fraction of instances (with $\epsilon = 1 / 2^{o(n)}$), where the number of edges varies from $0$ to $p-1$, then, in $\poly(n, \log p, 1/\epsilon)$ time and queries to $O$, we have a probabilistic algorithm that computes the number of unique Hamiltonian cycles modulo $p$ on $n$-vertex directed multigraphs with an error probability of less than $O \left( 2^{-\sqrt{n}} \right)$.

\end{theorem}

\begin{proof}

Here, we use the STV list decoder (Lemma \ref{lemma:2}) and have access to oracle machines $(M_i)_{i = 1}^t$, where $t = O(1 / \epsilon)$, each of which computes a polynomial $H_{n, p} : \mathbb{Z}_p^{n^2} \to \mathbb{Z}_p$ of degree at most $n$ with an error probability of at most $2^{-n^2}$. Note that due to the union bound, with a probability of $1- 2^{-n^{1.99}}$, all these machines compute all queries correctly, allowing us to use the tests of previous subsections. We can use the union bound to discard all cases where these machines do not behave identically to deterministically, always-correct machines.
    
We can run all our tests on each machine and pick the machine $M_i$ that computes $\HCY_{n, p}$ with high probability. The probability that any one thing goes wrong is at most $2^{-n}$. Due to the union bound, the probability that at least one of the undesired events occurs is at most
\begin{equation*}
\frac{t}{2^n} = \frac{O \left( 1 / \epsilon \right)}{2^n} = \frac{O \left( 2^{o(n)} \right)}{2^n} = O \left( 2^{-\sqrt{n}} \right).
\end{equation*}
Moreover, this algorithm takes $\poly(n, \log(p), 1/\epsilon)$-time, including the STV runtime.

\end{proof}

We can state the following corollary based on the well-believed structural conjecture that $\NP \not\subset \BPP$.

\begin{corollary}
\label{cor:2}

If $\NP \not \subset \BPP$, then for all $\alpha > 0$ and for all primes $p = \Theta \left( 2^{n^{O(1)}} \right)$, no polynomial-time randomized algorithm can count the number of unique Hamiltonian cycles modulo $p$ on $n$-vertex directed multigraphs with a probability of correctness greater than $2/3$ on even a $1/n^{\alpha}$-fraction of instances.

\end{corollary}

\begin{proof}

Due to the reduction from $\dHAMCYCLE$ to $\dHAMCYCLE_p$ in Lemma \ref{lemma:7}, and Theorem \ref{theorem:6}, we would have a polynomial-time randomized algorithm for $\dHAMCYCLE$, implying that $\NP \subset \BPP$.

\end{proof}

We now state our most important result for counting unique Hamiltonian cycles modulo $p$ on directed multigraphs.

\begin{theorem3}
\label{theorem:7}

If there is a randomized algorithm computing the number of Hamiltonian cycles modulo $p$ on $n$-vertex directed multigraphs for a prime $p = \Theta \left( 2^{n^{c}} \right)$ on a $1 / g(n)$-fraction of instances in $2^{o(n)}$-time for any $c \geq 1$, and $g(n) = 2^{o(n)}$, then $\rETH$ is false.

\end{theorem3}
\begin{proof}

Due to Lemma \ref{lemma:6}, Corollary \ref{corollary:1}, Lemma \ref{lemma:7}, and the reduction in \cite{Arora2009}, $3\SAT$ on $O(n)$ clauses reduces to counting the number of unique Hamiltonian cycles modulo $p$ on a directed multigraph $D_{n, p} = \left( V_n, E_{n, p} \right)$ with $\left| V_n \right| = O(n)$. Due to Theorem \ref{theorem:6}, if such an algorithm exists, we can decide $3\SAT$ in $2^{o(n)}$-time, violating $\rETH$.

\end{proof}

For context, it is crucial to realize that simply printing $0$, oblivious of the input, is a strategy that is successful on at least an $\Omega(1/2^{n})$-fraction of instances for $p = \Theta(2^n)$. Notice that we can write the function $\HCY_{n, p} \left( E_{n, p} \right)$ as
\begin{equation*}
\HCY_{n, p} \left( E_{n, p} \right) = e_{(1, 2)} H_1 \left( E_{n, p} \right) + H_2 \left( E_{n, p} \right),
\end{equation*}
where, the functions $H_1$ and $H_2$ are independent of the variable $e_{(1, 2)}$.

Notice that, due to the Schwartz-Zippel Lemma (\ref{lemma:1}), choosing $E_{n, p} - \{\, e_{(1, 2)} \,\}$ at random from $\mathbb{Z}_p^{n^2 - 1}$ results in $H_1 \left( E_{n, p} \right)$ being $0$ with a probability of at most $(n - 1) / p$. With a probability of greater than $1 - (n - 1) / p$, $H_1 \left( E_{n, p} \right)$ evaluates to a non-zero value. Since the coefficient of $e_{(1,2)}$ is non-zero with a probability of at least $1 - (n - 1) / p$, it can ``influence'' the value of $\HCY_{n, p} \left( E_{n, p} \right)$. Given that the coefficient of $e_{(1, 2)}$ is non-zero, the value of $\HCY_{n, p} \left( E_{n, p} \right)$ is $0$ with a probability of $1 / p$ over the choices of $e_{(1, 2)} \in \mathbb{Z}_p$. Hence, the probability that $\HCY_{n, p} \left( E_{n, p} \right) = 0$ on a random input is at least
\begin{equation*}
\frac{1}{p}\left(1-\frac{n-1}{p}\right) = 1/p - o\left(1/p\right) = \Omega(1/2^{n}).
\end{equation*}

These results show that there is a somewhat ``sharp'' increase in the plausible correctness of an algorithm trying to compute $\HCY_{n, p}$ as the time complexity approaches the worst-case complexity. 

\section{Half-Clique Counting is Hard for Algorithms}
\label{section:7}

In this section, similar to the case with counting directed Hamiltonian cycles, we will show that, given an oracle $M$ computing a polynomial $H_{n, p}: \mathbb{Z}_p^{n \choose 2} \to \mathbb{Z}_p$, in $\poly(n)$-time, we can decide if $H_{n, p} = \HCL_{n, p}$. As was the case with Section \ref{section:6}, for simplicity, we will assume that $M$ computes $H_{n, p}$ correctly with a probability of $1$. Since the STV list decoder (Lemma \ref{lemma:2}) has an exponentially small error, $M$ makes a mistake on at least one query with an exponentially small error due to the union bound. Hence, with a very high probability $M$ behaves identically to a machine that is always correct. In fact, in expectation, one has to make exponentially many queries to $M$ for it to make a mistake.

\subsection{Subgraph Counting and Isomorphic Polynomials}
\label{section:7.1}

In this subsection, we want to define some useful abstractions for ourselves, similar to Subsection \ref{section:6.1}. In place of $\GPR_{n, p}: \mathbb{Z}_p^{n!} \times \mathbb{Z}_p^{n^2} \to \mathbb{Z}_p$, we have the \textit{Multi-Linear Polynomial}, $\MLP_{n, p}$, defined as follows.
\begin{definition}
\label{def:16}

\textit{\textbf{The Multi-Linear Polynomial Function.} \\
We define the polynomial $\MLP_{n, p} : \mathbb{Z}_p^{2^{\binom{n}{2}}} \times \mathbb{Z}_p^{\binom{n}{2}} \to \mathbb{Z}_p$ as
\begin{equation}
\label{eq:17}
\MLP_{n, p} \left( A_{n, p}, F_{n, p} \right) = \sum_{S \subset F_{n, p}} a_S \prod_{e_{\{\, i, j \,\}} \in S} e_{\{\, i, j \,\}},
\end{equation}
where $A_{n, p}$ is the list of all $a_{S}$, each being a coefficient corresponding to the subset $S$ of edges.}

\end{definition}

In this case, the analogue of the generalized permanent vector space, $\mathbb{V}_{\GPR_{n, p}}$, is simply the vector space of multi-linear polynomials over $Z_p$, $\mathbb{V}_{\MLP_{n, p}}$ of dimension $2^{\binom{n}{2}}$. From Equation \eqref{eq:17}, we define the vector space $\mathbb{V}_{\MLP_{n, p}}$ as follows.
\begin{definition}
\label{def:17}

\textit{\textbf{The Multi-Linear Polynomial Vector Space.} \\
We define the vector space $\mathbb{V}_{\MLP_{n, p}}$ over $\mathbb{Z}_p$ with the usual operations of polynomial addition and scalar multiplication as}
\begin{equation*}
\mathbb{V}_{\MLP_{n, p}} = \left \{\, \MLP_{A_{n, p}} \mid A_{n, p} \in \mathbb{Z}_p^{2^{\binom{n}{2}}}, \MLP_{A_{n, p}} \left( F_{n, p} \right) = \MLP_{n, p} \left( A_{n, p}, F_{n, p} \right) \, \right\}.
\end{equation*}

\end{definition}

Every multi-linear polynomial on $F_{n, p}$ is a subgraph counting polynomial, similar to the class $\mathbb{V}_{\GPR_{n, p}}$ of polynomials that counted cycle covers. We define below the action of a permutation $\pi \in S_n$ on the edge list $F_{n, p}$.

\begin{definition}
\label{def:18}

\textit{\textbf{Permutation of an Undirected Multigraph.} \\
Let $\pi \in S_n$. Given the edge list $F_{n, p} = \left( \left( e_{\{\, i, j \,\}} \right)_{j \in [n] - [1]} \right)_{1 \leq i < j}$, we define $\pi(F_{n, p})$ as}
\begin{equation*}
\pi(F_{n, p}) = \left( \left( e_{\{\, \pi(i), \pi(j) \,\}} \right)_{j \in [n] - [1]} \right)_{1 \leq i < j}.
\end{equation*}

\end{definition}

\begin{definition}
\label{def:19}

\textit{\textbf{The Subgroup Induced by $H_{n, p}$.} \\
We define the subgroup induced by $H_{n, p}$, $\Pi \left( H_{n, p} \right)$, as the subgroup of $S_n$ such that for all $\pi \in \Pi \left( H_{n, p} \right)$, $H_{n, p} \left( F_{n, p} \right)$ and $H_{n, p} \left( \pi \left( F_{n, p} \right) \right)$ are formally the same polynomials:}
\begin{equation}
\label{eq:18}
\Pi \left( H_{n, p} \right) = \left \{\, \pi \mid \pi \in S_n, H_{n, p} \left( F_{n, p} \right) = H_{n, p} \left( \pi \left( F_{n, p} \right) \right) \, \right\}.
\end{equation}

\end{definition}

By slightly changing Definition \ref{def:14}, we get a subspace, $\mathbb{I}_{\MLP_{n, p}}$, of $\mathbb{V}_{\MLP_{n, p}}$ as follows.

\begin{definition}
\label{def:20}

\textit{\textbf{The Subspace of Isomorphic Polynomials.} \\
Using Equation \eqref{eq:18}, we define the vector space $\mathbb{I}_{\MLP_{n, p}}$ as}
\begin{equation*}
\mathbb{I}_{\MLP_{n, p}} = \left \{\, H_{n, p} \mid H_{n, p} \in \mathbb{V}_{\MLP_{n, p}}, \Pi \left( H_{n, p} \right) = S_n \, \right\}.
\end{equation*}

\end{definition}

It can be seen that $H_{n, p} \in \mathbb{I}_{\MLP_{n, p}}$ if and only if $\Pi \left( H_{n, p} \right) = S_n$. Informally, this can be seen as $H_{n, p}$ weighing isomorphic graphs equally. We will go further in categorizing the polynomials in $\mathbb{V}_{\MLP_{n, p}}$. Suppose we have a subset $S$ of $F_{n, p}$. We define the polynomial $H_{n, p}^S$ for counting the number of subgraphs at position $S$ as follows. 

\begin{definition}
\label{def:21}

\textit{\textbf{Subgraph Counting Polynomials.} \\
Let $S \subset F_{n, p}$. We define the function $H_{n, p}^S: \mathbb{Z}_p^{\binom{n}{2}} \to \mathbb{Z}_p$ as}
\begin{equation}
\label{eq:19}
H_{n, p}^S \left( F_{n, p} \right) = \prod_{e_{\{\, i, j \,\}} \in S} e_{\{\, i, j \,\}}.
\end{equation}

\end{definition}

However, if we want to count the number of subgraphs isomorphic to $S$, using Equation \eqref{eq:19}, the following polynomial is more useful.
\begin{definition}
\label{def:22}

\textit{\textbf{Isomorphic Subgraphs Counting Polynomials.} \\
Let $S \subset F_{n, p}$. We define the function $H_{n, p}^{I(S)}: \mathbb{Z}_p^{\binom{n}{2}} \to \mathbb{Z}_p$ as}
\begin{equation}
\label{eq:20}
\begin{split}
H_{n, p}^{I(S)} \left( F_{n, p} \right) & = \sum_{\Aut(S)\sigma \subset S_n} H_{n, p}^{\sigma(S)} \left( F_{n, p} \right) \\
& = \sum_{\Aut(S)\sigma \subset S_n} \prod_{e_{\{\, i, j \,\}} \in S} e_{\{\, \sigma(i), \sigma(j) \,\}}.
\end{split}
\end{equation}

\end{definition}

We make the following assertion about the structure of $\mathbb{I}_{\MLP_{n, p}}$.
\begin{lemma}
\label{lemma:14}

$\left \{\, H_{n, p}^{I(S)} \mid S \subset F_{n, p} \, \right\}$ forms a basis for $\mathbb{I}_{\MLP_{n, p}}$.

\end{lemma}

\begin{proof}

First, we will prove that $H_{n, p}^{I(S)} \in \mathbb{I}_{\MLP_{n, p}}$. From Equation \eqref{eq:20}, we have
\begin{equation*}
\begin{split}
H_{n, p}^{I(S)} \left( \pi \left( F_{n, p} \right) \right) & = \sum_{\Aut(S)\sigma \subset S_n} H_{n, p}^{\sigma(S)} \left( \left( \left( e_{\{\, \pi(i), \pi(j) \,\}} \right)_{j \in [n] - 1} \right)_{1 \leq i < j} \right) \\
& = \sum_{\Aut(S)\sigma \subset S_n} \prod_{e_{\{\, i, j \,\}} \in S} e_{\{\, \pi(\sigma(i)), \pi(\sigma(j)) \,\}} \\
& = \sum_{\Aut(S)\sigma \subset S_n} \prod_{e_{\{\, i, j \,\}} \in S} e_{\{\, \sigma \pi (i), \sigma \pi (j) \,\}} \\
& = \sum_{\Aut(S)\sigma^\prime \pi^{-1} \subset S_n} \prod_{e_{\{\, i, j \,\}} \in S} e_{\{\, \sigma^\prime (i), \sigma^\prime (j) \,\}} \\
& = \sum_{\Aut(S)\sigma^\prime \subset S_n \pi} \prod_{e_{\{\, i, j \,\}} \in S} e_{\{\, \sigma^\prime (i), \sigma^\prime (j) \,\}} \\
& = \sum_{\Aut(S)\sigma^\prime \subset S_n} \prod_{e_{\{\, i, j \,\}} \in S} e_{\{\, \sigma^\prime (i), \sigma^\prime (j) \,\}} \\
& = H_{n, p}^{I(S)} \left( F_{n, p} \right) \\
& \implies H_{n, p}^{I(S)} \in \mathbb{I}_{\MLP_{n, p}},
\end{split}
\end{equation*}
where $\sigma^\prime = \sigma \pi$, and we have used the fact that $S_n \pi = S_n$.

To prove that the functions $H_{n, p}^{I(S)}$ form a basis for $\mathbb{I}_{\MLP_{n, p}}$ is equivalent to saying that for any $H_{n, p} \in \mathbb{V}_{\MLP_{n, p}}$, $H_{n, p} \in \mathbb{I}_{\MLP_p}$ if and only if in the $\MLP_{A_{n, p}}$ representation of $H_{n, p}$, $S \sim S^\prime \implies a_{S} = a_{S^\prime}$. Since $H_{n, p} \in \mathbb{I}_{\MLP_p}$, from Definition \ref{def:20}, we have
\begin{equation}
\label{eq:21}
\begin{split}
H_{A_{n, p}} \left( \pi \left( F_{n, p} \right) \right) & = H_{n, p} \left( A_{n, p}, \pi \left( F_{n, p} \right) \right) \\
& = H_{n, p} \left( A_{n, p}, \pi \left( \left( \left( e_{\{\, i, j \,\}} \right)_{j \in [n] - [1]} \right)_{1 \leq i < j} \right) \right) \\
& = \sum_{S \subset F_{n, p}} a_S \prod_{e_{\{\, i, j \,\}} \in S} e_{\{\, \pi(i), \pi(j) \,\}} \\
& = \sum_{S \subset F_{n, p}} a_S \prod_{\{\, e_{\pi^{-1}(k), \pi^{-1}(l) \,\}} \in S} e_{\{\, k, l \,\}} \\
& = \sum_{S \subset F_{n, p}} a_S \prod_{e_{\{\, k, l \,\}} \in \pi(S)} e_{\{\, k, l \,\}} \\
& = \sum_{\pi^{-1} \left( S^\prime \right) \subset F_{n, p}} a_S \prod_{e_{\{\, k, l \,\}} \in S^\prime} e_{\{\, k, l \,\}} \\
& = \sum_{S^\prime \subset \pi \left( F_{n, p} \right)} a_S \prod_{e_{\{\, k, l \,\}} \in S^\prime} e_{\{\, k, l \,\}} \\
& = \sum_{S^\prime \subset F_{n, p}} a_S \prod_{e_{\{\, k, l \,\}} \in S^\prime} e_{\{\, k, l \,\}},
\end{split}
\end{equation}
where $k = \pi(i)$, $l = \pi(j)$, $\pi(S) = S^\prime$, and we have used the fact that $\pi$ only changes the ordering of the elements of $F_{n, p}$. Comparing Equations \eqref{eq:17} and \eqref{eq:21}, we get that $a_S = a_{S^\prime}$, proving our lemma.

\end{proof}

\subsection{Fast Test to Identify Polynomials in $V_{\MLP_{n, p}}$}
\label{section:7.2}

In this subsection, we intend to give a test to decide, given an oracle machine $M$ promised to compute a polynomial $H_{n, p}: \mathbb{Z}_p^{n \choose 2} \to \mathbb{Z}_p$, whether $H_{n, p}$ is multi-linear. This proof is not as ``slick'' as that for counting directed Hamiltonian cycles (Section \ref{section:6.2}), as the pigeonhole principle argument does not extend. We give our test in Algorithm \ref{alg:4}. The intuition for this is that the ``partial derivative'' of $H_{n, p}$ with $e_{\{\, i,j \,\}}$ as the only degree of freedom is constant if $H_{n, p}$ is linear in $e_{\{\, i,j \,\}}$, but it is far from the case when the degree of $e_{\{\, i,j \,\}}$ in $H_{n, p}$ is more than one.

\begin{algorithm}
\caption{$\text{MultiLinear}\left( M, n, p, e_{\{\, i, j \,\}} \right)$}
\label{alg:4}
\Comment{The Multi-linear in $e_{\{\, i, j \,\}}$ Test \hspace{10.6cm}} \\
\Comment{Input: $n$ is the number of vertices in the undirected multigraph}\\
\Comment{Input: $p > n (n - 1)$ is a prime} \\
\Comment{Input: for $j \in [n] - 1$, $1 \leq i < j$, $e_{\{\, i, j \,\}}$ is a variable} \\
\Comment{Input: The oracle machine $M$ computes a polynomial $H_{n, p}: \mathbb{Z}_p^{n \choose 2} \to \mathbb{Z}_p$} of degree at most $\displaystyle \binom{n}{2} $\\
\Comment{Output: \texttt{ACCEPT} if $H_{n, p} \in \mathbb{V}_{\MLP_{n, p}}$, \texttt{REJECT} otherwise} \\
\Comment{All computations are done over $\mathbb{Z}_p$}
\begin{algorithmic}
\State $F_{n, p} = \left( \left( e_{\{\, i, j \,\}} \right)_{j \in [n] - 1} \right)_{1 \leq i < j} \gets 0^{n \choose 2}$
\For{all $\left( i^\prime, j^\prime \right) \neq (i, j)$ in $F_{n, p}$}
    \State $e_{\{\, i^\prime, j^\prime \,\}} \gets_r \mathbb{Z}_p$
\EndFor
\State $a_1 \gets_r \mathbb{Z}_p$
\State $a_2 \gets_r \mathbb{Z}_p$
\State $F_1, F_2, F^\prime_1, F^\prime_2 \gets F_{n, p}$
\State $F_1.e_{\{\, i,j \,\}} \gets a_1$
\State $F_2.e_{\{\, i,j \,\}} \gets a_2$
\State $F^\prime_1.e_{\{\, i,j \,\}} \gets a_1 + 1$
\State $F^\prime_2.e_{\{\, i, j \,\}} \gets a_2 + 1$ \\
\Comment{Comparing the finite differences $\Delta_{\{\, i, j \,\}} H_{n, p} \left( F_{n, p} \right)$ at $e_{\{\, i,j \,\}} = a_1$ and $e_{\{\, i,j \,\}} = a_2$}
\If{$M \left( F^\prime_1 \right) - M (F_1) \neq M \left( F^\prime_2 \right) - M (F_2)$}
    \State \Return \texttt{REJECT}
\EndIf
\State \Return \texttt{ACCEPT}
\end{algorithmic}
\end{algorithm}
\begin{lemma}
\label{lemma:15}

For a prime $p > n(n - 1)$, given an oracle machine $M$ computing a polynomial $H_{n, p}: \mathbb{Z}_{p}^{n \choose 2} \to \mathbb{Z}_p$ of degree at most $\displaystyle \binom{n}{2}$, $\text{MultiLinear}\left( M, n, p, e_{\{\, i, j \,\}} \right)$ takes polynomial-time such that

\begin{enumerate}

\item If $H_{n, p}$ is multilinear, then $M$ passes the test with a probability of $1$.

\item If $H_{n, p}$ is not multilinear, then $M$ passes the test with a probability of at most $2^{-n^2}$.

\end{enumerate}

\end{lemma}

\begin{proof}

Note that, for some $m \in \mathbb{N} \cup \{\, 0 \,\}$, we can write the polynomial $H_{n, p}$ as
\begin{equation}
\label{eq:22}
H_{n, p} \left( F_{n, p} \right) = \sum_{k = 0}^m H_k \left( F_{n, p} \right) e_{\{\, i, j \,\}}^k,
\end{equation}
by grouping it into coefficients of $e_{\{\, i,j \,\}}$, where each $H_k$ for $k \in \{\, 0 \,\} \cup [m]$ is a polynomial independent of $e_{\{\, i, j \,\}}$. By applying the finite difference operator $\Delta_{\{\, i, j \,\}}$\footnote{$\Delta_{\{\, i, j \,\}} f \left( \ldots e_{\{\, i, j \,\}} \ldots \right) = f \left( \ldots e_{\{\, i, j \,\}} + 1 \ldots \right) - f \left( \ldots e_{\{\, i, j \,\}} \ldots \right).$} to Equation \eqref{eq:22}, we get
 \begin{equation*}
\begin{split}
\Delta_{\{\, i, j \,\}} H_{n, p} \left( F_{n, p} \right) & = \sum_{k = 0}^m H_k \left( F_{n, p} \right) \left( \left( e_{\{\, i, j \,\}} + 1 \right)^k - e_{\{\, i, j \,\}}^k \right) \\
& = H_1 \left( F_{n, p} \right) + \sum_{k = 2}^m H_k \left( F_{n, p} \right) \left( \left( e_{\{\, i, j \,\}} + 1 \right)^k - e_{\{\, i, j \,\}}^k \right).
\end{split}
\end{equation*}

For any $m > 1$, if the degree of $e_{\{\, i,j \,\}}$ in $H_{n, p}$ is $m$, then $H_m \left( F_{n, p} \right)$ is not uniformly $0$, and it is $0$ with a probability of at most
\begin{equation}
\label{eq:23}
\frac{\displaystyle {n \choose 2} - m}{p} \leq \frac{\displaystyle {n \choose 2}}{p},
\end{equation}
due to the Schwartz-Zippel Lemma (\ref{lemma:1}). If $H_m \left( F_{n, p} \right)$ is not $0$, then $\Delta_{\{\, i,j \,\}} H_{n, p} \left( F_{n, p} \right)$ is a polynomial of degree $m - 1$ in $e_{\{\, i, j \,\}}$, with the leading coefficient $m$. Since $m$ is at most $\displaystyle {n \choose 2}$, for any $p > \displaystyle {n \choose 2}$, the univariate polynomial $\Delta_{\{\, i,j \,\}} H_{n, p} \left( F_{n, p} \right)$ in $e_{\{\, i,j \,\}}$ is of degree $m - 1$.

With $a_1$ and $a_2$ chosen uniformly and independently from $\mathbb{Z}_p$, the probability that
\begin{equation*}
M \left( F^\prime_1 \right) - M (F_1) = M \left( F^\prime_2 \right) - M (F_2)
\end{equation*}
in $\text{MultiLinear}\left( M, n, p, e_{i, j} \right)$ is at most
\begin{equation}
\label{eq:24}
\frac{m - 1}{p} \leq \frac{\displaystyle \binom{n}{2} - 1}{p},
\end{equation}
due to the fundamental theorem of algebra.

From Equations \eqref{eq:23} and \eqref{eq:24}, using the union bound, if the degree of $e_{\{\, i,j \,\}}$ is greater than $1$ in the polynomial $H_{n, p}$, the probability that $\text{MultiLinear}\left( M, n, p, e_{\{\, i, j \,\}} \right)$ returns \texttt{ACCEPT} is at most
\begin{equation*}
\frac{\displaystyle {n \choose 2}}{p} + \frac{\displaystyle \binom{n}{2} - 1}{p} < \frac{n(n - 1)}{p}.
\end{equation*}

On the other hand, if the degree of $e_{\{\, i,j \,\}}$ is at most $1$, then $\Delta_{\{\, i,j \,\}} H_{n, p} \left( F_{n, p} \right)$ is always a constant, meaning that $\text{MultiLinear}\left( M, n, p, e_{\{\, i, j \,\}} \right)$ always returns \texttt{ACCEPT}.

For $p > n(n - 1)$, we can repeat this algorithm polynomially many times for each $e_{\{\, i,j \,\}}$ to reduce the error to an exponentially small probability. Even after union bounding over the $\displaystyle {n \choose 2}$ variables, we can have error probability at most $2^{-n^2}$.

\end{proof}

\subsection{Fast Test to Identify Polynomials in $\mathbb{I}_{\MLP_{n, p}}$}
\label{section:7.3}

In this section, we will show that given an oracle machine $M$ promised to compute a multi-linear polynomial $H_{n, p}: \mathbb{Z}_{p}^{n \choose 2} \to \mathbb{Z}_p$, we can determine if $H_{n, p} \in \mathbb{I}_{\MLP_{n, p}}$ very quickly. Note that this is the same as the test in section \ref{section:6.3}. For completeness, we describe the test again in Algorithm \ref{alg:5}.

\begin{algorithm}
\caption{IsoMult$\left( M, n, p \right)$}
\label{alg:5}
\Comment{The Isomorphic Polynomial Test \hspace{10.2cm}} \\
\Comment{Input: $n$ is the number of vertices in the undirected multigraph} \\
\Comment{Input: $p > n(n - 1)$ is a prime} \\
\Comment{Input: The oracle machine $M$ computes a polynomial $H_{n, p}: \mathbb{Z}_p^{n \choose 2} \to \mathbb{Z}_p \in \mathbb{V}_{\MLP_{n, p}}$} of degree at most $\displaystyle \binom{n}{2}$ \\
\Comment{Output: \texttt{ACCEPT} if $H_{n, p} \in \mathbb{I}_{\MLP_{n, p}}$, \texttt{REJECT} otherwise} \\
\Comment{All computations are done over $\mathbb{Z}_p$}
\begin{algorithmic}
\State $F_{n, p} \gets_r \mathbb{Z}^{n \choose 2}_p$
\State $\pi \gets_r S_n$
\State $F^\prime_{n, p} \gets \pi \left( F_{n, p} \right)$ \Comment{Using Definition \ref{def:18}}
\State $g \gets M \left( F_{n, p} \right)$
\State $g^\prime \gets M \left( F^\prime_{n, p} \right)$
\If{$g \neq g^\prime$}
    \State \Return \texttt{REJECT}
\EndIf
\State \Return \texttt{ACCEPT}
\end{algorithmic}
\end{algorithm}

We now prove the probabilistic correctness of this test.

\begin{lemma}
\label{lemma:16}

Given that the polynomial $H_{n, p}$ computed by $M$ is in $\mathbb{V}_{\MLP_{n, p}}$ and $p > n(n-1)$ is a prime,

\begin{enumerate}

\item If $H_{n, p} \in \mathbb{I}_{\MLP_{n, p}}$, then the test returns \texttt{ACCEPT} with a probability of $1$.

\item If $H_{n, p} \not \in \mathbb{I}_{\MLP_{n, p}}$, then the test returns \texttt{ACCEPT} with a probability of at most $1 / 2 + \displaystyle {n \choose 2} / p$.

\end{enumerate}

\end{lemma}

\begin{proof}

The first assertion is true by Definition \ref{def:20}. When $H_{n, p} \not \in \mathbb{I}_{\MLP_p}$, that is, when $\Pi \left( H_{n, p} \right) \neq S_n$, from Equation \eqref{eq:18} and using Lagrange's theorem \citep{Herstein1975}, given that $|\Pi \left( H_{n, p} \right)|$ divides $|S_n|$, we have
\begin{equation}
\label{eq:25}
\begin{split}
P_1 & = \mathcal{P}_{\substack{\pi \leftarrow_r S_n \\ F_{n, p} \leftarrow_r \mathbb{Z}_p^{\binom{n}{2}}}} \left[ H_{n, p} \left( F_{n, p} \right) = H_{n, p} \left( \pi \left( F_{n, p} \right) \right) \mid \text{ these are formally the same polynomials} \right] \\
& = \frac{|\Pi \left( H_{n, p} \right)|}{|S_n|} \\
& \leq \frac{1}{2}.
\end{split}
\end{equation}

Using the Schwartz-Zippel Lemma (\ref{lemma:1}), we also have
\begin{equation}
\label{eq:26}
\begin{split}
P_2 & = \mathcal{P}_{\substack{\pi \leftarrow_r S_n \\ F_{n, p} \leftarrow_r \mathbb{Z}_p^{\binom{n}{2}}}} \left[ H_{n, p} \left( F_{n, p} \right) = H_{n, p} \left( \pi \left( F_{n, p} \right) \right) \mid \text{ these are not formally the same polynomials} \right] \\
& \leq  \frac{\displaystyle \binom{n}{2}}{p}.
\end{split}
\end{equation}

From Equations \eqref{eq:25} and \eqref{eq:26}, we get
\begin{equation*}
\begin{split}
\mathcal{P}_{\substack{\pi \leftarrow_r S_n \\ F_{n, p} \leftarrow_r \mathbb{Z}_p^{\binom{n}{2}}}} \left[ H_{n, p} \left( F_{n, p} \right) = H_{n, p} \left( \pi \left( F_{n, p} \right) \right) \right] & = P_1 + P_2 \\
& \leq \frac{1}{2} + \frac{\displaystyle \binom{n}{2}}{p}.
\end{split}
\end{equation*}
    
Using this, we have that the probability that both evaluations are the same is at most $1 / 2 + \displaystyle \binom{n}{2} / p$.

\end{proof}

By repetition, we can lower the probability of error to at most $2^{-n^2}$ in polynomially many iterations.

\subsection{Fast Test to Identity the $\HCL_{n, p}$ Polynomials}
\label{section:7.4}

So far, we have given tests not specific to the problems of counting cliques of size $\lfloor n / 2 \rfloor$. The next test is this section's main feature and explicitly uses the combinatorial properties
of the polynomial $\HCL_{n, p}$. We present the test in Algorithm \ref{alg:6} to distinguish all other polynomials in $\mathbb{I}_{\MLP_{n, p}}$ that are not $\HCL_{n, p}$ from the polynomial $\HCL_{n, p}$.

\begin{algorithm}
\caption{$\text{Is}\HCL\left( M, n, p \right)$}
\label{alg:6}
\Comment{The $\HCL_{n, p}$ Test \hspace{13cm}} \\
\Comment{Input: $n$ is the number of vertices in the undirected multigraph} \\
\Comment{Input: $p > n (n - 1)$ is a prime} \\
\Comment{Input: The oracle machine $M$ computes a polynomial $H_{n, p}: \mathbb{Z}_p^{n \choose 2} \to \mathbb{Z}_p$ of degree at most $\displaystyle {{\lfloor n / 2 \rfloor} \choose 2}$} \\
\Comment{Output: \texttt{ACCEPT} if $H_{n, p} = \HCL_{n, p}$, \texttt{REJECT} otherwise} \\
\Comment{All computations are done over $\mathbb{Z}_p$}
\begin{algorithmic}
\State $F_{n, p} \gets 0^{n \choose 2}$
\For{each $1 < i < n$}
    \For{each $i < j \leq n$}
        \State $e_{\{\, i, j \,\}} \gets_r \mathbb{Z}_p$
    \EndFor
\EndFor
\State $F_1, F_2 \gets F_{n, p}$
\Comment{The vertex $1$ is isolated in $F_1$}
\For{each $1 < j < \displaystyle \left \lfloor \frac{n}{2} \right \rfloor$}
    \State $F_2.e_{\{\, 1,j \,\}} \gets 1$ 
    \Comment{The vertex $1$ has degree $\displaystyle \left \lfloor \frac{n}{2} \right \rfloor - 2$ in $F_2$}
\EndFor
\State $g_1 \gets M \left( F_1 \right)$
\State $g_2 \gets M \left( F_2 \right)$
\If{$g_1 \neq g_2$}
    \State \Return \texttt{REJECT}
\EndIf
\State \Return \texttt{ACCEPT}
\end{algorithmic}
\end{algorithm}

We present some intuition for this in the combinatorial view before proceeding to the proof, which will mainly be algebraic. Each multi-linear polynomial $H_{n, p}$ is a counting function that counts unique subgraphs $S$ when given an edge list $F_{n, p}$ as input, with some weight multiplier for each $S$. Given that we know $H_{n, p} \in \mathbb{I}_{\MLP_{n, p}}$, $H_{n, p}$ counts, with some weight, unique copies of subgraphs that are isomorphic to each $S$, where $S$ is a subgraph of $K_n$, unique up to isomorphism.

On $n$ vertices, we can construct the following two multigraphs described by the edge lists $F_1$ and $F_2$. $F_1$ is a random graph on the set of vertices $[n] - [1]$ with each ``number of edges'' chosen from $\mathbb{Z}_p$. The degree of the vertex $1$ in this graph, the number of edges having non-zero value involving the vertex $1$ is $0$. In $F_2$, we start with $F_1$ and add $\lfloor n / 2 \rfloor - 2$ edges from the vertex $1$ to the set of vertices $[\lfloor n / 2 \rfloor - 1] - [1]$.

Our idea here is that $H_{n, p} \left( F_2 \right) - H_{n, p} \left( F_1 \right)$ is a counting function that counts exactly the subgraphs $H_{n, p}$ intends to count that involves the vertex $1$. $\HCL_{n, p} \left( F_2 \right) - \HCL_{n, p} \left( F_1 \right)$ is uniformly $0$ since the degree of the vertex $1$ (or any vertex for that matter) a subgraph counted by $\HCL_{n, p}$ is at least $\lfloor n / 2 \rfloor - 1$. Suppose that $H_{n, p} = \MLP_{A_{n, p}}$ and there is a graph $S$ with $a_{S} \neq 0$ and a vertex $k$ in $S$ such that $0 < \text{deg}(k) < \lfloor n / 2 \rfloor - 1$. Since $H_{n, p}$ is invariant under graph isomorphism, that means there is a graph $S^\prime$ such that $a_{S^\prime} \neq 0$ and $0 < \text{deg}(1) < \lfloor n / 2\rfloor - 1$ in $S^\prime$. In such a case, $H_{n, p} \left( F_2 \right) - H_{n, p} \left( F_1 \right)$ is not uniformly $0$, and due to the Schwartz-Zippel Lemma (\ref{lemma:1}), it is in fact, seldom $0$. Hence, with high probability, any $H_{n, p}$ with non-zero weights to subgraphs containing vertices of non-zero degree less than $\lfloor n / 2 \rfloor - 1$ fails the test. We can specifically ask the STV list decoder (Lemma \ref{lemma:2}) to list machines computing polynomials of degree at most $\displaystyle {{\lfloor n / 2 \rfloor} \choose 2}$. Taking this into account and the fact that every counted vertex has degree at least $\lfloor n / 2 \rfloor - 1$, which implicitly says that at least $\lfloor n / 2 \rfloor$ vertices are involved in each subgraph, each counted subgraph must exactly be a copy of $K_{\lfloor n / 2 \rfloor}$. Taking the invariance under graph isomorphism into account, unless $H_{n, p} = \alpha \HCL_{n, p}$ for some $\alpha \in \mathbb{Z}_p$, $H_{n, p}$ fails the test with high probability. It is easy to determine $\alpha$ by testing $M$ on a trivial $\lfloor n / 2 \rfloor$-clique graph. We now formalize this and prove it.

\begin{lemma}
\label{lemma:17}

Given an oracle machine $M$ promised to compute a polynomial $H_{n, p}: \mathbb{Z}_p^{n \choose 2} \to \mathbb{Z}_p \in \mathbb{I}_{\MLP_{n, p}}$ of degree at most $\displaystyle {{ \lfloor n / 2 \rfloor} \choose 2}$, we have a test such that:

\begin{enumerate}

\item If $H_{n, p} = \HCL_{n, p}$, then $M$ passes the test with a probability of $1$.

\item If $H_{n, p} \neq \HCL_{n, p}$, then $M$ passes the test with a probability at most $\displaystyle \displaystyle {n \choose 2} / p$.

\end{enumerate}

\end{lemma}

\begin{proof}

We can write $H_{n, p} \left( F_{n, p} \right)$ as
\begin{equation*}
H_{n, p} \left( F_{n, p} \right) = \sum_{\substack{S \subset F_{n, p} \\ e_{\{\, 1,j \,\}} \not \in S \text{ for any } j }} a_S H_{n, p}^S \left( F_{n, p} \right) + \sum_{\substack{S \subset F_{n, p} \\ e_{\{\, 1,j \,\}} \in S \text{ for some } j }} a_S H_{n, p}^S \left( F_{n, p} \right).
\end{equation*}

As per our construction of $F_1$ and $F_2$,
\begin{equation}
\label{eq:27}
H_{n, p} \left( F_2 \right) - H_{n, p} \left( F_1 \right) = \sum_{\substack{S \subset F_{n, p} \\ e_{\{\, 1,j \,\}} \in S \text{ for some } j}} a_S H_{n, p}^S \left( F_2 \right).
\end{equation}

The following claim will be useful to us.
\begin{claim}
\label{claim:1}

If $H_{n, p} \neq \alpha \HCL_{n, p}$ for any $\alpha \in \mathbb{Z}_p$, then there is an $S$ with $a_S \neq 0$ such that, in $S$, $0 < \deg(1) < \lfloor n / 2 \rfloor - 1$.

\end{claim}

\begin{proof}

Notice that if there is any non-empty $S^\prime \subset F_{n, p}$ at all with $a_{S^\prime} \neq 0$ and $S^\prime$ not isomorphic to $K_{\lfloor n / 2 \rfloor}$ with the number of edges in $S^\prime$ being at most $\displaystyle {{\lfloor n / 2 \rfloor} \choose 2}$, then there must be a vertex $k$ in $S^\prime$ with non-zero degree and degree less than $\lfloor n / 2 \rfloor - 1$. Suppose for the sake of contradiction that every vertex $k$ with non-zero degree must have degree at least $\lfloor n / 2 \rfloor - 1$. There are at least $\lfloor n / 2 \rfloor$ such vertices since the vertex $k$, on its own qualifies for this, and is adjacent to $\lfloor n / 2 \rfloor - 1$ other vertices with non-zero degree. This construction is exactly $K_{\lfloor n / 2 \rfloor}$, showing that $K_{\lfloor n / 2 \rfloor} \subset S^\prime$. If this containment were strict, then $S^\prime$ would have more than $\displaystyle {{\lfloor n / 2\rfloor} \choose 2}$ edges, contradicting the assumption that the degree of $H_{n, p}$ is at most $\displaystyle {{\lfloor n / 2 \rfloor} \choose 2}$. Due to Lemma \ref{lemma:14}, if any such $S^\prime$ with $a_{S^\prime} \neq 0$ exists with any of the non-zero degree vertices of degree less than $\lfloor n / 2 \rfloor - 1$, then there is an $S \sim S^\prime$ with the degree of the vertex $1$ less than $\lfloor n / 2 \rfloor - 1$.

\end{proof}

In lieu of this, $H_{n, p} \left( F_2 \right) - H_{n, p} \left( F_1 \right)$ is uniformly $0$ if and only if $H_{n, p}$ is a multiple of $\HCL_{n, p}$. It is easy to see that if $H_{n, p}$ is a multiple of $\HCL_{n, p}$, then $H_{n, p} \left( F_2 \right) - H_{n, p} \left( F_1 \right)$ is uniformly $0$, since the $S$ in Equation \eqref{eq:27} is a clique of size $\lfloor n / 2 \rfloor$. If the polynomial $H_{n, p}$ is not a multiple of $\HCL_{n, p}$ and $H_{n, p} \in \mathbb{I}_{\MLP_{n, p}}$, then due to Claim \ref{claim:1}, there is an $S$ such that the degree of the vertex $1$ in $S$ is non-zero, less than $\lfloor n / 2 \rfloor - 1$, and $a_S \neq 0$. Due to the invariance of $H_{n, p}$ under graph isomorphism, we can pick $S$ such that all edges from vertex $1$ are adjacent to the set of vertices $[\lfloor n / 2 \rfloor - 1] - [1]$. Hence, Equation \eqref{eq:27} is not uniformly $0$. Due to the Schwartz-Zippel Lemma (\ref{lemma:1}), we have 
\begin{equation*}
\mathcal{P}_{(F_1, F_2) \gets_r \mathcal{D}} \left[ H_{n, p} \left( F_2 \right) - H_{n, p} \left( F_1 \right) = 0 \right] \leq \frac{\displaystyle {{\lfloor n / 2 \rfloor} \choose 2}}{p},
\end{equation*}
where $\mathcal{D}$ is the probability distribution in Algorithm \ref{alg:6}. Once $M$ passes the test, we can be reasonably sure that $H_{n, p} = \alpha \HCL_{n, p}$. To check whether $\alpha = 1$, one can test $M$ with $F_{n, p}$ set to the simple complete graph $K_{\lfloor n / 2 \rfloor}$ on the set of vertices $[\lfloor n / 2 \rfloor]$.

\end{proof}

Once again, we can reduce the error probability to exponentially low merely by repetition.

We can state the following lemma, that follows from Lemmas \ref{lemma:15}, \ref{lemma:16}, and \ref{lemma:17}, and the repetition of each of these tests polynomially many times.

\begin{lemma}
\label{lemma:18}

There is a test such as that given an oracle machine $M$ promised to compute a polynomial $H_{n, p}: \mathbb{Z}_p^{n \choose 2} \to \mathbb{Z}_p$ of degree at most $\displaystyle {{\lfloor n / 2 \rfloor} \choose 2}$, with $p > n(n-1)$:

\begin{enumerate}

\item If $H_{n, p} = \HCL_{n, p}$, then $M$ passes the test with a probability of $1$.

\item If $H_{n, p} \neq \HCL_{n, p}$, then $M$ passes the test with a probability of at most $2^{-n^2}$

\end{enumerate}

\end{lemma}

\subsection{Main Results for Half-Clique Counting}
\label{section:7.5}

We are now ready to state the main results on the half-clique counting problem modulo $p$.
\begin{theorem}
\label{theorem:8}

Given a prime $p = 2^{n^{O(1)}} > n(n-1)$, and an oracle $O$ counting the number of unique cliques of size $\lfloor n / 2 \rfloor$ modulo $p$ correctly on more than an $\epsilon$-fraction of instances (for any $\epsilon = 1 / 2^{o(n)}$) on $n$-vertex multigraphs, where the number of edges varies from $0$ to $p - 1$, then in $\poly(n, \log p, 1/\epsilon)$ time and queries to $O$, we have a probabilistic algorithm that computes the number of unique cliques of size $\lfloor n / 2 \rfloor$ modulo $p$ with an error probability of less than $2^{-n}$.

\end{theorem}

\begin{proof}

Similar to the proof of Theorem \ref{theorem:6}, we ask the STV list decoder (Lemma \ref{lemma:2}) to return polynomials of degree at most $\displaystyle {{\lfloor n / 2 \rfloor} \choose 2}$. The STV list decoder returns $O(1/\epsilon)$ oracle machines $M$, each of which computes a polynomial $H_{n, p}$ according to our requirements. One of these machines compute $\HCL_{n, p}$ with high probability. Using the tests outlined above, we can identify $M$ and start using it for our algorithm. The union bound implies that the sum of all errors possible is
\begin{equation*}
\frac{O(1 / \epsilon)}{2^{n^2}} = \frac{2^{o(n)}}{2^{n^2}} = \frac{1}{2^{\omega(n)}} < \frac{1}{2^n}.
\end{equation*}

\end{proof}

We show the following corollary of this theorem.
\begin{theorem4}
\label{cor:3}

\textbf{Theorem \ref{thm:4} Restated.} \\
If there is an algorithm computing the number of half-cliques modulo $p$ on $n$-vertex multigraphs for a prime $p = \Theta \left( 2^{n^{c}} \right)$ on a $\displaystyle 1 / g(n)$-fraction of instances in $2^{o(n)}$-time for any $c \geq 1$, and $g(n) = 2^{o(n)}$, then $\rETH$ is false.

\end{theorem4}

\begin{proof}

Due to Theorem \ref{theorem:8}, if such an algorithm exists, then in $2^{o(n)}$ time and queries to the algorithm, we can compute the number of half-cliques modulo $p$ with high probability. Due to the randomized reduction from $\HALF$ to $\HALF_p$ (Lemma \ref{lemma:8}) and the $\ETH$-hardness of $\HALF$ (Appendix  \ref{appendix:A}), we would have a $2^{o(n)}$-time randomized algorithm for $3\SAT$, violating $\rETH$.

\end{proof}

On the other hand, note that it is easy to have
\begin{equation*}
\frac{1}{2^{O(n)}} = \frac{1}{p}-o \left( \frac{1}{p^{1.999}} \right)
\end{equation*}
correctness if $p = 2^{cn}$, simply by returning $0$ without even reading the input. We represent $\HCL_{n, p} \left( F_{n, p} \right)$ as
\begin{equation*}
\HCL_{n, p} \left( F_{n, p} \right) = H_1 \left( F_{n, p} \right) + e_{\{\, 1, 2 \,\}} H_2 \left( F_{n, p} \right),
\end{equation*}
where $H_1 \left( F_{n, p} \right)$ and $H_2 \left( F_{n, p} \right)$ are independent of the variable $e_{\{\, 1, 2 \,\}}$. Choosing $F_{n, p} - \{\, e_{\{\, 1, 2 \,\}} \,\}$ at random from $\mathbb{Z}_p^{\binom{n}{2} - 1}$, $H_2 \left( F_{n, p} \right)$ is non-zero with a probability of at least $1 - \left( \displaystyle {n \choose 2} - 1 \right) / p$, due to the Schwartz-Zippel Lemma (\ref{lemma:1}), and there is a unique $e_{\{\, 1,2 \,\}} \in \mathbb{Z}_p$ such that $\HCL_{n, p} \left( F_{n, p} \right) = 0$. In view of this, $\HCL_{n, p} \left( F_{n, p} \right)$ has at least $\left( 1 - \left( \displaystyle {n \choose 2} - 1 \right) / p \right) p^{\binom{n}{2} - 1}$ roots, meaning that for at least a
\begin{equation*}
\frac{1 - \displaystyle \frac{\displaystyle {n \choose 2} - 1}{p}}{p} = \frac{1}{p} - o \left( \frac{1}{p^{1.999}} \right) = \Omega \left( \frac{1}{2^{cn}} \right)
\end{equation*}
fraction of instances $F_{n, p}$, $\HCL_{n, p} \left( F_{n, p} \right) = 0$.\footnote{Note that this was this easy for us to show since $\HCL_{n, p}$ is multilinear. If every variable had degree at least two, we might have had to use heavy tools from algebraic geometry such as the Lang-Weil bound \citep{LW54}.}

In other words, assuming $\rETH$, any algorithm that has even a marginally smaller exponent than the worst-case algorithm, cannot hope to be correct on even a $1 / 2^{o(n)}$-fraction of instances, while an algorithm that refuses to read the input and simply outputs $0$ is correct at least on a
\begin{equation*}
\frac{0.99}{p} = \Omega \left( \frac{1}{2^{cn}} \right)
\end{equation*}
fraction of instances.

\section{Hardness Amplification for Parity $k$-Clique Counting}
\label{section:8}

In this section, we consider the question of average-case hardness for counting problems on simple graphs. Consider the problem of deciding the parity of the number of half-cliques in a simple $n$-vertex undirected graph $U_n$. This is the same as evaluating the polynomial $\HCL_{n, 2}$ over $\mathbb{Z}_2$. Tools such as the STV list decoder (Lemma \ref{lemma:2}) and the Schwartz-Zippel Lemma (\ref{lemma:1}) are powerless in this scenario since these tools inherently require that the size of the field is large. Such questions have been considered by \cite{Boix2019}. In the case of simple graphs, they show that with a $k-\CLIQUE$ counting oracle that has an error fraction that is at most $O \left( 1 / (\log k)^{k \choose 2} \right)$, we can have a quick probabilistic algorithm using essentially the runtime of the oracle multiplied by a factor of $O \left( (\log k)^{k \choose 2} \right)$. \cite{Goldreich2020} improves upon the error rate, providing a reduction requiring $O(n^2)$ time, $e^{O(k^2)}$ queries, and the oracle is correct on a $1-2^{-k^2}$ fraction of instances.

We recall the definition of the random experiment $O^{H_n}_{c}$ from Subsection \ref{section:1.1.2}.

\begin{definition1}
\label{def:23}

Given any set $\mathbb{D}$ and a function $H_n: \{\, 0,1 \,\}^{n \choose 2} \to \mathbb{D}$ defined over $n$-vertex simple undirected graphs that is invariant under graph isomorphism, the random experiment $O^{H_n}_{c}$ selects a set $S \subset \{\, 0,1 \,\}^{n \choose 2}$ of size $c 2^{n \choose 2}$ with uniform probability and gives an oracle $O$ that correctly answers queries for computing $H_n$ on the set $S$. The other answers of $O$ can be selected adversarially, randomly, or to minimize the time complexity,  $T_O$, of the fastest deterministic algorithm implementing it.

\end{definition1}

For $\epsilon > 0$, we show that no matter how these values in $\{\, 0,1 \,\}^{n \choose 2} - S$ are chosen, with a probability of greater than $1 - 2^{-\Omega(n^2)}$, the oracle $O$ sampled from $O^{H_n}_{1/2 + \epsilon}$ can be error-corrected to provide correct answers to $H_n$ for a simple $n$-vertex undirected graph $U_n$ when $U_n$ is not a highly symmetric graph, using $O(1 / \epsilon^2)$ queries to $O$. When we say ``not a highly symmetric graph'', we mean that the automorphism group of $U_n$, $\Aut \left( U_n \right)$ is of size at most $n! \epsilon^2 / n^3$. It was proved by \cite{Polya1937} and \cite{Erdos1963} that at least a $1-2^{-0.99n}$-fraction of $n$ vertex simple undirected graphs have a trivial automorphism group. When the graph is highly symmetric, with a very large set of symmetries, one can imagine that computing isomorphism-invariant functions\footnote{This is true for at least such as $k-\CLIQUE$ where the difficulty of the problem increases with the decrease in symmetry in the graph.} on highly symmetric graphs is easy and can be done fairly quickly. We show that this is the case for counting $k$-cliques on highly symmetric graphs in time polynomial in $n$ and the size of the isomorphism class of $U_n$. We show that paired with a quick test to see whether a graph $U_n$ has a large automorphism group, from almost all oracles $O$ sampled from $O^{H_n}_{1 / 2 + \epsilon}$, where $H_n$ computes the number of $k$-cliques (or the parity, or detects them), we have a probabilistic algorithm computing $H_n$ in $O \left( \left( n^{8 + o(1)}/\epsilon^{2 + o(1)} + T_O \right) /\epsilon^2 \right)$-time.

A rough intuitive sketch of the algorithm is:

\begin{enumerate}

\item Check whether the automorphism group of $U_n$ is ``large.''

\item If it indeed is ``large,'' since it is a highly symmetric graph, we compute $H_n$ very quickly by listing all graphs isomorphic to $U_n$ and utilizing the symmetry to reduce the number of redundant computations. In this case, it is helpful that the isomorphism class of $U_n$ is small.

\item If the automorphism group is not ``large,'' then the isomorphism class of $U_n$ is sufficiently large, and we can use our error correction procedure. The procedure is straightforward: Choose a random graph isomorphic to $U_n$ and query $O$ on that graph. The probabilistic method guarantees that all sufficiently large isomorphism classes get sufficient correctness to amplify from with high probability over the randomness of $O^{H_n}_{1/2+\epsilon}$.

\end{enumerate}

For the specialized case for constant $k > 2$ and constant $\epsilon > 0$, to reduce the time complexity of the reduction to $\tilde{O} \left( n^2 \right)$, for sufficiently large $n$, we provide a classification of all graphs with an automorphism group of size $\omega \left( n! / n^3 \right)$. There are twelve highly symmetric cases, and in $\tilde{O} \left( n^2 \right)$-time, we can not only check if a graph is isomorphic to one of these cases but also compute the number of $k$-cliques. When the automorphism group is of size $O \left( n! / n^3 \right)$, we have good guarantees for the error correction procedure. What follows from this is a $\tilde{O}(n^2)$-time randomized reduction to almost any oracle that is correct for parity-counting on a $1 / 2 + \epsilon$-fraction of instances for $\epsilon > 0$. In Section \ref{section:8.5}, we discuss how this partially addresses an open question of \cite{Goldreich2020} and a possibly fruitful direction toward full resolution.

The following is the rough organization of this result. In Subsection \ref{section:8.1}, we discuss the error correction procedure. In Subsection \ref{section:8.2}, we provide a quick approximate test to check if a graph has a very large automorphism group and a specialized result for the case where the automorphism graph is of size $\omega \left( n! / n^3 \right)$. In Subsection \ref{section:8.3}, we provide a quick algorithm to count the number of $k$-cliques on these highly symmetric graphs with very large automorphism groups. In Subsection \ref{section:8.4}, we combine these results to show that most oracles from the experiment $O^{H_n}_{1 / 2 + \epsilon}$ provide a good probabilistic algorithm when $H_n$ is isomorphism-invariant and trivial to compute from the count of $k$-cliques. 

\subsection{Error Correction via Querying for Isomorphic Graphs}
\label{section:8.1}

Our method of error correction is straightforward. As shown in Algorithm \ref{alg:7}, given an $n$-vertex undirected simple graph $U_n$, we simply permute the vertices and edges based on a random permutation $\pi$ sampled uniformly from $S_n$ to obtain $U_n^\prime$. We query the oracle $O$ on $U_n^\prime$ and output that answer. If we want a higher probability of correctness, we repeat this many times and output the answer we received most of the time.

\begin{algorithm}
\caption{$H_n\text{-Amplifier}\left(O, n, U_n \right)$}
\label{alg:7}
\Comment{The $H_n$-Amplifying Algorithm \hspace{10.5cm}} \\
\Comment{Input: $O$ is an average-case oracle for $H_n: \{\, 0, 1 \,\}^{n \choose 2} \to \mathbb{D}$, where $H_n$ is a function that is invariant under graph isomorphism} \\
\Comment{Input: $n$ is the number of vertices in the undirected simple graph $U_n$} \\
\Comment{Input: $U_n$ is a formal list of edges for an undirected simple graph with $n$ vertices} \\
\Comment{Output: $H_n \left( U_n \right)$}
\begin{algorithmic}
\State $\pi \gets_{r} S_n$
\State $U^\prime_n \gets \pi \left( U_n \right)$ 
\Comment{Using Definition \ref{def:18}}
\State \Return $O \left( U^\prime_n \right)$ 
\end{algorithmic}
\end{algorithm}

If we randomly sample from $O^{H_n}_{1/2+\epsilon}$, due to the tail bounds, any sufficiently large subset of $\{0,1\}^{n \choose 2}$ is very likely to have least a $1 / 2 + \epsilon / 2$-fraction of correctness over $O$ with high probability over the randomness of $O^{H_n}_{1 / 2 + \epsilon}$. The following lemma shows that this is true for any sufficiently large isomorphism classes of graphs, but it is true for any subsets of sufficiently large size.

\begin{lemma}
\label{lemma:19}

Given an oracle $O$ randomly sampled from $O^{H_n}_{1 / 2 + \epsilon}$, for any $\epsilon > 0$, any isomorphism class $\mathcal{C}_n$ of $n$-vertex undirected simple graphs has over a $1 / 2 + \epsilon / 2$-fraction of correctness in $O$ with a probability of greater than $1 - e^{-\epsilon^2 |\mathcal{C}_n| / 8}$.

\end{lemma}

\begin{proof}

Consider the random variables $\left( X_i \right)_{i \in [|\mathcal{C}_n|]}$. The random variable $X_i$ attains the value of $1$ if the oracle $O$ is correct on $U_n^i \in \mathcal{C}_n$, and $0$ otherwise. $X = \sum_{i \in [|\mathcal{C}_n|]}X_i$ is the random variable counting $O$'s correct instances over $\mathcal{C}_n$.

Using Definition \ref{def:23}, from the symmetry of the distribution in $O^{H_n}_{1 / 2 + \epsilon}$, we have
\begin{equation}
\label{eq:28}
\frac{1}{2} + \epsilon \leq \mathcal{E} \left[ X_i \right] \leq 1.
\end{equation}

Applying the linearity of expectation in Equation \eqref{eq:28}, we get 
\begin{equation}
\label{eq:29}
\left( \frac{1}{2} + \epsilon \right) |\mathcal{C}_n| \leq \mathcal{E}[X] \leq \left| \mathcal{C}_n \right|.
\end{equation} 

Since all the variables $X_i$ are negatively correlated in the worst-case, using Equation \eqref{eq:29}, we can use the following Chernoff bound \citep{Mitzenmacher2005} to tail-bound the random variable $X$:
\begin{equation}
\label{eq:30}
\mathcal{P}_{O \gets_r O^{H_n}_{1 / 2 + \epsilon}} \left[ X \leq (1 - \delta) \mathcal{E}[X] \right] \leq e^{-\delta^2 \mathcal{E}(X) / 2}.
\end{equation}

In particular, taking the value of
\begin{equation*}
1 - \delta = \frac{1 + \epsilon}{1 + 2 \epsilon} \iff \delta = \frac{\epsilon}{1 + 2 \epsilon}
\end{equation*}
in Equation \eqref{eq:30}, and using Equation \eqref{eq:29}, we get
\begin{equation*}
\mathcal{P}_{O \gets_r O^{H_n}_{1 / 2 + \epsilon}} \left[ X \leq \left( \frac{1}{2} + \frac{\epsilon}{2} \right) \left| \mathcal{C}_n \right| \right] \leq e^{-\epsilon^2 \left| \mathcal{C}_n \right| / 4 (1 + 2 \epsilon)} \leq e^{-\epsilon^2 \left| \mathcal{C}_n \right| / 8}.
\end{equation*}

Hence, with high probability over the randomness of $O^{H_n}_{1/2+\epsilon}$, $O$ is sufficiently correct for large isomorphism classes $\mathcal{C}_n$.

\end{proof}

Now, simply by the union bound, we can show that all sufficiently large isomorphism classes have a sufficient amount of correctness from $O$ with high probability over the randomness of $O^{H_n}_{1 / 2 + \epsilon}$. Since almost all graphs with $n$ vertices have the maximum isomorphism class size, with high probability over the randomness of $O^{H_n}_{ 1 / 2 + \epsilon}$, this error correction technique is successful for almost all graphs.

\begin{lemma}
\label{lemma:20}

Given an oracle $O$ randomly sampled from $O^{H_n}_{1 / 2 + \epsilon}$, for any $\epsilon = \omega \left( n^{3 / 2} / \sqrt{n!} \right)$ (or a constant $\epsilon > 0$), with a probability of $1 - o(1)$ over the randomness of $O^{H_n}_{1 / 2 + \epsilon}$, we have a probabilistic algorithm solving $H_n$ correctly on a $1 - o(1)$-fraction of instances with a $1 / 2 + 1 / \poly(n)$ probability of correctness on the correct instances with $O \left( 1 / \epsilon^2 \right)$ ($1$ for constant $\epsilon$) queries to $O$. With $\poly(n)$ many repetitions of the algorithm, we can reduce the error on these correct instances to $2^{-\poly(n)}$.

\end{lemma}

\begin{proof}

From Lemma \ref{lemma:19}, we know that given an oracle $O$ randomly sampled from $O^{H_n}_{1 / 2 + \epsilon}$, with a probability of greater than $1 - e^{\epsilon^2 |\mathcal{C}_n| / 8}$, $O$ is correct on at least a $ 1 / 2 + \epsilon / 2$-fraction of instances. Given a graph $U_n^i \in \mathcal{C}_n$, with a probability of greater than $1 - e^{\epsilon^2 |\mathcal{C}_n| / 8}$ (over the randomness of $O^{H_n}_{1 / 2 + \epsilon}$), $H_n\text{-Amplifier}\left(O, n, U_n \right)$ (Algorithm \ref{alg:7}) computes $H_n \left( U_n^i \right)$ correctly with at least a $1 / 2 + \epsilon / 2$ probability of correctness. This is due to the invariance of $H_n$ under graph isomorphism. In fact, by repeating this process $O \left( 1 / \epsilon^2 \right)$ times and returning the majority outcome as the answer, due to the Chernoff bound, once again, we get the correct answer with exponentially small error probability.

We can show that with high probability (over the distribution $O^{H_n}_{1 / 2 + \epsilon}$), each isomorphism class $\mathcal{C}_n$ with $\left| \mathcal{C}_n \right| > n^3 / \epsilon^2$ has over a $1 / 2 + \epsilon / 2$-fraction of correctness over $O$. Due to the union bound and Lemma \ref{lemma:19}, the probability that at least one isomorphism class $\mathcal{C}_n$ of size at least $100n^2 / \epsilon ^ 2$ is correct on less than a $1 / 2 + \epsilon/2$-fraction of instances is 
\begin{equation*}
\begin{split}
\sum_{\left| \mathcal{C}_n \right| > 100 n^2 / \epsilon^2} \frac{1}{e^{\epsilon^2 \left| \mathcal{C}_n \right| / 8}} & \leq \sum_{\left| \mathcal{C}_n \right| > 100 n^2 / \epsilon^2} \frac{1}{2^{\epsilon^2 \left| \mathcal{C}_n \right| / 8}} \\
& \leq \sum_{\left| \mathcal{C}_n \right| > 100 n^2 / \epsilon^2} \frac{1}{2^{100 n^2 / 8}} \\
& \leq \frac{2^{n \choose 2}}{2^{100n^2/8}} \\
& = O \left( 2^{-\Omega \left( n^2 \right)} \right),
\end{split}
\end{equation*}
since there are at most $2^{n \choose 2}$ such graphs with $n$ vertices, we have at most that many isomorphism classes. Hence, we have shown that all large isomorphism classes get sufficient coverage with high probability. It was proved by \cite{Polya1937} and \cite{Erdos1963} that with a probability of $1 - {n \choose 2} 2^{-n - 2}(1 + o(1))$, a random undirected simple graph $U_n$ of $n$ vertices has a trivial automorphism, group $\Aut \left( U_n \right)$. Since the isomorphism class $\mathcal{C}_n$ can be seen as ``isomorphic'' to $S_n / \Aut \left( U_n \right)$, with a high probability over the random choices of $U_n$, the isomorphism class $\mathcal{C}_n$ of $U_n$ has a size
\begin{equation*}
\left| \mathcal{C}_n \right| = \left| S_n \right| / |\Aut \left( U_n \right)| = n!.
\end{equation*}

Since all isomorphism classes of size greater than $100 n^2 / \epsilon^2 = o(n!)$, and subsequently, all isomorphism classes of size $n!$ are covered with a high probability, with a probability of $1 - o(1)$ over $O^{H_n}_{1 / 2 + \epsilon}$, Algorithm \ref{alg:7} (or the repeat-majority version of it) is probabilistically correct on a $1 - O\left( \displaystyle {n \choose 2} 2^{-n} \right)$-fraction of instances.

\end{proof}

\subsection{Testing the Size of $\Aut \left( U_n \right)$}
\label{section:8.2}

Now that we have shown that for undirected simple graphs with large isomorphism classes (equivalently, ``slightly smaller'' automorphism groups than $S_n$), we can error-correct; we must know when we can and cannot use the error correction procedure. Algorithm \ref{alg:8} tests whether the automorphism group of $U_n$ is larger than $n! / t(n)$ or smaller than $n! / t(n)^{1 + o(1)}$ for any $t(n) = \omega(1)$. This algorithm estimates $Aut \left( U_n \right) n t(n) / n!$ and given that we have given enough independent samples, the bulk of the weights of the normal distribution from the central limit theorem viewpoint are on either side of the $n / 2$-point in either case.

\begin{algorithm}
\caption{$\text{Aut-Size-Test} \left( n, U_n, t(n) \right)$}
\label{alg:8}
\Comment{The Automorphism-Size-Test Algorithm \hspace{8.8cm}} \\
\Comment{Input: $n$ is the number of vertices in the undirected simple graph $U_n$} \\
\Comment{Input: $U_n$ is a formal list of edges for an undirected simple graph with $n$ vertices} \\
\Comment{Input: $t(n):\mathbb{N}\to\mathbb{N}$ is an $\omega(1)$ function}\\
\Comment{Output: \texttt{ACCEPT} if $ \left| \Aut \left( U_n \right) \right| \geq \displaystyle \frac{n!}{t(n)}$, \texttt{REJECT} if $ \left| \Aut \left( U_n \right) \right| \leq \displaystyle \frac{n!}{t(n)^{1+\alpha}}$ for any constant $\alpha > 0$}
\begin{algorithmic}
\State $I \gets 0$
\State $C \gets 0$
\While{$I < t(n)n$}
    \State $\pi \gets_r S_n$
    \If{$\pi \left( U_n \right) = U_n$}
        \State $C \gets C + 1$
    \EndIf
    \State $I \gets I + 1$
\EndWhile
\If{$C < \displaystyle \frac{n}{2}$}
    \State \Return \texttt{REJECT}
\EndIf
\State \Return \texttt{ACCEPT}
\end{algorithmic}
\end{algorithm}

We now prove the correctness and time complexity of this algorithm.

\begin{lemma}
\label{lemma:21}

Given an undirected simple graph $U_n$ with $n$ vertices, for every constant $k > 0$, we have an $O \left( t(n)n^{3} \right)$-time algorithm such that given a function $t(n): \mathbb{N} \to \mathbb{N}$ such that $t(n) = \omega(1)$:

\begin{enumerate}

\item If $\left| \Aut \left( U_n \right) \right| \geq n! / t(n)$, then the algorithm accepts with a probability greater than $1 - e^{n / 4}$.

\item If $\left| \Aut \left( U_n \right) \right| \leq n! / t(n)^{1 + \alpha}$ for any constant $\alpha > 0$, then the algorithm rejects with a probability of at least $1 - 2^{n / 2}$.
  
\end{enumerate}

\end{lemma}
\begin{proof}

Let us prove the first assertion. If $\left| \Aut \left( U_n \right) \right| \geq n! / t(n)$, then the probability that a randomly sampled $\pi \in S_n$ belongs to $\Aut \left( U_n \right)$ is at least $1 / t(n)$. Due to the linearity of expectation, the variable $C$ in Algorithm \ref{alg:8} has an expectation of
\begin{equation}
\label{eq:31}
\mathcal{E}[C] \geq n. 
\end{equation}
Since the sampling is independent, due to the Chernoff bound \citep{Mitzenmacher2005}, and Equation \eqref{eq:31}, the probability of rejection is exponentially small:
\begin{equation*}
\begin{split}
\mathcal{P} \left[ C < \frac{n}{2} \right] & \leq \mathcal{P} \left[ C < \frac{\mathcal{E}(C)}{2} \right] \\
& < e^{-\mathcal{E}(C) / 4} \\
& \leq e^{-n / 4}.
\end{split}
\end{equation*}

On the other hand, suppose that $\left| \Aut \left( U_n \right) \right| \leq n! / t(n)^{1 + \alpha}$. The probability that a randomly sampled permutation $\pi \in S_n$ belongs to $\Aut \left( U_n \right)$ is at most $1 / t(n)^{1 + \alpha}$. Hence, due to the linearity of expectation, we have
\begin{equation}
\label{eq:32}
\mathcal{E}[C] \leq \frac{n}{t(n)^{\alpha}}.
\end{equation}
Due to the Chernoff bound  \citep{Mitzenmacher2005}, once again, and using Equation \eqref{eq:32}, the probability of acceptance is exponentially small for sufficiently large $n$, and depending on $t(n)$ and $\alpha$:
\begin{equation*}
\mathcal{P} \left[ C \geq \frac{n}{2} \right] \leq 2^{-n / 2}.
\end{equation*}
In this step, it is crucial that $t(n)$ is eventually larger than every constant, otherwise $\alpha$ would have to be lower bounded.

Algorithm \ref{alg:8} also runs in $O \left( n^3 t(n) \right)$-time since the while loop runs for $n t(n)$ times, the complexity of sampling from $S_n$ is $O(n \log n)$ and the operations on $U_n$ take $O \left( n^2 \right)$-time.

\end{proof}

\subsubsection{A Classification of Graphs With $\left| \Aut \left( U_n \right) \right| = \omega \left( n! / n^3 \right)$}
\label{section:8.2.1}

In this subsection, first we will prove Lemma \ref{lemma:22} that gives the properties of the graphs with $\left| \Aut \left( U_n \right) \right| = \omega \left( n! / n^3 \right)$ in terms of the number of partitions based on the degree of vertices, the size of such partitions, and the degree distribution of the vertices.

\begin{lemma}
\label{lemma:22}

For sufficiently large $n$, any graph $U_n$ with $\left| \Aut \left( U_n \right) \right| = \omega \left( n! / n^3 \right)$ satisfies the following properties.

\begin{enumerate}

\item If the vertices of $U_n$ are partitioned based on degree, then there are at most three partitions.

\item No partition can be simultaneously larger than $2$ and smaller than $n-2$.

\item The degree of any vertex $v$ can be in the set $\{\, 0, 1, 2, n-2, n-1 \,\}$.

\end{enumerate}

\end{lemma}

\begin{proof}

Suppose that we have $m \geq 4$ partitions of size $(\alpha_i)_{i \in [m]}$, in ascending order, with each $\alpha_i \geq 1$. The probability that $\pi \in S_n$ is in $\Aut \left( U_n \right)$ is bounded from above by
\begin{equation*}
\frac{\prod_{i \in [m]} \alpha_i!}{n!} \leq \frac{(\sum_{i \in [m-1]} \alpha_i)! \alpha_m!}{n!},
\end{equation*}
since $\pi$ is not allowed to permute vertices across partitions. Since $\alpha_m \geq n / m$ due to the pigeonhole principle, and $\alpha_m \leq n - (m - 1)$ due to each $\alpha_i$ being positive, we have that
\begin{equation*}
\frac{\left( \sum_{i \in [m-1]} \alpha_i \right)! \alpha_m!}{n!} = \frac{1}{\displaystyle {n \choose \alpha_m}} \leq \frac{1}{\displaystyle {n \choose 3}} = O \left( \frac{1}{n^3} \right).
\end{equation*}
If this is the case, then $\left| \Aut \left( U_n \right) \right|$ is upper bounded by $O\left( n! / n^3 \right)$, leading to a contradiction. This proves the first statement of the lemma.

If there is a partition of size $\alpha$, then the probability that $\pi \in S_n$ is in $\Aut \left( U_n \right)$ is upper bounded by $1 / \displaystyle {n \choose \alpha} = O \left( 1 / n^3 \right)$ for the forbidden range. This implies the second statement of the lemma.

Let us consider that the degree $m$ of $v$ in $U_n$ is greater than $2$ and less than $n-2$. Let the neighbors of $v$ be $(u_i)_{i \in [m]}$. The probability that $\pi \in S_n$ is in $\Aut \left( U_n \right)$ is bounded from above by
\begin{equation*}
\frac{n (m!) (n-1-m)!}{n!} = \frac{1}{\displaystyle {{n-1}\choose m}} \leq \frac{1}{\displaystyle {{n -1} \choose 3}} = O(1/n^3),
\end{equation*}
since $n$ is the maximum number of vertices $v$ could map to, $m!$ is the number of ways the neighbors of $v$ could distribute themselves among the neighbors of the image of $v$, and $(n-m-1)!$ is the number of ways the remaining vertices can distribute. Also, $m$ is between $3$ and $n-3$. Due to a similar argument as before, this implies the third statement of the lemma.

\end{proof}

Now, using Lemma \ref{lemma:22}, we will prove Lemma \ref{lemma:23} that gives the structure of the graphs with $\left| \Aut \left( U_n \right) \right| = \omega \left( n! / n^3 \right)$.

\begin{lemma}
\label{lemma:23}

Only the following graphs have $\left| \text{Aut} \left( U_n \right) \right| = \omega \left( n! / n^3 \right)$.

\begin{enumerate}

\item $K_n$ and its complement.

\item $K_n$ with one edge missing and its complement.

\item $K_{n - 1}$ with an isolated vertex and its complement.

\item $K_{n - 1}$ with one vertex of degree $1$ adjacent to it and its complement.

\item $K_{n - 2}$ with two isolated vertices and its complement.

\item $K_{n - 2}$ with two vertices of degree $1$ adjacent to each other and its complement.

\end{enumerate}

\end{lemma}

\begin{proof}

Using Lemma \ref{lemma:23}, the only possible partition sizes based on degree we can have are $(n)$, $(n-1, 1)$, $(n-2, 1, 1)$ and $(n-2, 2)$. Now, by a case-by-case analysis, we will determine which graphs can have such large automorphism groups. Since $\Aut \left( U_n \right) = \Aut \left( \overline{U_n} \right)$, we will categorize by the degree of the largest partition and assume that the degree is less than or equal to $2$. This way, we will either allow a graph and its complement or reject both. We will also assume that $n$ is sufficiently large, say $n \geq 100$.

\textbf{\textit{Case 1: The Largest Partition Degree is $0$.}} \\
Now, for the $(n)$ partition, the graph is either empty or the complete graph $K_n$. Clearly,
\begin{equation*}
\left| \Aut \left( U_n \right) \right| = n! = \omega \left( \frac{n!}{n^3} \right),
\end{equation*}
in both cases, so we allow both.

When the partition is $(n-1, 1)$, this is technically not allowed since even the vertex of the partition of size $1$ must have degree zero, meaning such a partition with these degrees cannot exist.

When the partition is $(n-2, 1, 1)$, this cannot exist since the partitions of size $1$ must have the same degree.

When the partition is $(n-2, 2)$, the only allowed case is that both the vertices in the partition of size $2$ are adjacent. Otherwise, they would also have degree $0$, and we would have $(n)$ again. The other case is $K_n$ with one edge missing. Both of them have
\begin{equation*}
\left| \Aut \left( U_n \right) \right| = 2 (n-2)! = \frac{n!}{O(n^2)} = \omega \left( \frac{n!}{n^3} \right),
\end{equation*}
hence, we allow them both.

From this case, we allow the graphs as described in statements 1 and 2 of the lemma.

\textbf{\textit{Case 2: The Largest Partition Degree is $1$.}} \\
For the partition type $(n)$, this is only allowed when $n$ is even due to the handshake lemma. When so, the vertices arrange themselves in pairs. Visually, we have $n / 2$ ``sticks''. We can permute these sticks in $(n / 2)!$ ways and flip them in $2^{n / 2}$ ways. In particular, the size of the automorphism group is
\begin{equation*}
\left| \Aut \left( U_n \right) \right| = \left( \frac{n}{2} \right)! \cdot 2^{n / 2} \leq \frac{n!}{n^3} = O \left( \frac{n!}{n^3} \right),
\end{equation*}
for sufficiently large $n$. Hence, we reject this case.

For the partition type $(n-1, 1)$, we have the following possibilities: The vertex in the partition of size $1$ may have possible degrees $n-1$, $n-2$, $2$, or $0$.

\begin{enumerate}

\item The vertices in the $n - 1$-partition are all adjacent to the vertex in the $1$-partition. This is allowed, with
\begin{equation*}
\left| \Aut \left( U_n \right) \right| = (n - 1)! = \frac{n!}{n} = \omega \left( \frac{n!}{n^3} \right),
\end{equation*}
and hence, we allow $K_{n-1}$ with an isolated vertex and its complement graph. This covers the case where the $1$-partition vertex has degree $n-1$.

\item If the degree of the $1$-partition vertex is $n-2$, this is disallowed for the following reason: The vertex in the $n-1$-partition not adjacent to the $1$-partition vertex must be adjacent to one of the other vertices, if it needs a degree of $1$. This creates a vertex of degree $2$ in the $n-1$-partition.

\item The degree of the $1$-partition vertex is $2$. In this case, we have two vertices $u$ and $v$ in the $n - 1$-partition that are adjacent to the $1$-partition vertex. The others are arranged similarly to the $(n)$ case for degree $1$. Here, the automorphism group size is
\begin{equation*}
\left| \Aut \left( U_n \right) \right| = 2 \left( \frac{n - 3}{2} \right)! \cdot 2^{(n - 3) / 2} = O \left( \frac{n!}{n^3} \right),
\end{equation*}
which for sufficiently large $n$ is too small; hence we reject this case when $n$ is odd. The graph is not possible when $n$ is even

\item If the degree of the $1$-partition vertex is $0$ and the others have degree $1$, this suffers from the same pitfalls as the $(n)$ case, having an automorphism group of size
\begin{equation*}
\left| \Aut \left( U_n \right) \right| = \left( \frac{n - 1}{2} \right)! \cdot 2^{(n - 1) / 2} = O \left( \frac{n!}{n^3} \right),
\end{equation*}
and hence, we reject this case as well when $n$ is odd. The graph is not possible when $n$ is even.

\end{enumerate}

For the partition type $(n - 2, 1, 1)$, let the two $1$-partition vertices be $u_1$ and $u_2$ with degrees $d_1$ and $d_2$, respectively. Without loss of generality, assume that $d_2 > d_1$. Since the unique degrees $d_1$ and $d_2$ are different, the $n - 2$-partition implicitly partitions itself into three parts: The partition that is adjacent to the vertex $u_1$ of size $\alpha_1$, the partition that is adjacent to the vertex $u_2$ of size $\alpha_2$, and the remaining vertices that pair themselves. These partitions are rigid in that no $\pi$ from the automorphism group can map vertices across the partition. Hence, assuming the correct parity for $n$, the probability that a random $\pi$ from $S_n$ is in the automorphism group is
\begin{equation*}
\begin{split}
\frac{\left| \Aut \left( U_n \right) \right|}{n!} & \leq \frac{\alpha_1! \alpha_2! \displaystyle \left( \frac{n - \alpha_1 - \alpha_2 - 2}{2} \right)! 2^{\left( n - \alpha_1 - \alpha_2 - 2 \right) / 2}}{n!} \\
& \leq \frac{d_1! d_2! \displaystyle \left( \frac{n - \alpha_1 - \alpha_2 - 2}{2} \right)! 2^{\left( n - \alpha_1 - \alpha_2 - 2 \right) / 2}}{n!} \\
& = O \left( \frac{1}{n^3} \right),
\end{split}
\end{equation*}
if $d_1$ and $d_2$ are both from the set $\{\, 0, 2 \,\}$. Therefore, $d_2$ is either $n - 1$ or $n - 2$. The value of $d_2$ cannot be $n - 1$, since then $u_2$ is connected to all the other vertices, forcing $d_1 = 1$, which is not allowed. The only possibility that remains is $d_2 = n - 2$. If $d_1 = 0$, then the vertex $u_1$ is isolated, and $u_2$ is connected to all vertices in the $n - 2$-partition. In this case, we have
\begin{equation*}
\left| \Aut \left( U_n \right) \right| = (n - 2)! = \frac{n!}{O \left( n^2 \right)} = \omega \left( \frac{n!}{n^3} \right),
\end{equation*}
so that we allow this graph. We also allow the complement of this graph, a $K_{n-1}$ with a vertex of degree $1$ adjacent to it.

If $d_2 = n - 2$ and $d_1 = 2$, then the vertex $u_2$ is connected to all but one vertex in the $n - 2$-partition, and it is also connected with the vertex $u_1$. The vertex $u_1$ is also connected with the isolated vertex in the $n - 2$-partition. In this case, we have
\begin{equation*}
\left| \Aut \left( U_n \right) \right| = (n - 3)! = \frac{n!}{O \left( n^3 \right)} = O \left( \frac{n!}{n^3} \right),
\end{equation*}
so that this graph is rejected.

For the $(n - 2, 2)$ case, we have two vertices $v_1$, and $v_2$ of degree $d$. We have the following cases.

\begin{enumerate}

\item If $d = 0$, we have a case similar to that of $(n)$ with degree $1$, where the automorphism group size is
\begin{equation*}
\left| \Aut \left( U_n \right) \right| = 2 \cdot 2^{(n - 2) / 2} \left( \frac{n - 2}{2} \right)! = O \left( \frac{n!}{n^3} \right).
\end{equation*}
We disallow this case when $n$ is even. When $n$ is odd, the graph is not possible.

\item If $d = 1$, this is not allowed since we have defined the partition class this way.

\item For $d = 2$, this implicitly partitions the $n - 2$-partition into two parts: Adjacent to a vertex of degree $2$ and not adjacent to a vertex of degree $2$. Suppose that these vertices are partitioned into partitions of size $\alpha_1$ and $\alpha_2$, respectively, the probability that $\pi \in S_n$ is in the automorphism group is
\begin{equation*}
\frac{\left| \Aut \left( U_n \right) \right|}{n!} \leq \frac{2 \cdot \alpha_1! \alpha_2!}{n!} = \frac{2}{n(n-1)} \cdot \frac{1}{\displaystyle {{n-2} \choose \alpha_1}} = O \left( \frac{1}{n^3} \right),
\end{equation*}
since $\alpha_1$ and $\alpha_2$ are necessarily positive. If they were not, we would either have $v_1$ and $v_2$ have very high degree, or degree $0$ or $1$. We reject this graph.

\item For $d = n - 2$ and $d = n - 1$, this is not allowed since at least one vertex from the $n - 2$-partition would have to have a degree larger than $1$.

\end{enumerate}

This case covers the statements 3 and 4 of the lemma.

\textbf{\textit{Case 3: The Largest Partition Degree is $2$.}} \\
For the $(n)$ case, for sufficiently large $n$, we must have $v_1, v_2, v_3, v_4, v_5$, and $v_6$ such that $v_1$ and $v_2$ are adjacent, $v_2$ and $v_3$ are adjacent, $v_4$ and $v_5$ are adjacent, and $v_5$ and $v_6$ are adjacent. If we pick a random permutation $\pi$ from $S_n$, the probability that it is in the automorphism group is
\begin{equation*}
\frac{\left| \Aut \left( U_n \right) \right|}{n!} \leq \frac{n \cdot 2 \cdot (n - 3) \cdot 2 \cdot (n - 6)!}{n!} = O \left( \frac{1}{n^4} \right),
\end{equation*}
since $v_2$ can map to at most $n$ vertices, $v_1$ and $v_3$ can only swap their positions as a neighbor of $v_2$; similarly, $v_5$ can map to at most $n - 3$ vertices, $v_4$ and $v_6$ can only swap their positions as a neighbor of $v_5$, and the remaining vertices can map freely to give an upper bound. Hence, we reject this case.

When we have the partition type $(n - 1, 1)$, we have the following cases, based on the degree $d$ of the $1$-partition vertex $u$.

\begin{enumerate}

\item If $d = 0$, this graph suffers the same pitfalls as the $(n)$ partition case and has the automorphism group size of
\begin{equation*}
\left| \Aut \left( U_n \right) \right| \leq \frac{(n - 1) \cdot 2 \cdot (n - 4) \cdot 2 \cdot (n - 7)!}{n!} = O \left( \frac{1}{n^5} \right),
\end{equation*}
hence, we reject this case.

\item If $d = 1$, suppose $v_1$ is in the $(n-1)$-partition and adjacent to $u$. If $v_2$ is adjacent to $v_1$, we must find a $v_3$ adjacent to $v_2$ since $v_3$ cannot be adjacent to any of the vertices we already numbered. Otherwise, $u$'s degree would be too high, and a similar case would go for $v_1$ and $v_2$. Once we continue this process and reach $v_{n-1}$, this vertex has no chance of having a degree $2$ since all other vertices have their promised degrees. Such a graph does not exist.

\item If $d = 2$, we violate the definition of our partition structure.

\item If $d = n - 2$, we only have one possibility: Suppose $u$ is adjacent to $v_1$ through $v_{n - 2}$. The vertex $v_{n-1}$ is adjacent to $v_1$ and $v_2$. From $i = 1$ onwards, $v_{2i+1}$ is also adjacent to $v_{2i+2}$. The automorphism group of this graph is of the size of
\begin{equation*}
\left| \Aut \left( U_n \right) \right| = 2 \cdot 2^{(n - 4) / 2} \left( \frac{n - 4}{2} \right)! = O \left( \frac{n!}{n^3} \right),
\end{equation*}
since the vertices $v_1$ and $v_2$ can swap themselves; and all other remaining $(n - 4) / 2$ pairs can swap and rearrange themselves. We reject this case when $n$ is even. When $n$ is odd, the graph is not possible. 

\item If $d = n - 1$, then the structure would be $u$, connected to each $v_i$ and the $v_i$'s forming pairs again, like the sticks. The automorphism group is of the size of
\begin{equation*}
\left| \Aut \left( U_n \right) \right| = 2^{(n - 1) / 2} \left( \frac{n - 1}{2} \right)! = O \left( \frac{n!}{n^3} \right),
\end{equation*}
for sufficiently large $n$, and we reject this case when $n$ is odd. The graph is not possible when $n$ is even.

\end{enumerate}

When we have a partition structure $(n-2, 2)$, we have the following cases, where $d$ is the degree of the $2$-partition.

\begin{enumerate}

\item If $d = 0$, then this suffers from the same asymptotic pitfalls as the $(n)$-case for degree $2$ and we reject this case:
\begin{equation*}
\left| \Aut \left( U_n \right) \right| \leq \frac{2 \cdot (n - 2) \cdot 2 \cdot (n - 5) \cdot 2 \cdot (n - 8)!}{n!} = O \left( \frac{1}{n^6} \right).
\end{equation*}

\item If $d = 1$, suppose $u_1$ and $u_2$ are from the $2$-partition. If $u_1$ and $u_2$ are adjacent, this suffers from the same pitfall as the $(n)$-case again (as shown in case 1 above), and we reject this case. If they are not adjacent, then suppose that $v_1$ is adjacent to $u_1$. The vertex $v_1$ is adjacent to $v_2$. The vertex $v_2$ cannot be adjacent to any of the vertices we visited, so we require a new vertex $v_3$. Similarly, we go on until $v_{n - 2}$. The vertex $v_{n - 2}$ must be adjacent to $u_2$, since all the others already have the promised degree. This resulting graph has an automorphism group size of $2$: Only reflectional symmetry. Another alternative is one chain from $u_1$ to $u_2$, and a cover of cycles. Once again, the $u_1$, $u_2$ component with the chain only has reflectional symmetry, so we have an automorphism group of size
\begin{equation*}
\left| \Aut \left( U_n \right) \right| \leq 2 \cdot (n-3)! = \frac{n!}{O \left( n^3 \right)} = O \left( \frac{n!}{n^3} \right),
\end{equation*}
and we reject this case.

\item The case of $d = 2$ is again not allowed.

\item If $d = n - 2$, we have two cases:

\begin{itemize}

\item If $u_1$ and $u_2$ are not adjacent, then they are connected to each vertex of the $n - 2$-partition. This graph has an automorphism group of size
\begin{equation*}
\left| \Aut \left( U_n \right) \right| = 2 \cdot (n - 2)! = \frac{n!}{O \left( n^2 \right)} = \omega \left( \frac{n!}{n^3} \right),
\end{equation*}
and we accept this and its complement: $K_{n-2}$ with the other component being an edge.

\item If $u_1$ and $u_2$ are adjacent, then $u_1$ and $u_2$ are adjacent to $n-3$ vertices each in the $n - 2$-partition. The vertices $v_3$ through $v_{n - 2}$ are adjacent to both, and $v_1$ (adjacent to $u_1$) is adjacent to $v_2$ (adjacent to $u_2$). The automorphism group size is
\begin{equation*}
\left| \Aut \left( U_n \right) \right| = 2 \cdot (n - 4)! = O \left( \frac{n!}{n^4} \right),
\end{equation*}
and hence, we reject it.

\end{itemize}

\item If $d = n - 1$, then this graph is a complement of $K_{n - 2}$ along with two isolated vertices. This graph has an automorphism group size of
\begin{equation*}
\left| \Aut \left( U_n \right) \right| = 2 \cdot (n - 2)! = \frac{n!}{O \left( n^2 \right)} = \omega \left( \frac{n!}{n^3} \right), 
\end{equation*}
and we accept this and its complement.

\end{enumerate}

When we have a partition structure of $(n - 2, 1, 1)$, we reject. We can categorize this graph as follows: The vertices in the $n - 2$-partition form a cycle within the partition, there is a chain starting at $u_1$ and ending at $u_2$, or starting at $u_i$ and ending at $u_i$ (for $i = 1$ or $2$). There must be at least one such cycle containing some $u_i$, since otherwise, $d_1$ would be equal to $d_2$. Let $\alpha$ be the length of the chain and $\beta$ be the number of such isomorphic chains. The number of permutations in the automorphism group is
\begin{equation*}
\left| \Aut \left( U_n \right) \right| \leq 2^\beta \beta! ( n - \alpha \beta)! = O \left( \frac{n!}{n^3} \right),
\end{equation*}
since $\beta$ is at least $1$ and $\alpha$ is at least $3$.

This case covers the statements 5 and 6 of the lemma.

\end{proof}

Due to the above case-by-case analysis, we have the following lemma.

\begin{lemma}
\label{lemma:24}

For sufficiently large $n$, in $\tilde{O} \left( n^2 \right)$-time, given an $n$-vertex undirected simple graph $U_n$, we can check whether $\text{Aut} \left( U_n \right) = \omega \left( n! / n^3 \right)$ and also compute the number of $k$-cliques for any $k > 2$.

\end{lemma}

\begin{proof}

Our algorithm will proceed as follows. If the number of edges in $U_n$ is larger than $\displaystyle {n \choose 2} / 2$, then we check if it is one of the large-clique structures. If not, then we compute $\overline{U_n}$ and check for one of the large clique structures. Both counting edges and computing the complement of $U_n$ requires $O \left( n^2 \right)$-time. Hence, assuming $U^\prime_n = U_n$ or $\overline{U_n}$ with at least $\displaystyle {n \choose 2} / 2$ edges, our algorithm proceeds as follows.

\begin{enumerate}

\item \textit{Checking if $U^\prime_n$ is $K_n$:} Simply check if every entry in $U^\prime_n$ is $1$ confirms this. If this test is passed, if $U^\prime_{n} = U_{n}$, then the number of $k$-cliques is $\displaystyle {n \choose k}$. If $U^\prime_{n} = \overline{U_n}$, then the number of $k$-cliques is $0$. If $U^\prime_{n}$ does not pass this test, then we move to the next test.

\item \textit{Checking if $U^\prime_n$ is $K_n$ with one missing edge:} Simply checking if exactly one entry in $U^\prime_n$ is $0$ confirms this. If the test is passed and $U^\prime_n = U_n$, then the number of $k$-cliques is $\displaystyle {n \choose k} - {{n - 2} \choose {k - 2}}$, since the subtracted number is the number of $k$-cliques that, in $K_n$, would contain the excluded edge. If $U^\prime_n = \overline{U_n}$, then the number of $k$-cliques is $0$. If this test fails, then we move to the next test.

\item \textit{Checking if $U^\prime_n$ is $K_{n - 1}$ with an isolated vertex:} It suffices to check if $n - 1$ vertices have degree $n - 2$ and one has degree $0$. If $U^\prime_n$ passes this test and $U^\prime_n = U_n$, then the number of $k$-cliques is $\displaystyle {{n - 1} \choose k}$. If $U^\prime_n = \overline{U_n}$, then the number of $k$-cliques is $0$. If this test fails, then we move to the next test.

\item \textit{Checking if $U^\prime_n$ is $K_{n - 1}$ with one vertex of degree $1$ adjacent to it:} First, we count the degrees of the vertices. If there is agreement with the expected number of vertices of each degree, then we are done, since all vertices with degree $n - 2$ must form a $K_{n - 2}$ subgraph, all adjacent to the vertex of degree $n - 1$, since the vertex of degree $n - 1$ is already adjacent to the vertex of degree $1$. If this test passes and $U^\prime_n = U_n$, then the number of $k$-cliques is $\displaystyle {{n - 1} \choose k}$. If $U^\prime_n = \overline{U_n}$, then the number of $k$-cliques is $0$. If this test fails, then we move to the next test.

\item \textit{Checking if $U^\prime_n$ is $K_{n - 2}$ with two isolated vertices:} It suffices in this case to check alignment with the expected degrees of the vertices. If the test passes and $U^\prime_n = U_n$, then the number of $k$-cliques is $\displaystyle {{n - 2} \choose k}$. If $U^\prime_n = \overline{U_n}$, then the number of $k$-cliques is $0$. If this test fails, then we move to the next test.

\item \textit{Checking if $U^\prime_n$ is $K_{n - 2}$ with two vertices of degree $1$ adjacent to each other:} First, we compute the degrees of the vertices and check that the vertices of degree $1$ are adjacent to each other. This forces the other $n - 2$ vertices to form an $n - 2$-clique. If this test passes and $U^\prime_n = U_n$, then the number of $k$-cliques is $\displaystyle {{n - 2} \choose k}$. If $U^\prime_n = \overline{U_n}$, then the number of $k$-cliques is $0$. If this test fails as well, and after all other tests, we know from our classification that
\begin{equation*}
\left| \Aut \left( U_n \right) \right| = O \left( \frac{n!}{n^3} \right).
\end{equation*}

\end{enumerate}

In all six cases, in $\tilde{O} \left( n^2 \right)$-time (depending on one's preferred model of computation), we can determine the appropriate classification if
\begin{equation*}
\left| \Aut \left( U_n \right) \right| = \omega \left( \frac{n!}{n^3} \right),
\end{equation*}
and also compute the number of $k$-cliques in $\tilde{O} \left( n^2 \right)$-time. If not, we can determine that 
\begin{equation*}
\left| \Aut \left( U_n \right) \right| = O \left( \frac{n!}{n^3} \right).
\end{equation*}

\end{proof}

\subsection{Counting Cliques Quickly Over Highly Symmetric Graphs}
\label{section:8.3}

Now, given that we have a test that quickly distinguishes graphs that, on a logarithmic scale, have almost the largest possible automorphism group from those that are slightly smaller, as well as a way to deal with graphs that have slightly smaller than $S_n$ itself, we now need to deal with graphs that have almost maximum-sized automorphism groups. We are in luck since, like many natural graph problems, the difficulty of detecting or counting $k$-cliques is inversely related to the size of the automorphism group of $U_n$. In particular, the algorithm we show below has time complexity $O\left( 1 / \left| \Aut \left( U_n \right) \right|^2 \right)$, keeping $n$ constant.

\begin{algorithm}
\caption{$\text{Sym-Graph-Clique-Count}\left( U_n, n, t, k \right)$}
\label{alg:9}
\Comment{The Symmetric Graph $k$-Clique Counting Algorithm \hspace{6.5cm}} \\
\Comment{Input: $n$ is the number of vertices in the undirected simple graph $U_n$} \\
\Comment{Input: $t$ is a positive integer} \\
\Comment{Input: $k$ is a positive integer} \\
\Comment{Input: $U_n$ is a list of edges for an undirected simple graph with $n$ vertices and $\left| \Aut \left( U_n \right) \right| \geq \displaystyle \frac{n!}{t}$} \\
\Comment{Output: The number of $k$-cliques in $U_n$}
\begin{algorithmic}
\State $I \gets 0$
\State $C \gets \{\, U_n \,\}$ 
\While{$I < t n^2$}
    \State $\pi \gets_r S_n$
    \If{$\pi \left( U_n \right) \not\in C$}
        \State $C \gets C \cup \{\, \pi \left( U_n \right) \,\}$
    \EndIf
    \State $I \gets I + 1$
\EndWhile
\State $m \gets 0$
\For{each $U_n^\prime \in C$}
    \If{$\left( \left( U_n^\prime.e_{\{\, i, j \,\}} \right)_{j = 1}^k \right)_{i < j} = (1)^{l \in {k \choose 2}}$}
        \State $m \gets m + 1$
    \EndIf
\EndFor
\State \Return $m \displaystyle {n \choose k} / |C|$
\end{algorithmic}
\end{algorithm}

Roughly, the algorithm proceeds by trying to list all graphs isomorphic to $U_n$ and checking how often the first $k$ vertices in each one is a clique. As can be seen in Algorithm \ref{alg:9}, when everything is computed correctly, if $m$ is the number of such graphs in $\mathcal{C}_n$, then the number of $k$-cliques is $m \displaystyle {n \choose k} / |\mathcal{C}_n|$. We attempt to list all possible graphs in $\mathcal{C}_n$ in time at least $|\mathcal{C}_n|n^2$. In particular, it is probabilistically unlikely that we will miss any of the graphs in $\mathcal{C}_n$, given this much time. Since, $\mathcal{C}_n$ is small, we can do this fairly quickly. We prove these assertions in the following lemma.

\begin{lemma}
\label{lemma:25}

Given an $n$-vertex undirected simple graph $U_n$ whose automorphism group, $\Aut \left( U_n \right)$ is promised to be of size at least $n! / t$, for any $k \in \mathbb{N}$, there is an algorithm counting the number of $k$-cliques correctly in time $O \left( t^2 n^4 \right)$ with high probability.

\end{lemma}

\begin{proof}

We intend to prove that Algorithm \ref{alg:9} satisfies these claims. Suppose that our graph $U_n$ has $\left| \Aut \left( U_n \right) \right| \geq n! / t$. Then, since $n! = \left| \mathcal{C}_n \right| \left| \Aut \left( U_n \right) \right|$, we would have that $\left| \mathcal{C}_n \right| \leq t$, where $\mathcal{C}_n$ is the set of all distinct graphs isomorphic to $U_n$. We prove in the following claim that we actually ``hit'' all graphs in $\mathcal{C}_n$ with high probability.

\begin{claim}
\label{claim:2}

With high probability, the set $C$ in Algorithm \ref{alg:9} is exactly $\mathcal{C}_n$.

\end{claim}

\begin{proof}

It is easy to see that $C$ is a subset of $\mathcal{C}_n$ since every adjacency list added to $C$ is a distinct permuted adjacency list of $U_n$. Our task is to show that the adjacency list of every graph in $\mathcal{C}_n$ is added to $C$ with high probability. Suppose that the graph $U_n^\prime$ is in $\mathcal{C}_n$. Since the set $\Pi \left( U_n^\prime \right) = \{\, \pi | \pi \in S_n, \pi \left( U_n \right) = U_n^\prime \,\}$ is a coset of the automorphism group $\Aut \left( U_n \right)$ of $U_n$, $\left| \Pi \left( U_n^\prime \right) \right| = \left| \Aut \left( U_n \right) \right| \geq n! / t$. The probability that $\pi \left( U_n \right) = U_n^\prime$ is $\left| \Aut \left( U_n \right) \right| / \left| S_n \right| \geq 1 / t$ for each $\pi$ sampled uniformly from $S_n$. Over $n^2t$ random samples of $\pi$, the probability that none of those samples are from $\Pi \left( U_n^\prime \right)$ is less than
\begin{equation*}
\left( 1 - \frac{1}{t} \right)^{t n^2}.
\end{equation*}
Since $(1 - 1 / t)^t < 1 / e$ for all $t > 1$, we have that the probability that $\Pi \left( U_n^\prime \right)$ is never hit is at most
\begin{equation*}
\left( 1 - \frac{1}{t} \right)^{t n^2} \leq e^{-n^2}.
\end{equation*}
Due to the union bound, the probability that there is a graph $U_n^\prime$ in $\mathcal{C}_n$ whose corresponding coset is never hit is upper bounded by
\begin{equation*}
\sum_{U_n^\prime \in \mathcal{C}_n} e^{-n^2} \leq t e^{-n^2} \leq n! e^{-n^2} = O \left( e^{-0.99n^2} \right).
\end{equation*}

\end{proof}

Now, we prove the relation between the correctly proven values of $m$, $|C|$, and the number of $k$-cliques.

\begin{claim}
\label{claim:3}

In Algorithm \ref{alg:9}, if $C = \mathcal{C}_n$, then
\begin{equation*}
\mathcal{P}_{\pi \gets_r S_n} \left[ \left( \left( e_{\{\, \pi(i), \pi(j) \,\}} \right)_{j \in [k] - [1]} \right)_{1 \leq i < j} = (1)_{l \in \binom{k}{2}} \right] = \frac{m}{|C|}.
\end{equation*} 

\end{claim}

\begin{proof}

When given a graph $U_n$, we want to compute how often the first $k$ vertices of $\pi \left( U_n \right)$ form a $k$-clique. Since $\left| \Aut \left( U_n \right) \right|$ many permutations map to the same graph in $C$, we can add $\left| \Aut \left( U_n \right) \right| / n! = 1 / |C|$ probability for each graph in $C$ that has a complete subgraph of $k$ vertices in the first $k$ vertices of the graph. In Algorithm \ref{alg:9}, the variable $m$ counts exactly the number of such graphs in $C$, and hence, the probability of interest to us is $m / |C|$.

\end{proof}

\begin{claim}
\label{claim:4}

The number of $k$-cliques in a graph $U_n$ is exactly
\begin{equation*}
\mathcal{P}_{\pi \gets_r S_n} \left[ \left( \left( e_{\{\, \pi(i), \pi(j) \,\}} \right)_{j \in [k] - [1]} \right)_{1 \leq i < j} = (1)_{l \in \binom{k}{2}} \right]  {n \choose k}.
\end{equation*}

\end{claim}

\begin{proof}

It is easy to see that the number of permutations $\pi \in S_n$ such that the vertices described by the set $\pi([k]) = (\pi(i))_{i \in [k]} \subset [n]$ induce a complete graph is 
\begin{equation*}
\mathcal{P}_{\pi \gets_r S_n} \left[ \left( \left( e_{\{\, \pi(i), \pi(j) \,\}} \right)_{j \in [k] - [1]} \right)_{1 \leq i < j} = (1)_{l \in \binom{k}{2}} \right] n!.
\end{equation*}
Since it does not matter what order the vertices are inside $\pi([k])$ and outside $\pi([k])$, respectively, $k!(n-k)!$ permutations $\pi$ describe the same set. Hence, the number of sets $S \subset [n]$ of size $k$ which induce a complete graph is exactly 
\begin{equation*}
\mathcal{P}_{\pi \gets_r S_n} \left[ \left( \left( e_{\{\, \pi(i), \pi(j) \,\}} \right)_{j \in [k] - [1]} \right)_{1 \leq i < j} = (1)_{l \in \binom{k}{2}} \right]  {n \choose k}.
\end{equation*}

\end{proof}

Hence, using Claims \ref{claim:3} and \ref{claim:4}, we obtain that the number of $k$-cliques in the graph $U_n$ is exactly $m \displaystyle {n \choose k} / |C|$, whenever $C = \mathcal{C}_n$, which is true with high probability due to claim \ref{claim:2}.

Now, to the runtime analysis, each time we sample $\pi$ and check $\pi \left( U_n \right)$ against all the other adjacency vectors in $C$, we have to make at most $t$ graph equality comparisons, which takes $O \left( tn^2 \right)$-time. Sampling $\pi$ from $S_n$ takes $n \log n$-time. Considering we repeat this process $tn^2$ times, the time complexity of this section of the algorithm is $O \left( t^2 n^4 \right)$. The proceeding section makes at most $t$ graph comparisons requiring $O \left( tn^2 \right)$-time. The time complexity of computing a binomial coefficient is at most $n^2 \polylog(n)$ \citep{Harvey2021}. Subsequently, we obtain that the time complexity of the algorithm is at $O \left( t^2 n^4 \right)$.
\end{proof}

If we are interested in the case where $t = o(n^3)$, then, due to the classification of highly symmetric graphs and analysis in Lemmas \ref{lemma:22}, \ref{lemma:23}, and \ref{lemma:24}, this can be done in $\tilde{O} \left( n^2 \right)$-time.

\subsection{Probabilistic Algorithms From Typically Generated Oracles}
\label{section:8.4}

We are now ready to prove our main theorem for this section.

\begin{theorem1}
\label{theorem:9}

For any $k \in \mathbb{N}$ (not necessarily a constant), given an $\epsilon = \omega \left( n^{3/2} / \sqrt{n!} \right)$, given an oracle $O$ sampled from $O^{H_n}_{1 / 2 + \epsilon}$, where $H_n: \{\, 0,1 \,\}^{n \choose 2} \to \mathbb{D}$ is any function defined over $n$-vertex undirected simple graphs that is invariant under graph isomorphism and can be computed in $O \left( n^{8 + o(1)} / \epsilon^{4 + o(1)} \right)$-time given the number of $k$-cliques in the graph, then with a probability of at least $1 - 2^{-\Omega \left( n^2 \right)}$ over the randomness of $O^{H_n}_{1 / 2 + \epsilon}$, we have an algorithm that, with access to $O$ computes $H_n$ with a high probability in time $O \left( \left( n^{8 + o(1)} / \epsilon^{2 + o(1)} + T_{O} \right) / \epsilon^2 \right)$, where $T_O$ is the time complexity of a hypothetical algorithm simulating the oracle $O$.

\end{theorem1}

\begin{proof}

Given a graph $U_n$ to compute $H_n$ on, we first use the automorphism group size-tester, $\text{Aut-Size-Test} \left( n, U_n, t(n) \right)$ (Algorithm \ref{alg:8}), to test whether $\left| \Aut \left( U_n \right) \right|$ is of size at least $n! / t(n)$, where $t(n) = 100n^2 / \epsilon^2$. This takes $O \left( n^3 t(n) \right) = O \left( n^5 / \epsilon^2 \right)$-time.

In the case when $\text{Aut-Size-Test} \left( n, U_n, t(n) \right)$ (Algorithm \ref{alg:8}) accepts, we run the $k$-clique counter, $\text{Sym-Graph-Clique-Count}\left( U_n, n, t, k \right)$ (Algorithm \ref{alg:9}), for fixed automorphism group sizes on $U_n$ with the parameter $t = t(n)^{1 + \alpha}$ for some constant $\alpha > 0$. This takes $O \left( t(n)^{2 + 2 \alpha} n^4 \right) = O \left( n^{8 + 4 \alpha} / \epsilon^{4 + 4 \alpha} \right)$-time. If $\text{Aut-Size-Test} \left( n, U_n, t(n) \right)$ (Algorithm \ref{alg:8}) rejects, then we use $H_n\text{-Amplifier}\left(O, n, U_n \right)$ (Algorithm \ref{alg:7}) to query $O$ on $O \left( 1 / \epsilon^2 \right)$ randomly chosen graphs that are isomorphic to $U_n$. This takes $O \left( T_O / \epsilon^2 \right)$-time.

Suppose that $U_n$ is a graph with an automorphism group $\Aut \left( U_n \right)$ of size at least $n! / t(n)$. In this case, $\text{Aut-Size-Test} \left( n, U_n, t(n) \right)$ (Algorithm \ref{alg:8}) accepts with a high probability. Then, with a high probability, $\text{Sym-Graph-Clique-Count}\left( U_n, n, t, k \right)$ (Algorithm \ref{alg:9}) computes the number of $k$-cliques correctly in $O \left( n^{8 + 4 \alpha} / \epsilon^{4 + 4 \alpha} \right)$-time. If $U_n$ has an automorphism group smaller than $n! / t(n)^{1 + \alpha}$, then with a high probability, $\text{Aut-Size-Test} \left( n, U_n, t(n) \right)$ (Algorithm \ref{alg:8}) rejects. Due to our analysis in Lemma \ref{lemma:19}, the isomorphism class, $\mathcal{C}_n$, has at least a $1 / 2 + \epsilon / 2$-fraction of correctness over $O$ (with a high probability over the randomness of $O^{H_n}_{1 / 2 + \epsilon}$) and hence, we can make $O \left( 1 / \epsilon^2 \right)$ queries to $O$ via $H_n\text{-Amplifier}\left(O, n, U_n \right)$ (Algorithm \ref{alg:7}) and take the majority. With a very high probability, we are correct.

When $U_n$ has $n! / t(n) \geq \left| \Aut \left( U_n \right) \right| \geq n! /t(n)^{1 + \alpha}$, no matter what $\text{Aut-Size-Test} \left( n, U_n, t(n) \right)$ (Algorithm \ref{alg:8}) returns, we have the correct answer with high probability. If the algorithm accepts, by choosing the parameter $t$ as $t(n)^{1 + \alpha}$, we have given $\text{Sym-Graph-Clique-Count}\left( U_n, n, t, k \right)$ (Algorithm \ref{alg:9}) sufficient time to list all graphs in the isomorphism classes of $U_n$ and compute the number of $k$-cliques. If $\text{Aut-Size-Test} \left( n, U_n, t(n) \right)$ (Algorithm \ref{alg:8}) rejects, then since $\left| \text{Aut} \left( U_n \right) \right| \leq n! / t(n) = n! \epsilon^2 /(100 n^2)$, its isomorphism class $\mathcal{C}_n$ is larger than $100 n^2 / \epsilon^2$ and with high probability over the randomness of $O^{H_n}_{1 / 2 + \epsilon}$ (Lemma \ref{lemma:19}), we have an oracle $O$ from which all isomorphism classes of that size can be error-corrected from $O$ into a probabilistic algorithm.

This algorithm takes $O \left( \left( n^{8 + 4 \alpha} / \epsilon^{2 + 4 \alpha} + T_{O} \right) / \epsilon^2 \right)$-time for every $\alpha > 0$ with a high probability over $O^{H_n}_{1 / 2 + \epsilon}$. By taking $\alpha = o(1)$, we get the $O \left( \left( n^{8 + o(1)} / \epsilon^{2 + o(1)} + T_{O} \right) / \epsilon^2 \right)$-time complexity.

\end{proof}

This already implies some progress on the open problem of \cite{Goldreich2020}, as shown below in the following corollary.

\begin{corollary}
\label{cor:4}
For any constant $\epsilon > 0$, an oracle $O$ sampled from $O^{H_n}_{1 / 2 + \epsilon}$, where $H_n$ is the function counting the $k$-clique parity in $n$-vertex undirected simple graphs, with a probability of over $1 - 2^{-\Omega \left( n^2 \right)}$, $O$ can be error-corrected to provide an $O \left( n^{8 + o(1)}+ T_O \right)$-time probabilistic algorithm for computing the parity of the number of $k$-cliques on any $n$-vertex undirected simple graph.
\end{corollary}

Using the classification of graphs $U_n$ with $|\textit{Aut}(U_n)| = \omega(n!/n^3)$ in Lemmas \ref{lemma:22}, \ref{lemma:23}, and \ref{lemma:24}, we obtain the following theorem, partially resolving the open problem of \cite{Goldreich2020} in an ``almost always'' sense. We discuss this further in Section \ref{section:8.5}.

\begin{theorem2}
\label{thm:10}

Given any constants $k > 2$ and $\epsilon > 0$, with a probability of at least $1 - 2^{-\Omega \left( n^2 \right)}$ over the randomness of sampling $O$ from $O^{H_n}_{1 / 2 + \epsilon}$, where $H_n$ is the function counting the number of $k$-cliques modulo $2$ in an $n$-vertex undirected simple graph, we have an $\tilde{O} \left( n^2 \right)$-time randomized reduction from counting $k$-cliques modulo $2$ on all instances to counting $k$-cliques modulo $2$ correctly over the $1 / 2 + \epsilon$-fraction of instances required of $O$. Moreover, this reduction has a success probability of greater than $2 / 3$.

\end{theorem2}

\begin{proof}

Using the classification of graphs with $\left| \Aut \left( U_n \right) \right| = \omega \left( n! / n^3 \right)$, and Lemmas \ref{lemma:22}, \ref{lemma:23}, and \ref{lemma:24}, in $\tilde{O}(n^2)$-time, for sufficiently large $n$, we can check whether a graph $U_n$ has an automorphism group $\Aut \left( U_n \right)$ of size $\omega \left( n! / n^3 \right)$; if so, count the number of $k$-cliques as well. We can return this value modulo $2$. If the graph $U_n$ has an automorphism group $\Aut \left( U_n \right)$ of size $O \left( n! / n^3 \right)$, then due to the analyses in Lemmas \ref{lemma:19} and \ref{lemma:20}, with a probability of at least $1 - 2^{-\Omega \left( n^2 \right)}$ over the randomness of $O^{H_n}_{1 / 2 + \epsilon}$, each isomorphism class $\mathcal{C}_n$ of size larger than $100 n^2 / \epsilon ^2 = o \left( n^3 \right)$ has more than a $1 / 2 + \epsilon / 2$-fraction of its queries over $O$ set to the correct answer. With $1 / \epsilon^2$ queries to $O$, in $\tilde{O} \left( n^2 / \epsilon^2 \right)$-time, we can take the majority of the answers, and this is correct with a probability of at least $1 - 2^{-\Omega(n)}$.

\end{proof}

We now show some corollaries related to the problem of deciding $\HALF$ and related counting problems.

\begin{corollary}
\label{cor:5}

Under \rETH, with a probability of more than $1 - 2^{-\Omega \left( n^2 \right)}$ over the randomness of $O^{\HALF}_{1 / 2 + 1 / 2^{o(n)}}$, an oracle $O$ sampled from $O^{\HALF}_{1 / 2 + 1 / 2^{o(n)}}$ has $T_O = 2^{\gamma n}$ for some $\gamma > 0$.

\end{corollary}

\begin{proof}

If this is false, then there would be a $2^{o(n)}$-time randomized algorithm deciding \HALF, and, due to the \ETH-hardness of \HALF (Appendix \ref{appendix:A}), there would also be a $2^{o(n)}$-time randomized algorithm deciding $3\SAT$.

\end{proof}

\begin{corollary}
\label{cor:6}

If $H_n$ is the function counting the parity of half-cliques on $n$-vertex undirected simple graphs, then with a probability of at least $1 - 2^{\Omega \left( n^2 \right)}$, given an oracle $O$ that is correct on a $1 / 2 + 1 / 2^{o(n)}$-fraction of instances sampled from $O^{H_n}_{1 / 2 + 1 / 2^{o(n)}}$, we can error-correct from $O$ and have a probabilistic algorithm running in $2^{o(n)} T_O$-time.

\end{corollary}

\begin{proof}

This follows directly from setting $\epsilon$ as some $1 / 2^{o(n)}$ function and $k$ as $\lfloor n / 2\rfloor$ in Theorem \ref{theorem:9}.

\end{proof}

\begin{corollary}
\label{cor:7}

Suppose that for all $\delta > 0$, no $2^{n (1 - \delta)}$-time randomized algorithm exists for detecting cliques of size $\lfloor n / 2 \rfloor$, then with a probability of greater than $1 - 2^{-\Omega \left( n^2 \right)}$, an oracle $O$ sampled from $O^{\HALF_2}_{1 / 2 + 1 / 2^{o(n)}}$ has $T_O$ larger than $O \left( 2^{\lceil n / 2 \rceil (1 - \Delta)} \right)$ for every $\Delta > 0$.

\end{corollary}

\begin{proof}

It was shown by \cite{Boix2019}, inspired by a similar argument of \cite{Ball2017}, that there is a probabilistic reduction from deciding whether there is a clique of size $k$ to deciding the parity of the number of cliques of size $k$ in simple undirected graphs. Moreover, the resulting algorithm takes $O \left( k 2^k T \right)$-time, where $T$ is the complexity of computing the parity of the number of $k$-cliques. By setting the variables and using the previous corollary, we achieve this result.

\end{proof}

\subsection{Discussion: On an Open Problem by \cite{Goldreich2020}}
\label{section:8.5}

In their work, \cite{Goldreich2020} proved that there is a worst-case to average-case reduction in the case of counting $k$-cliques\footnote{They refer to it as $t$-cliques in their paper.} modulo $2$. More specifically, they proved that if there is an oracle $O$ computing the number of $k$-cliques modulo $2$ in $n$-vertex undirected simple graphs on a $1 - 2^{-k^2}$-fraction of instances, then there is a probabilistic algorithm solving the same problem on all instances with a high probability, requiring $O \left( n^2 \right)$-time and $e^{O \left( k^2 \right)}$ queries to $O$. They leave open if there is such a reduction in the case where the fraction of instances the oracle $O$ is incorrect on is any constant less than $1 / 2$, where the time used is also $\tilde{O} \left( n^2 \right)$.

In this section, we showed that this is possible for almost all oracles with a $1 / 2 - \epsilon$-error fraction for every constant $\epsilon > 0$. In particular, the probabilistic method guarantees that almost all oracles with sufficient correctness overall have amplifiable correctness within each large isomorphism class. We stress that this result holds for almost all corrupt oracles and not all corrupt oracles, displaying a tradeoff in how our techniques make a significant improvement in the error tolerance at the cost of universality.

We leave open whether similar reductions exist for all oracles $O$ with a $1 / 2 + \epsilon$-fraction of correct instances or just those that are guaranteed by the probabilistic method. If such reductions exist, based on our analysis in Lemmas \ref{lemma:19} and \ref{lemma:20}, it seems fruitful to turn our attention toward the case where oracles have a $1 / 2 + \epsilon$-fraction of correctness. However, the oracle is entirely correct in some isomorphism classes and wholly incorrect in others. These ``bottleneck oracles'' are the ones that make finding a black-box reduction\footnote{One that makes no assumption regarding the distribution of correct instances other the fraction of them.} difficult.

\subsubsection*{Limitations of our Technique} 

We turn our attention to Corollary \ref{cor:5}. For $\HALF$, the $\NP$-complete problem of determining whether an $n$-vertex undirected simple graph has a clique of size $\lfloor n / 2 \rfloor$, we showed using non-adaptive queries to the oracle $O$, that almost all oracles $O$ with a $1 / 2 + \epsilon$-fraction of correctness can be amplified to give a probabilistic algorithm for $\HALF$ with reduction time $O(n^{8 + o(1)} / \epsilon^{4 + o(1)})$ and $O \left( 1 / \epsilon^2 \right)$ queries to $O$. If our techniques could be extended, still using non-adaptive queries, to show this result for all $O$ with the same error correction algorithm, then $\PH$ would collapse due to the works of \cite{Feigenbaum1993} and \cite{Bogdanov2006}. Even an adaptive modification to show such results would imply breakthrough-level results. In this view, if there is a reduction extending to the $2^{-\Omega(n^2)}$-fraction of oracles we were not able to show a hardness amplification for the problem of \cite{Goldreich2020}, then it seems that the error correction procedure would require time complexity or query complexity increasing with $k$. In particular, it should not be a polynomial-time procedure when $k$ grows as $\lfloor n / 2 \rfloor$. This, of course, seems evident from the works of \cite{Feigenbaum1993} and \cite{Bogdanov2006}, but since our work shows a polynomial-time correction procedure whose time and query complexity do not depend on $k$ (Theorem \ref{theorem:9}), we stress that an approach leading to a full resolution of the open problem of \cite{Goldreich2020} likely ``looks different'' from our approach.

\section{Open Problems} 
\label{section:9}

The following natural unresolved research directions remain.

\begin{enumerate}

\item Is it possible to show that either of these reductions still hold under $\ETH$? Of course, derandomizing polynomial identity testing \citep{Daniel2018} would go a long way, but can one circumvent the need to derandomize by providing tests based on other ideas? Although we have not gone into the depths of the working of the STV list decoder \citep{Sudan2001}, it does require randomness. Consider an oracle $O$ that correctly evaluates a polynomial $f:\mathbb{F}^n\to\mathbb{F}$ on an $\epsilon$-fraction of instances and is $0$ on all other inputs. For any deterministic strategy to probe the input oracle, we can set up $O$ such that the first $|\mathbb{F}|^n(1-\epsilon)$ entries probed are $0$, requiring exponential time for any $\epsilon < 1$. It seems black-box methods will not work without randomness, and we will need structural insights into our problems of interest. 
      
On the other hand, going all the way with derandomization and proving $\P = \BPP$ would imply that $\rETH$ is equivalent to $\ETH$ via a padding argument, making our implications true under $\ETH$.

\item Can we find tight reductions from $k\SAT$ to the problem of counting directed Hamiltonian cycles such that a $2^{n(1-\epsilon)}$-time would falsify the randomized exponential time hypothesis ($\rSETH$) \citep{Dell2014, Stephens2019}? Informally, such reductions would imply that we need roughly $2^{n}$-time to slightly improve upon simply printing $0$. The standard polynomial-time Karp reduction from $k\SAT$ seems to require $O(n+km)$ vertices in our graph. It is also essential to keep in mind that $1.657^n$-time randomized algorithms exist for Hamiltonian cycle detection \citep{Bjorklund2014}. Hence, it is crucial that the reduction does not merely require an algorithm that checks Hamiltonicity, as is the case with our reduction. Looking for $2^{o(n)}$-time reductions requiring many evaluations of the directed Hamiltonian cycle counting function on $n$ vertices is possibly the most productive direction.
      
For the problem of counting cliques of size $\lfloor n / 2 \rfloor$, the polynomial-time Karp reduction here suffers from the same pitfalls as that for the directed Hamiltonian cycle reduction, requiring $O(n + km)$ vertices. Here, as well, it might be productive to look for reductions requiring $2^{o(n)}$-time and many evaluations.

\end{enumerate}

\section*{Acknowledgments}
We would like to thank Oded Goldreich for his insightful feedback on an early draft of this paper.

\bibliographystyle{apalike}
\bibliography{main}

\appendix
\section{$\ETH$-Hardness of $\HALF$}
\label{appendix:A}

In this section, we will prove that there is a polynomial-time Karp reduction from $3\SAT$, which takes an instance with $n$ variables and $m$ clauses and gives us a $\HALF$ instance with $O(n + m)$ vertices. We do this by showing a Karp reduction from $\CLIQUE$ to $\HALF$ that takes a graph with $n$ vertices as input and gives us a graph with at most $2n$ vertices (depending on the parameter $k$ of the $\CLIQUE$ instance determining the size of the clique). This implies an $O(n+m)$-size Karp reduction from $3\SAT$ to $\HALF$ due to the famous $O(n+m)$-size reduction from $3\SAT$ to $\CLIQUE$ \citep{Arora2009}.

\begin{lemma}
\label{lemma:26}

There is a polynomial-time Karp reduction from $\CLIQUE$, taking an instance with an undirected graph $U$ on $n$ vertices and a parameter $k \leq n$, and outputting a $\HALF$ instance with an undirected graph $U^\prime$ of at most $2n$ vertices.

\end{lemma}

\begin{proof}

Given $U = \left( V_n, F_m \right)$ and $k$, we construct $U^\prime = \left( V^\prime_{n^\prime}, F^\prime_{m^\prime} \right)$ as follows. If $k > \lfloor n / 2 \rfloor$, then we construct $U^\prime$ with $n^\prime = 2k$ and $m^\prime = m$ such that $U^\prime$ consists of a copy of $U$ along with $2k - n$ isolated vertices:
\begin{equation*}
V_n = [n], \quad V_{n^\prime} = [2k], \quad F^\prime_{m^\prime} = F_m.
\end{equation*}
It is easy to see that $U^\prime$ has a clique of size $\lfloor n^\prime / 2\rfloor = k$ if and only if $U$ has a clique of size $k$. 

If $k \leq \lfloor n / 2 \rfloor$, then we construct $U^\prime = \left( V^\prime_{n^\prime}, F^\prime_{m^\prime} \right)$ such that we have $n^\prime = 2n - 2k$ and $m^\prime = m + \displaystyle \binom{n - 2k}{2} + n(n - 2k)$, and $U^\prime$ consists of a copy of $U$ and the complete graph $K_{n - 2k}$. Furthermore, $U$ and $K_{n - 2k}$ in $U^\prime$ are connected through edges in all possible combinations:
\begin{equation*}
\begin{split}
& V_n = [n], \quad V^\prime_{n^\prime} = [2n - 2k], \\
& F^\prime_{n^\prime} = F_n \cup \{\, \{\, i, j \,\} \mid n < i < 2n - 2k, i < j \leq 2n - 2k \,\} \cup \{\, \{\, i, j \,\} \mid i \in [n], j \in [2n - 2k] - [n] \,\}.
\end{split}
\end{equation*}
$U^\prime$ has a clique of size $\lfloor (2n-2k) / 2 \rfloor = n - k = (n - 2k) + k$ if and only if $U$ has a clique of size $k$.

This reduction takes $O \left( n^2 \right)$-time and the number of vertices in $V^\prime$ is at most $2n$ since $2k \leq 2n$ and $2n - 2k \leq 2n$.

\end{proof}

This also gives us the following corollary about the $\ETH$-hardness of $\HALF$.

\begin{corollary}
\label{corollary:8}

If $\ETH$ is true, then there is a constant $c > 0$ such that deciding $\HALF$ on a graph with $n$ vertices requires $2^{cn}$-time.

\end{corollary}

\begin{proof}

If $\ETH$ is true, then due to the Sparsification Lemma (\ref{lemma:6}) \citep{Impagliazzo2001}, no $2^{o(n)}$-time algorithm exists for $3\SAT$ on $n$ variables and $ln$ clauses for some sufficiently large constant $l > 0$. If $\HALF$ had a $2^{o(n)}$-time algorithm, then we would have a $2^{o(n)}$-time algorithm for $3\SAT$ on $n$ variables and $l^\prime n$ clauses for every fixed $l' > 0$.

\end{proof}

\section{An $O^{*}\left( n^{\omega k / 3} \right)$-Time Algorithm for Counting $k$-Cliques on Multigraphs}
\label{appendix:B}

This section will show an $O^{*}\left(n^{\omega k / 3} \right)$-time algorithm for counting $k$-cliques on multigraphs, provided that the number of edges between any two vertices is bounded by a prime $p = O \left( 2^{\poly(n)} \right)$. This has already been shown by \cite{Goldreich2018}. We show another one, only slightly different from the algorithm of \cite{Goldreich2018}, in that we use arithmetic over the reals for counting. We believe this algorithm provides some helpful intuition for their algorithm as well. This algorithm, unfortunately does not offer any savings over the trivial algorithm for counting half-cliques. However, it does offer savings when $k = o(n)$, and does better than $\displaystyle {n \choose k}$-time up to $k \approx 0.0499n$.

In some steps, we use finite-precision arithmetic over $\mathbb{R}$. We do this because not all elements of a field have square roots, but when necessary, we can always simulate arithmetic over fields of prime size over the reals.

\begin{algorithm}
\caption{$\text{$\#k$-\CLIQUE}\left( A, n, p \right)$}
\label{alg:10}
\Comment{The $k$-Clique-Counting Algorithm \hspace{9.9cm}} \\
\Comment{Input: $n$ is the number of vertices in the undirected multigraph represented by the adjacency matrix $A$} \\
\Comment{Input: $p = O \left( 2^{\poly(n)} \right)$ is a prime} \\
\Comment{$A$ is a symmetric $n \times n$ matrix with entries in $\mathbb{Z}_p$, $0$s in the diagonal, and $3$ divides $k$} \\
\Comment{Output: The number of unique cliques of size $k$ in the multigraph represented by $A$, modulo $p$.} \\
\Comment{Almost all computations are done over $\mathbb{R}$}
\begin{algorithmic}
\For{each $S \in \displaystyle \binom{[n]}{k / 3}$}
    \State $\mathcal{C}_{S} \gets \prod_{(i, j) \in S^2, i < j}A_{i, j}$ over $\mathbb{Z}_p$
    \State $\sqrt{\mathcal{C}_S} \gets$ the square root of $\mathcal{C}_S$ over $\mathbb{R}$ up to max$ \left \{\log^{100}(p), n^{100} \right \}$-bits of precision
\EndFor

\State $A^\prime \gets$ an $\displaystyle \binom{n}{k / 3} \times \binom{n}{k / 3}$ matrix initialized with all $0$s over $\mathbb{R}$ with rows and columns indexed by the sets $S$
\For{each $S$ and $S^\prime$ such that $S \cap S^\prime = \phi$}
    \State $A^\prime_{S, S^\prime} \gets \prod_{i \in S, j \in S^\prime\text{ or }i \in S^\prime, j \in S}A_{i, j}$ over $\mathbb{Z}_p$
    \State $A^\prime_{S, S^\prime} \gets A^\prime_{S, S^\prime}\sqrt{\mathcal{C}_{S}}\sqrt{\mathcal{C}_{S^\prime}}$ over $\mathbb{R}$
\EndFor

\State $M \gets A^{\prime 3}$ over $\mathbb{R}$
\State $m \gets \sum_{S \in \binom{[n]}{k / 3}} M_{S,S} / 3$ rounded to the nearest integer
\State $m^\prime \gets m \mod p$
\State \Return $m^\prime$
\end{algorithmic}
\end{algorithm}

We will now prove that $\text{$\#k$-\CLIQUE}\left( A, n, p \right)$ is correct and runs in the claimed time.

\begin{lemma}
\label{lemma:27}

The $k$-$\CLIQUE$ counting algorithm given in Algorithm \ref{alg:10} is correct and runs in $O^{*} \left( n^{\omega k / 3} \right)$-time where $\omega$ is the matrix multiplication constant, provided that $p$ grows as a function bounded by $2^{\poly(n)}$.

\end{lemma}

\begin{proof}

First, we justify the choice of computing over the real numbers up to some degree of precision and show that rounding to the nearest integer gives the same answer as when we use exact values and the answer is an integer.

\begin{claim}
\label{claim:5}

Let $A^\prime$ be the matrix computed in $\text{\#k-\CLIQUE}\left( A, n, p \right)$ and $B^\prime$ be the same matrix with infinite precision. Then $\sum_{S \in \binom{[n]}{k / 3}} \left( {A^\prime}^3 \right)_{S,S}$ rounded to the nearest integer is equal to $\sum_{S \in \binom{[n]}{k / 3}} \left( {B^\prime}^3 \right)_{S,S}$ provided that $\sum_{S \in \binom{[n]}{k / 3}} \left( {B^\prime}^3 \right)_{S,S}$ is an integer and
\begin{equation*}
\left| A^\prime_{S, S^\prime} - B^\prime_{S,S^\prime} \right| < \frac{1}{p^{n^3}}, \quad \forall S \in \binom{[n]}{k / 3}, \forall S^\prime \in \binom{[n]}{k / 3}.
\end{equation*}

\end{claim}

\begin{proof}

Note that the error of at most $\epsilon$ in the entries of the matrix is amplified to an error of order $\displaystyle {n \choose {k / 3}}^3 p^2 \epsilon$ in the trace provided that $\epsilon = o \left( 1 / \displaystyle {n \choose k / 3} p \right)$. For our choice of $\epsilon < 1 / p^{n^3}$, clearly
\begin{equation*}
\left| \sum_{S \in \binom{[n]}{k / 3}} {\left( A^\prime \right)}^3_{S,S} - \sum_{S \in \binom{[n]}{k / 3}}{\left( B^\prime \right)}^3_{S, S} \right| = o(1).
\end{equation*}

For all sufficiently large $n$, if $\sum_{S \in \binom{[n]}{k / 3}}{\left( B^\prime \right)}^3_{S, S}$ is an integer, then $\sum_{S \in \binom{[n]}{k / 3}}{\left( A^\prime \right)}^3_{S,S}$ rounded to the nearest integer is equal to it.

\end{proof}

We will show that for $B^{\prime3}$, the trace divided by $3$ is the correct answer, and subsequently, the trace of $A'^3$ rounded is as well. We will follow the algorithm. 

We represent
\begin{equation*}
w(S, S^\prime) = \prod_{i \in S, j \in S\prime \text{ or }i \in S^\prime, j \in S} A_{i, j},
\end{equation*}
if $S\cap S^\prime = \phi$, and $w(S) = \mathcal{C}_S$ for clarity. Using this notation, we have
\begin{equation*}
B^\prime_{S, S^\prime} = \sqrt{w(S)} \sqrt{w(S^\prime)} w(S, S^\prime),
\end{equation*}
if $S \cap S^\prime = \phi$, and $0$ otherwise. 

We have
\begin{equation*}
\left( B^{\prime2} \right)_{S, S^\prime} = \sqrt{w(S) w (S^\prime)} \sum_{S^{\prime\prime} \in \binom{[n]}{k / 3}} w(S^{\prime\prime}) w(S, S^{\prime\prime}) w(S^\prime, S^{\prime\prime}),
\end{equation*}
and subsequently
\begin{equation*}
\left( B^{\prime 3} \right)_{S, S} = w(S) \sum_{\left( S^\prime, S^{\prime\prime} \right) \in \binom{[n]}{k / 3}^2} w(S^\prime) w(S^{\prime\prime}) w(S, S^\prime) w(S^\prime, S^{\prime\prime}) w(S^{\prime\prime}, S).
\end{equation*}
And hence,
\begin{equation*}
\sum_{S \in \binom{[n]}{k / 3}} \left( B^{\prime 3} \right)_{S, S} = \sum_{\left( S, S^\prime, S^{\prime\prime} \right) \in \binom{[n]}{k / 3}^3} w(S) w(S^\prime) w(S^{\prime\prime}) w(S, S^\prime) w(S^\prime, S^{\prime\prime}) w(S^{\prime\prime}, S).
\end{equation*}

This sum is exactly three times the $k$-$\CLIQUE$ count in the multigraph defined by the adjacency matrix $A$, since each ``triangle'' is counted three times, once for each $S$ it is part of\footnote{Note that $w(S, S') = 0$ if $S \cap S' \neq \phi$}. And hence, the algorithm is correct.

Note that arithmetic over $\poly(n)$-bits and arithmetic over $\mathbb{Z}_p$ for $p = O \left( 2^{\poly(n)} \right)$ takes $\poly(n)$-time. Hence, each arithmetic operation amplifies the time by at most a polynomial factor and we can analyze the algorithm purely in terms of the number of arithmetic operations. We must take $\displaystyle {n \choose k/3}$ square roots, $O^{*} \left( \displaystyle {n \choose k / 3}^2 \right)$ time to construct $A^\prime$, $\displaystyle {n \choose k/3}^{\omega}$-time for matrix multiplications and $\displaystyle {n \choose k / 3}$-time for trace computation. Overall, this algorithm requires $O \left( \displaystyle {n \choose k/3}^{\omega} \right)$ arithmetic steps. Hence, Algorithm \ref{alg:10} requires $O^{*} \left( \displaystyle {n \choose k/3}^{\omega} \right) = O^{*}\left( n^{\omega k / 3} \right)$-time.

\end{proof}

When $k$ is not divisible by $3$, we can add $t = 3 \lceil k / 3 \rceil - k$ vertices to our graph $G$. If $t = 1$, we add one vertex $v_{n+1}$ and set the edge values $e_{\{\, i, n + 1 \,\}}$ for each $i \in [n]$ to $0$ in $G_1$, and to $1$ in $G_2$, respectively. By evaluating the $\#3 \lceil k / 3 \rceil$-$\CLIQUE$ on $G_1$ and $G_2$ and subtracting the $G_1$ result from $G_2$, we get the number of $k + 1$-sized cliques involving $v_{n+1}$, with a one-to-one correspondence with the $k$-$\CLIQUE$ count in $G$.

When $t = 2$, we use two vertices $v_{n+1}$ and $v_{n+2}$, and construct four graphs.

\begin{enumerate}

\item $G_1$ with $e_{i, n+1} = 0$, $e_{i, n+2} = 0$, and $e_{n+1, n+2} = 0$, for all $i \in [n]$.

\item $G_2$ with $e_{i, n+1} = 1$, $e_{i, n+2} = 0$, and $e_{n+1, n+2} = 0$, for all $i \in [n]$. 

\item $G_3$ with $e_{i, n+1} = 0$, $e_{i, n+2} = 1$, and $e_{n+1, n+2} = 0$, for all $i \in [n]$. 

\item $G_4$ with $e_{i, n+1} = 1$, $e_{i, n+2} = 1$, and $e_{n+1, n+2} = 1$, for all $i \in [n]$. 

\end{enumerate}

$\#3 \lceil k / 3 \rceil$-$\CLIQUE$ counting algorithm along with the \textit{Inclusion-Exclusion Principle} computes the number of cliques of size $k$ in $G$.\footnote{More, specifically, $f(G_4) - f(G_3) - f(G_2) + f(G_1)$, where $f$ is the output of the $\#3 \lceil k / 3 \rceil$-$\CLIQUE$ algorithm.}

\end{document}